\documentclass[a4paper]{article}
\usepackage[left=4cm, right=4cm, bottom=3.5cm, top=3.5cm]{geometry}

\usepackage{amsmath}
\usepackage{amsthm}
\usepackage{amsfonts}
\usepackage{graphicx}
\usepackage{xcolor}
\definecolor{Green}{RGB}{0,130,20}
\definecolor{linkcol}{RGB}{50,170,50}
\definecolor{citecol}{RGB}{230,120,0}
\definecolor{urlcol}{rgb}{0,0,0.9}
\usepackage{hyperref} 
\hypersetup{colorlinks, breaklinks, 
linkcolor=linkcol, 
citecolor=citecol, 
urlcolor=urlcol 
} 
\newcommand{\Zbb}{\mathbb{Z}}

\newtheorem{assumption}{Assumption}
\newcommand{\dif}{\mathrm{d}}
\theoremstyle{plain}
\newtheorem{theorem}{Theorem}
\newtheorem{lemma}{Lemma}
\theoremstyle{definition}
\newtheorem{definition}{Definition}
\newtheorem{remark}{Remark}
\providecommand{\keywords}[1]{\textbf{Keywords: } #1}
  
  \title{Phase space structures causing the reaction rate decrease in the collinear hydrogen exchange reaction.}

  \author{Vladim{\'i}r Kraj{\v{n}}{\'a}k\footnote{School of Mathematics, University of Bristol, Fry Building, Woodland Road, Bristol, UK}~\footnotemark[2]~~and~Holger Waalkens\footnote{Bernoulli Institute for Mathematics, Computer Science and Artificial Intelligence, University of Groningen, Nijenborgh 9, 9747 AG Groningen, The Netherlands}}
  
  \date{}
  
\begin{document}
  \maketitle

  \begin{abstract}
  The collinear hydrogen exchange reaction is a paradigm system for understanding chemical reactions. It is the simplest imaginable atomic system with $2$ degrees of freedom modeling a chemical reaction, yet it exhibits behaviour that is still not well understood - the reaction rate decreases as a function of energy beyond a critical value. Using lobe dynamics we show how invariant manifolds of unstable periodic orbits guide trajectories in phase space. From the structure of the invariant manifolds we deduce that insufficient transfer of energy between the degrees of freedom causes a reaction rate decrease. In physical terms this corresponds to the free hydrogen atom repelling the whole molecule instead of only one atom from the molecule. We further derive upper and lower bounds of the reaction rate, which are desirable for practical reasons.

  \keywords{hydrogen exchange, invariant manifolds, phase space structures, reaction dynamics, transition state theory}
  \end{abstract}
  
  \section{Introduction}\label{sec:introh2h}
  We study the dynamics of the \emph{collinear hydrogen exchange reaction} $\text{H}_2+\text{H}\rightarrow \text{H}+\text{H}_2$, which is an invariant subsystem of the spatial hydrogen exchange reaction, using the potential provided by Porter and Karplus in \cite{PorterKarplus64}. In literature it is considered a paradigm system for understanding chemical reactions due to its simplicity and variety of exhibited dynamics. Because the system consists of three identical atoms confined to a line, it is the simplest imaginable system with $2$ degrees of freedom modeling a chemical reaction.
  
  The hydrogen atoms themselves are the simplest atoms in the universe. Because each consist of one proton and one electron only, an accurate potential energy surface for this reaction can be obtained via the Born-Oppenheimer approximation. Intriguingly enough, this system exhibits behaviour that is still not well understood.

  The phenomenon we examine here is the counterintuitive observation that the reaction rate decreases as energy increases beyond a critical value. After all, one would expect to break bonds more easily using more energy. So far a satisfactory explanation of this phenomenon is missing and only an upper bound and a lower bound to the rate have been found. The upper bound is obtained by means of \emph{transition state theory} (TST), due to \cite{Wigner37}. TST is a standard tool for studying reaction rates due to its simplicity and accuracy for low energies, but it does not capture the decline of the reaction rate. The improvement brought by \emph{variational transition state theory} (VTST) \cite{Horiuti38}, does not capture this behaviour either.
  
  \emph{Unified statistical theory}, due to \cite{Miller76}, which is in a certain sense an extension of TST to more complicated system, does capture the culmination of the reaction rate, but does not yield higher accuracy. The lower bound on the other hand does come quite close. It is obtained using the so-called \emph{simple-minded unified statistical theory} \cite{PollakPechukas79UST}.
  
  A review of reaction rate results including TST can be found in \cite{Keck67}. \cite{Pechukas81} and \cite{Truhlar84} review various extensions of TST.
  
  Using lobe dynamics (introduced in \cite{Rom-Kedar90}) we show how invariant manifolds of unstable periodic orbits guide trajectories in phase space. From the structure of the invariant manifolds we deduce that insufficient transfer of energy between the degrees of freedom causes a reaction rate decrease. In physical terms this corresponds to the free hydrogen atom repelling the whole molecule instead of only one atom from the molecule. We further derive bounds of the reaction rate, which are desirable for practical reasons.
  
  In the remainder of this Section we introduce the system, give an overview of TST and explain the current state of affairs with regards to the collinear hydrogen exchange reaction. Section \ref{sec:orbits} focuses on relevant periodic orbits and definition of regions of phase space. In Section \ref{sec:define sos} we introduce new coordinates using which we define a surface of section. In Section \ref{sec:transport barriers} we explain how we study invariant manifolds on the surface of section. In Section \ref{sec:tangles reaction} we give a detailed insight into the structures formed by invariant manifolds and their role in the reaction. Section \ref{sec:intricate interval} is devoted to a novel way of breaking down heteroclinic tangles to provide a better understanding of the interplay of invariant manifolds of three TSs. In Section \ref{sec:bounds} we calculate various upper and lower bounds of the reaction rate.
    
   \begin{figure}  
      \centering
      \includegraphics{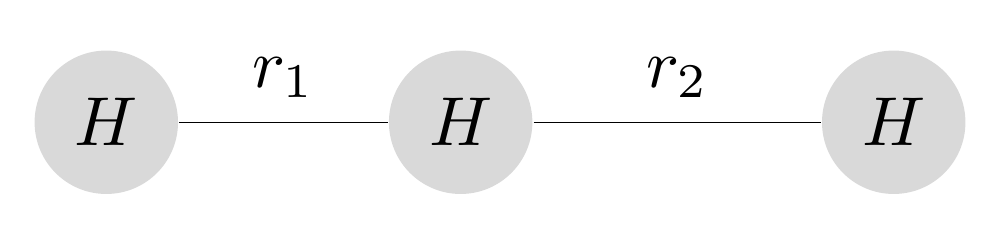}
      \caption{Collinear hydrogen atoms and distances.}
    \label{fig:H3}
  \end{figure}

  \subsection{Porter-Karplus potential}\label{subsec:potential}
  The collinear hydrogen exchange system consists of three hydrogen atoms confined to a line, as shown in Fig. \ref{fig:H3}, where $r_1$ and $r_2$ denote the distances in atomic units between neighbouring atoms. Forces between the atoms are given by the Porter and Karplus potential \cite{PorterKarplus64} is the standard potential for the hydrogen exchange reaction (collinear and spatial) used for example in \cite{MorokumaKarplus71,ChapmanHornsteinMiller75,PollakPechukas78,PollakPechukas79UST,PollakChildPechukas80,Davis87,Inarrea11}. The system is considered to react, if it passes from the region of reactants ($r_1>r_2$) to the region of products ($r_1<r_2$) and remains there.
  
  We point out two key properties of the Porter-Karplus potential:
  \begin{itemize}
   \item the discrete reflection symmetry with respect to the line $r_1=r_2$,
   \item saddle point at $r_1=r_2=R_s:=1.70083$.
  \end{itemize}
  
  The symmetry expresses the fact that we cannot distinguish between three identical hydrogen atoms, we can only measure distances between them. Hence, any statement referring to $r_1<r_2$ automatically also holds for $r_1>r_2$.
  
  Potential saddle points represent the activation energy needed for a reaction to be possible. In the all of this work we give energies as values in atomic units above the minimum of the system. In this convention the energy of the saddle point is $0.01456$. 
  
  From a configuration space perspective, such a potential barrier is the sole structure separating reactants from products and the sole obstacle the system needs to overcome in order to react. This perspective implicitly assumes that the system does not recross the potential barrier back into reactants. Dynamical structures that cause recrossings are only visible from a phase space perspective.
    
  \begin{figure}
	\centering
	\includegraphics[width=0.49\textwidth]{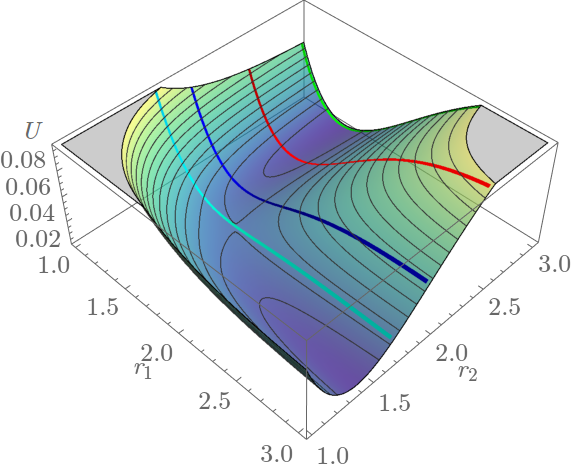}
	\includegraphics[width=0.49\textwidth]{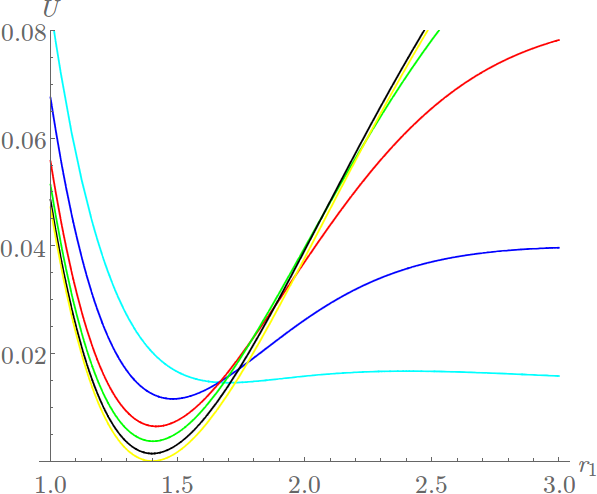}
	\caption{The Porter-Karplus potential energy surface with contours and its cross sections for fixed values of $r_2=1.70083$ (cyan), $2$ (blue), $2.5$ (red), $3$ (green), $4$ (black), $50$ (yellow).}
	\label{fig:potential}
  \end{figure}
  
  Figure \ref{fig:potential} shows the potential energy surface near the potential saddle and cross sections of the potential at various values of $r_2$. Due to diminishing forces between the atom and the molecule over large distances the differences between the cross sections fade after $r_2=4$ and are indistinguishable in double precision beyond $r_2=40$.
    
  \subsection{Definitions}\label{subsec:definitions}
  The collinear hydrogen exchange reaction is described by the Hamiltonian
  \begin{equation}\label{eq:Ham}
  H(r_1,p_{r_1},r_2, p_{r_2})=\frac{p_{r_1}^2+p_{r_2}^2-p_{r_1}p_{r_2}}{m_H} + U(r_1,r_2),
  \end{equation}
  where $p_{r_1}$, $p_{r_2}$ are the momenta conjugate to interatomic distances $r_1$, $r_2$, $m_H$ is the mass of a hydrogen atom and $U$ is the Porter-Karplus potential described above.
    
  The equations of motion associated to $H$ are as follows:
  \begin{equation}
  \begin{split}
  \dot{r}_1 &= \frac{2p_{r_1}-p_{r_2}}{m_H},\\
  \dot{p}_{r_1} &= -\frac{\partial U(r_1,r_2)}{\partial r_1},\\
  \dot{r}_2 &= \frac{2p_{r_2}-p_{r_1}}{m_H},\\
  \dot{p}_{r_2} &= -\frac{\partial U(r_1,r_2)}{\partial r_2}.
  \end{split}
  \label{eq:eqHam}
  \end{equation}
  The discrete symmetry of the potential translates into the invariance of $H$ and the equations of motion under the map $(r_1,p_{r_1},r_2, p_{r_2}) \mapsto (r_2,p_{r_2},r_1, p_{r_1})$.
  
  The Hamiltonian flow generated by equations \eqref{eq:eqHam} preserves the energy of the system $E=H(r_1,p_{r_1},r_2, p_{r_2})$ and the phase space of this system is therefore foliated by energy surfaces $H=E$.
  \begin{figure}
      \centering
      \includegraphics[width=0.7\textwidth]{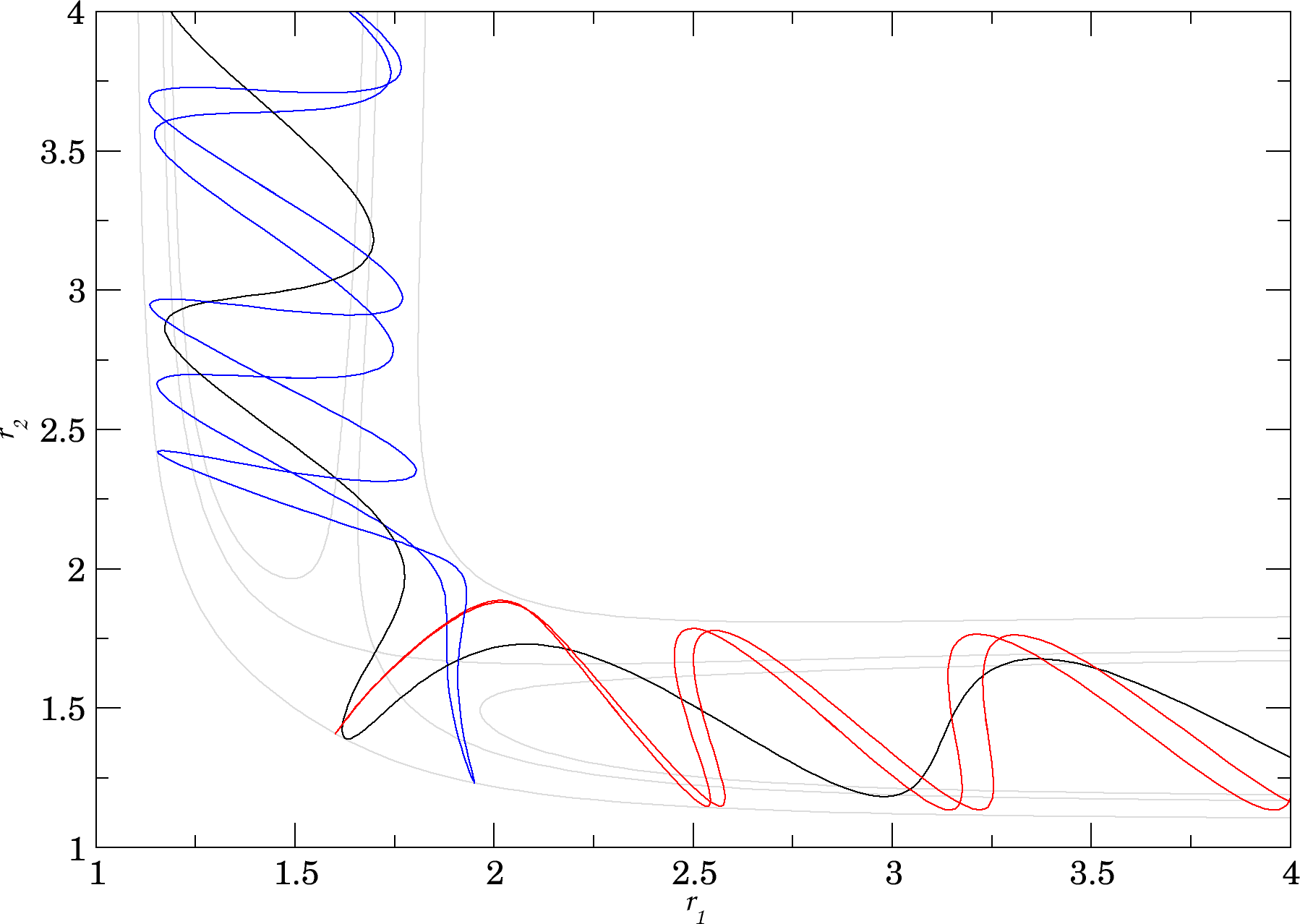}
	\caption{Examples of reactive (black) and nonreactive (red, blue) trajectories in configuration space at energy $0.02400$.}
	\label{fig:traj}
  \end{figure}
        
  \begin{definition}\label{def:trajectories}
   A trajectory passing through the point $\bigl(r_1^0,p_{r_1}^0,r_2^0, p_{r_{2}}^0\bigr)$ is said to be a \emph{reactive trajectory} if the solution
   $(r_1(t),p_{r_1}(t),r_2(t), p_{r_2}(t))$
   of the system with the initial condition
   $$(r_1(0),p_{r_1}(0),r_2(0), p_{r_2}(0))=\bigl(r_1^0,p_{r_1}^0,r_2^0, p_{r_{2}}^0\bigr),$$
   satisfies
   $r_1(t)<\infty$ and $r_2(t)\rightarrow\infty$ as $t\rightarrow\infty$ and $r_1(t)\rightarrow\infty$ and $r_2(t)<\infty$ as $t\rightarrow-\infty$ or vice versa.
   
   A \emph{nonreactive trajectory} is one for which the solution satisfies $r_1(t)\rightarrow\infty$ and $r_2(t)<\infty$ as $t\rightarrow\pm\infty$ or $r_1(t)<\infty$ and $r_2(t)\rightarrow\infty$ as $t\rightarrow\pm\infty$.
  \end{definition}
  
  Examples of reactive and nonreactive trajectories are shown in Figure \ref{fig:traj}. Note that nonreactive trajectories may cross the potential barrier in the sense that they cross the line $r_1=r_2$.

  From the above it follows that the reaction rate at a fixed energy $E$ can be calculated using a brute force Monte Carlo method as the proportion of initial conditions of reactive trajectories at infinity. Since the system decouples in a numerical sense around $r_2=40$, it is enough to sample a sufficiently remote surface in the reactants ($r_1>r_2$) that is transversal to the flow, for example
  \begin{equation}
   r_1+\frac{r_2}{2}=50,\quad p_{r_2}<0.
   \label{eq:MCsurface}
  \end{equation}    
  Since $r_1$, $r_2$ is not a centre of mass frame, $r_2=const$ is not transversal to the flow. 
  We remark that $(r_2, p_{r_2}-\frac{p_{r_1}}{2})$ are canonical coordinates on $r_1+\frac{r_2}{2}=50$ that yield a uniform random distribution of initial conditions.
  
  \subsection{Transition state theory}\label{subsec:TST}
  Since its formulation in \cite{Wigner37}, TST became the standard tool for estimating rates of various processes not only in chemical reactions \cite{Keck67}. It has found use in many fields of physics and chemistry, such as celestial mechanics \cite{Henrard82}, \cite{Jaffe02}, plasma confinement \cite{Meissetal85} and fluid mechanics \cite{Ottino89}. 
    
  Key element of TST is the \emph{transition state} (TS), a structure that is between reactants and products. There is no single generally accepted definition unfortunately, because in some publications concerning systems with $2$ degrees of freedom TS refers to an unstable periodic orbit while in others TS is a dividing surface (DS) associated with the unstable periodic orbit.
  We adopt the following definition of a TS from \cite{MacKay2014}:
  \begin{definition}[TS]
    A transition state for a Hamiltonian system is a closed, invariant, oriented, codimension-$2$ submanifold of the energy surface that can be spanned by two surfaces (the TS is the surfaces' boundary) of unidirectional flux, whose union divides the energy surface into two components and has no local recrossings.
  \end{definition}
  For a system with $2$ degrees of freedom as considered in this work, a closed, invariant, oriented, codimension-$2$ submanifold of the energy surface is a periodic orbit and it can be shown that the periodic orbit must be unstable \cite{Pechukas76}, \cite{PollakPechukas78}, \cite{Sverdlik78}.
  In general, the TS has to be a normally hyperbolic invariant manifolds (NHIM), an invariant manifolds with linearised transversal instabilities that dominate the linearised tangential instabilities (\cite{Fenichel71}, \cite{Hirsch77}).  
    
  \begin{theorem}[TST]
    \label{th:TST}
    In a system that admits a TS and all trajectories that pass from reactants to products the DS precisely once, the flux across a DS is precisely the reaction rate.
  \end{theorem}
    
  We remark that in general the flux through a DS associated with a TS is an upper bound to the reaction rate \cite{Wigner37}, \cite{Pechukas81}.
    
  Since its early applications, developments in the field led to a shift in the understanding of the TS to be an object in phase space rather than configuration space \cite{Wigginsetal01}, \cite{Uzeretal02}, \cite{Waalkensetal04a}, \cite{Waalkens04}, \cite{Waalkensetal04b}, \cite{Waalkensetal05a}, \cite{Waalkensetal05c}, \cite{Waalkens08}.  
  
  All relevant periodic orbits in this system are self-retracing orbits whose configuration space projections oscillate between equipotential lines, so called brake orbits (\cite{Ruiz75}). As suggested by \cite{PollakPechukas78}, let $(r_1^{po},r_2^{po})$ be the configuration space projection of a brake orbit at energy $E$, then the associated DS is the set of all phase space points $(r_1^{po},p_{r_1},r_2^{po}, p_{r_2})$ that satisfy $H(r_1^{po},p_{r_1},r_2^{po}, p_{r_2})=E$. For constructions of a DS near a saddle type equilibrium point in systems with more than $2$ degrees of freedom see \cite{Wigginsetal01}, \cite{Uzeretal02}, \cite{Waalkens04}.
  
  Hydrogen exchange results and evolution of understanding of TST follow.
   
 \subsection{Known results}\label{subsec:known results}
  In $1971$, Morokuma and Karplus \cite{MorokumaKarplus71} evaluated three representatives of different classes of reactions. They found the collinear hydrogen exchange reaction to be the best suited for a study of the accuracy of TST due to smoothness, symmetry and simplicity. They found that TST agreed with Monte Carlo calculations up to a certain energy, but became inaccurate rather quickly after that.
  
  In $1973$ \cite{PechukasMcLafferty73} Pechukas and McLafferty stated that for TST to be exact, every trajectory passing through the DS does so only once. In other words, TST fails in the presence of trajectories that oscillate between reactants and products.
  
  In $1975$ Chapman, Hornstein and Miller \cite{ChapmanHornsteinMiller75} present numerical results showing that transition state theory ``fails substantially'' for the hydrogen exchange reaction (collinear and spatial) above a certain threshold.
  
  Pollak and Pechukas \cite{PollakPechukas78} proved in $1978$ that flux through a DS constructed using an unstable brake orbit gives the best approximation of the reaction rate. In the presence of multiple TSs the authors introduce \emph{Variational TST} (VTST) - using the DS with the lowest flux to approximate the reaction rate. These results detach TST from potential saddle points. The authors find for the collinear hydrogen exchange reaction that when TST breaks down, VTST can be significantly more accurate, even though both fail to capture the reaction rate decrease.
  
  In $1979$ Pollak and Pechukas \cite{PechukasPollak79TST} proved that TST is exact provided there is only one periodic orbit. Simultaneously, they derived the best estimate of the reaction rate so far for the collinear hydrogen exchange reaction in \cite{PollakPechukas79UST} using what they called \emph{Simple-minded unified statistical theory} (SMUST).
  
  \emph{Unified statistical theory} (UST), due to Miller \cite{Miller76}, attempts to take advantage of the difference of fluxes through all DSs and essentially treat regions of simple and complicated dynamics separately. The authors of \cite{PollakPechukas79UST} found that UST captures the drop in the reaction rate and elaborate on the deviation of UST from the actual rate. The derivation of a lower bound (subject to assumptions) of the rate using the difference between TST and VTST is presented in the appendix of \cite{PollakPechukas79UST}.
  
  A rigorous lower bound is presented in \cite{PollakChildPechukas80}. It uses a DS constructed using a stable periodic orbit between two TS to estimate the error of TST. The accuracy of this lower bound for the hydrogen exchange reaction is remarkable.
   
  In $1987$ M. Davis \cite{Davis87} studied the hydrogen exchange reaction in phase space and considered the role of invariant structures. For low energies he showed that TST can be exact even if several TSs are present, provided that their invariant manifolds do not intersect. At higher energies he made some numerical observations of heteroclinic tangles of invariant manifolds and nearby dynamics. At high energies Davis found that a particular heteroclinic tangle grows in size and by assuming that it contains exclusively nonreactive trajectories he found a very accurate lower bound. The idea of this lower bound is very similar to \cite{PollakChildPechukas80}, but Davis endures a computational cost to quantify trajectories instead of fluxes through DSs.
  
  Davis also formulated an estimate of the reaction rate based on the observation that not many trajectories undergo a complicated evolution, as found by \cite{PollakPechukas79UST}. The estimate assumes that beyond a certain time dynamics in the heteroclinic tangle is randomised and $50\%$ of the remaining trajectories are reactive.

  Davis' observations hint at the crucial role played by invariant manifolds, but the precise manner in which this happens is not understood. Our aim is to explain the role of invariant manifolds in the reaction mechanism and extending it to the energy interval that Davis did not study, the interval with three TSs. We provide new understanding of the interactions between invariant manifolds of two and three TSs and consequently explain the counterintuitive reaction rate decrease.
  

 \section{Periodic orbits and geometry}\label{sec:orbits} 
 \subsection{Local geometry}\label{subsec:local geometry}
 Before we introduce periodic orbits that are relevant to the reaction mechanism, we describe the local energy surface geometry near a potential saddle point. We show that the neighbourhood necessarily contains an unstable periodic orbit and we highlight the importance of invariant manifolds to the local dynamics. The description remains true near unstable periodic orbits that do not lie near saddle points.
 
 Consider the Williamson normal form \cite{Williamson36}, \cite{Uzeretal02} of a system near a saddle point.
 In the neighbourhood $V$ of a potential saddle point, the system is accurately described in some suitable canonical coordinates $(q_1,p_1,q_2,p_2)$ by
 \begin{equation*}
 H_2(q_1,p_1,q_2,p_2)=\frac{1}{2}\lambda (p_1^2-q_1^2)+\frac{1}{2}\omega (p_2^2+q_2^2), 
 \end{equation*}
 where $\lambda,\omega>0$. For a fixed energy $H_2=h_2$, this is equivalent to
 \begin{equation}
  h_2+\frac{1}{2}\lambda q_1^2=\frac{1}{2}\lambda p_1^2+\frac{1}{2}\omega (p_2^2+q_2^2).
  \label{eq:sphere}
 \end{equation}
 \begin{figure}
   \centering
      \includegraphics{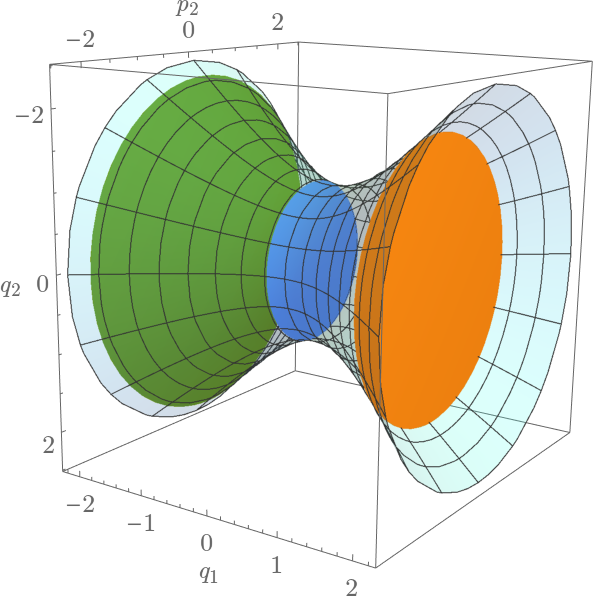}
	\caption{Illustration of local energy surface geometry in the neighbourhood of a saddle point. Sections for fixed values of $q_1$ define spheres (with $\pm p_1$ given implicitly by $H_2(q_1,p_1,q_2,p_2)=h_2$), shown are $q_1=1.5,-.25,-2$.}
	\label{fig:bottleneck}
  \end{figure}
 For a each fixed $q_1$ such that $h_2+\frac{1}{2}\lambda q_1^2>0$ this defines a sphere, as shown in Figure \ref{fig:bottleneck}. Depending on $h_2$, the energy surface has the following characteristics:
 \begin{itemize}
  \item If $h_2<0$, the energy surface consists of two regions locally disconnected near $q_1=0$, reactants ($q_1>0$) and products ($q_1<0$).
  \item Reactants and products are connected by the saddle point for $h_2=0$.
  \item For $h_2>0$, the energy surface is foliated by spheres. The radius of the spheres increases with $|q_1|$. Locally the energy surface has a wide-narrow-wide geometry usually referred to as a \emph{bottleneck}.
 \end{itemize}
 We remark that $q_1$ can be referred to as a \emph{reaction coordinate}.
 To understand transport through a bottleneck, fix an energy $h_2$ slightly above $0$ and consider the Hamiltonian equations for $H_2$:
  \begin{align*}
  \dot{q}_1 &= \lambda p_1,\qquad\qquad\dot{q}_2 = \omega p_2,\\
  \dot{p}_1 &= \lambda q_1,\qquad\qquad\dot{p}_2 = -\omega q_2.
  \end{align*}
 
 The degrees of freedom are decoupled with hyperbolic dynamics in $(q_1,p_1)$ and elliptic in $(q_2,p_2)$. Moreover $q_1=p_1=0$ defines an unstable periodic orbit and $q_1=0$ defines a DS separating reactants from products. This DS, similarly to the one defined in Sec. \ref{subsec:TST}, is a sphere that is due to the instability of $q_1=p_1=0$ transversal to the flow and does not admit local recrossings. The sphere itself is divided by its equator $q_1=p_1=0$ into two hemispheres with unidirectional flux - trajectories passing from reactants to products cross the hemisphere $p_1>0$, while trajectories from products to reactants cross $p_1<0$. Therefore $q_1=p_1=0$ satisfies the definition of a TS. We remark that the DS can be perturbed and as long as its boundary remains fixed and transversality is not violated, the flux through the perturbed and unperturbed DS remains the same.
 
 This description breaks down at high energies, when the periodic orbit may become stable, an event commonly referred to as loss of normal hyperbolicity. Then TST is inaccurate due to local recrossings of the DS. Loss of normal hyperbolicity occurs in the hydrogen exchange reaction, yet TST breaks down at lower energies due the presence of multiple transition states.

 Having the same energy distribution between the degrees of freedom as the periodic orbit $q_1=p_1=0$, its invariant manifolds are given by $$p_1^2-q_1^2=0,$$ the stable being $q_1=-p_1$ and the unstable $q_1=p_1$. They consist of two branches each - one on the reactant side with $q_1>0$, one on the product side with $q_1<0$. These manifolds are cylinders with the periodic orbit as its base. They are codimension-$1$ in the energy surface and separate reactive and nonreactive trajectories - reactive ones inside the cylinders
 $$\frac{1}{2}\lambda (p_1^2-q_1^2)>0,$$
 and nonreactive outside
 $$\frac{1}{2}\lambda (p_1^2-q_1^2)<0.$$
 Only reactive trajectories reach the DS.
 
 Note that in a configuration space projection, the separation between reactive and nonreactive trajectories is not as natural/obvious as in a phase space perspective. Therefore we study the structures made up of invariant manifolds that cause the reaction rate decrease in phase space.
 
 We remark that bottlenecks are related to TSs rather than potential saddle points. Sec. \ref{sec:tangles reaction} contains examples of bottlenecks unrelated to potential saddle points and a saddle point without a bottleneck.
   
 \begin{figure}
  \centering
  \includegraphics[width=0.7\textwidth]{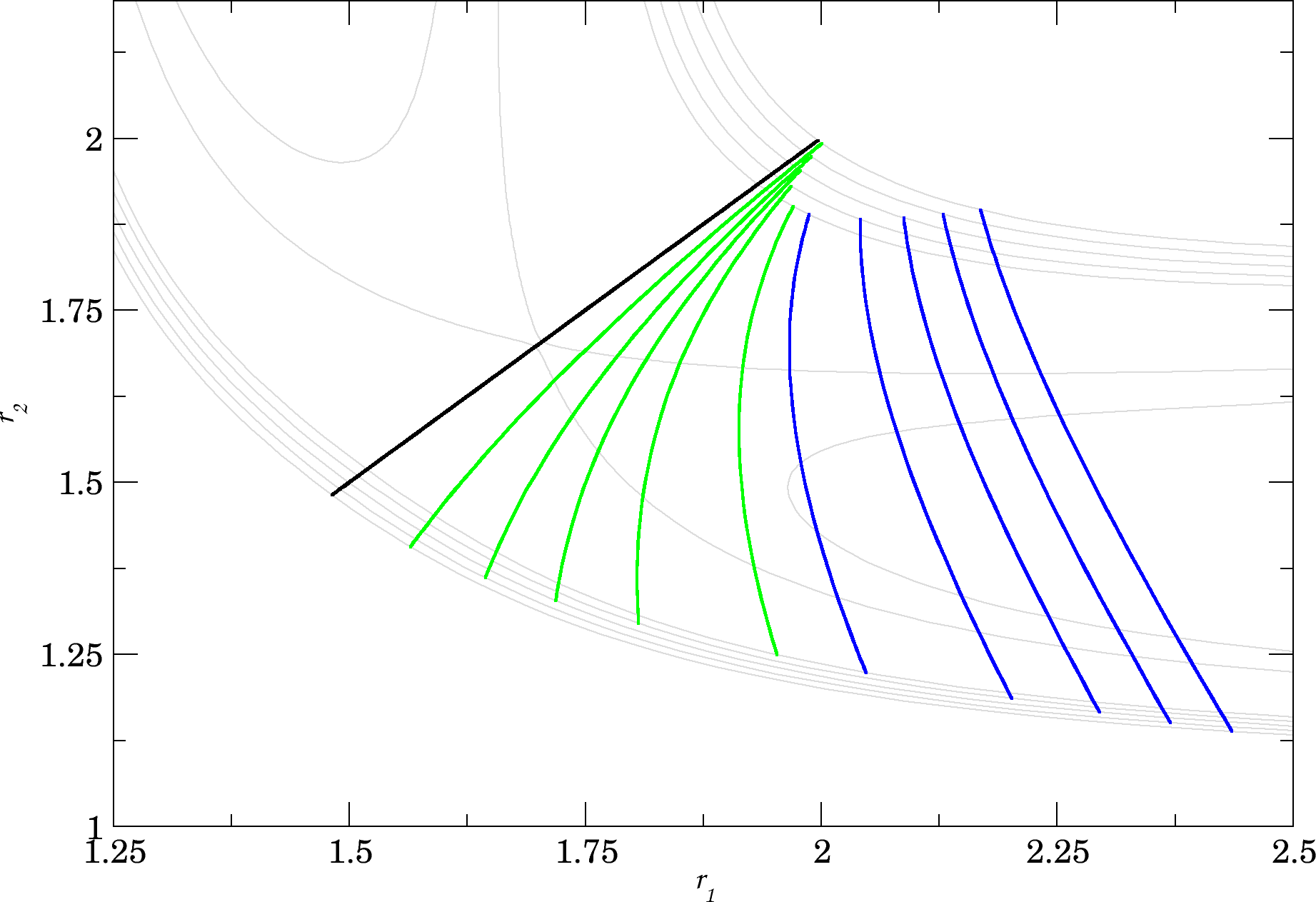}
  \caption{The projections of the periodic orbits of $F_0$ (black), $F_1$ (blue) and $F_2$ (green) onto configuration space at energies $0.02210$, $0.02300$, $0.02400$, $0.02500$ and $0.02600$ and the corresponding equipotential lines (grey).}
  \label{fig:orbits}
 \end{figure}
 \subsection{Periodic orbits}\label{subsec:po}
  For energies $E$ above $0.01456$, the energy of the saddle point, the system \eqref{eq:Ham} admits periodic orbits that come in one-parameter families parametrised by energy. Initially we focus on each family separately and subsequently we investigate the interplay that governs the complicated dynamics exhibited by this system. We adopt the notation of \cite{Inarrea11} for different families of periodic orbits $F_n$, where $n\in\mathbb{N}$, and briefly describe their evolution with increasing energy. We remark that many families come in pairs related by symmetry and for simplicity we restrict ourselves to the $r_1\geq r_2$ half plane. We will refer to orbits of the family $F_n$ on the other half plane by $\widehat{F}_n$.
  
  By $F_0$ we denote the family of Lyapunov orbits associated with the potential saddle, which as explained in Sec. \ref{subsec:local geometry} must be unstable for energies slightly above the saddle. The orbits lie on the axis of symmetry of the system $r_1=r_2$, see Fig. \ref{fig:orbits}. Orbits of this family were used in TST calculations in many of the previous works.

  A saddle-centre bifurcation at approximately $0.02204$ results in the creation of two families - the unstable $F_1$ and the initially stable $F_2$. The configuration space projections of these orbits are shown in Fig. \ref{fig:orbits}.
  The unstable family $F_1$ is the furthest away from $F_0$ and does not undergo any further bifurcations. The $F_2$ family is initially stable, but undergoes a period doubling bifurcation at $0.02208$ creating the double period families $F_{21}$ and $F_{22}$. Unlike reported by \cite{Inarrea11}, we do not find these families disappear in an inverse period doubling bifurcation of $F_2$ at $0.02651$. Instead $F_{21}$ and $F_{22}$ persist with double period until $0.02654$, when they collide together with $F_2$ and $F_0$, see Fig. \ref{fig:bif2}. Consequently $F_0$ becomes stable. We would like to enhance the findings of \cite{Inarrea11} by remarking that $F_{21}$ and $F_{22}$ are briefly stable between switching from hyperbolic to inverse hyperbolic and vice versa, see Fig. \ref{fig:bif1}.
  
  At $0.02661$, $F_0$ is involved in a bifurcation with a double period family $F_4$ that originates in a saddle-centre bifurcation at $0.02254$. $F_4$ is a family symmetric with respect to $r_1=r_2$. For dynamical purposes we point out that above $0.02661$ $F_0$ is inverse hyperbolic. 
  \begin{figure}
   \centering
   \includegraphics[width=0.49\textwidth]{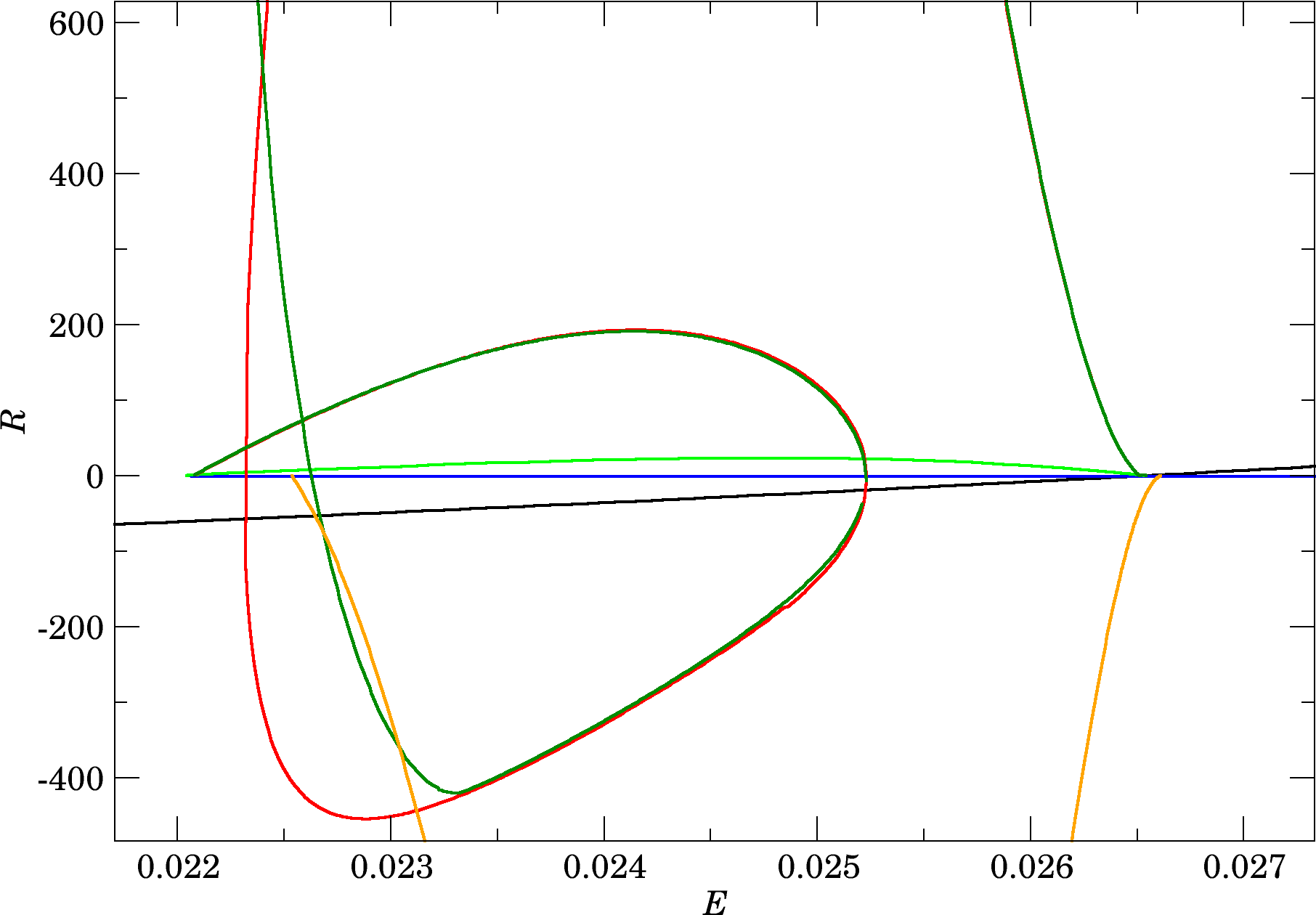}%
   \includegraphics[width=0.48\textwidth]{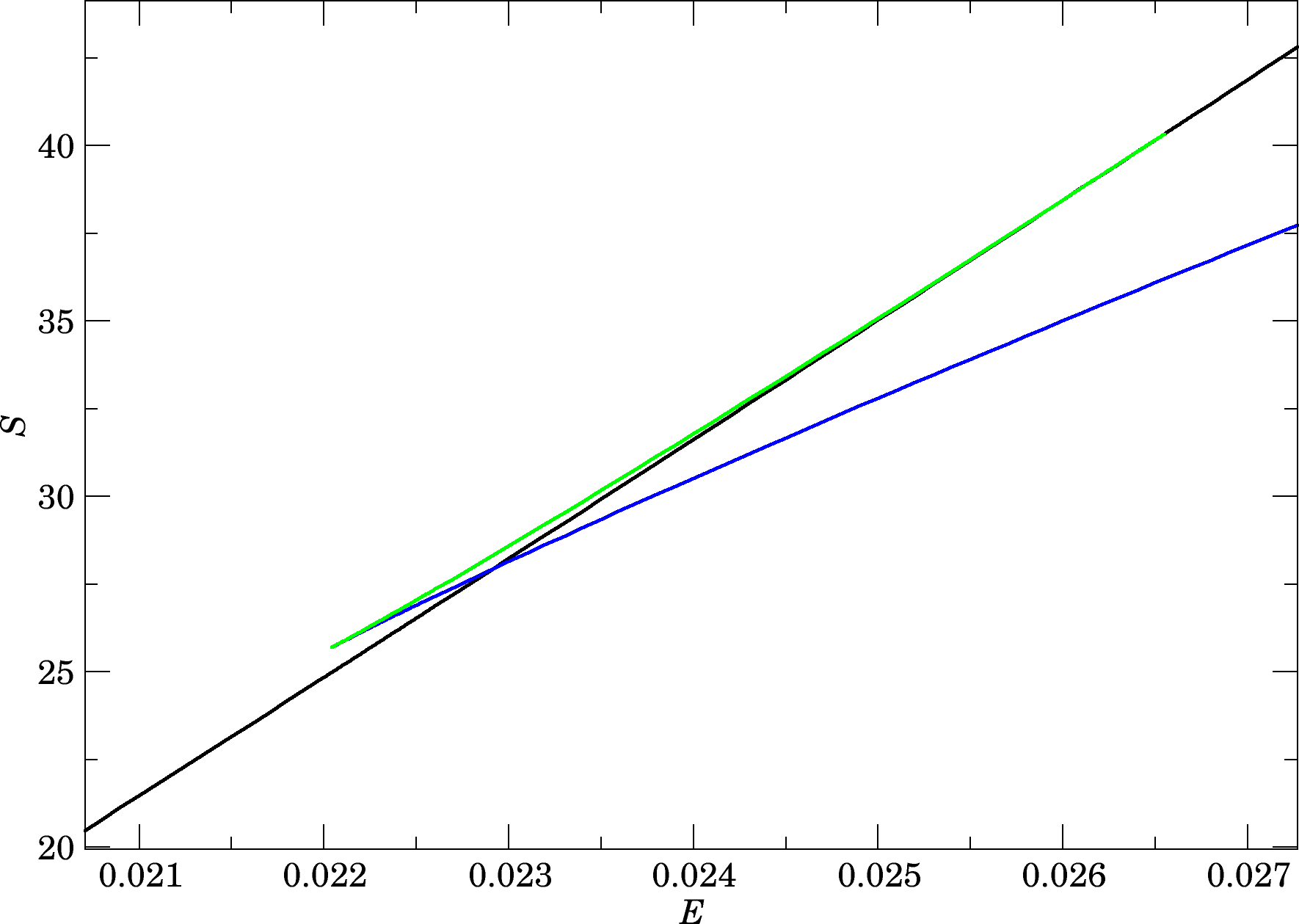}
   \caption{Bifurcation diagrams showing the evolution of $F_0$ (black), $F_1$ (blue), $F_2$ (light green), $F_{21}$ (dark green), $F_{22}$ (red) and $F_4$ (orange) on the energy-residue ($E,R$) and the energy-action ($E,S$) plane. The residues of other families and the action of orbits of period higher than $1$ are omitted for the sake of clarity.
   }\label{fig:bif1}
  \end{figure}
  \begin{figure}
   \centering
   \includegraphics[width=0.48\textwidth]{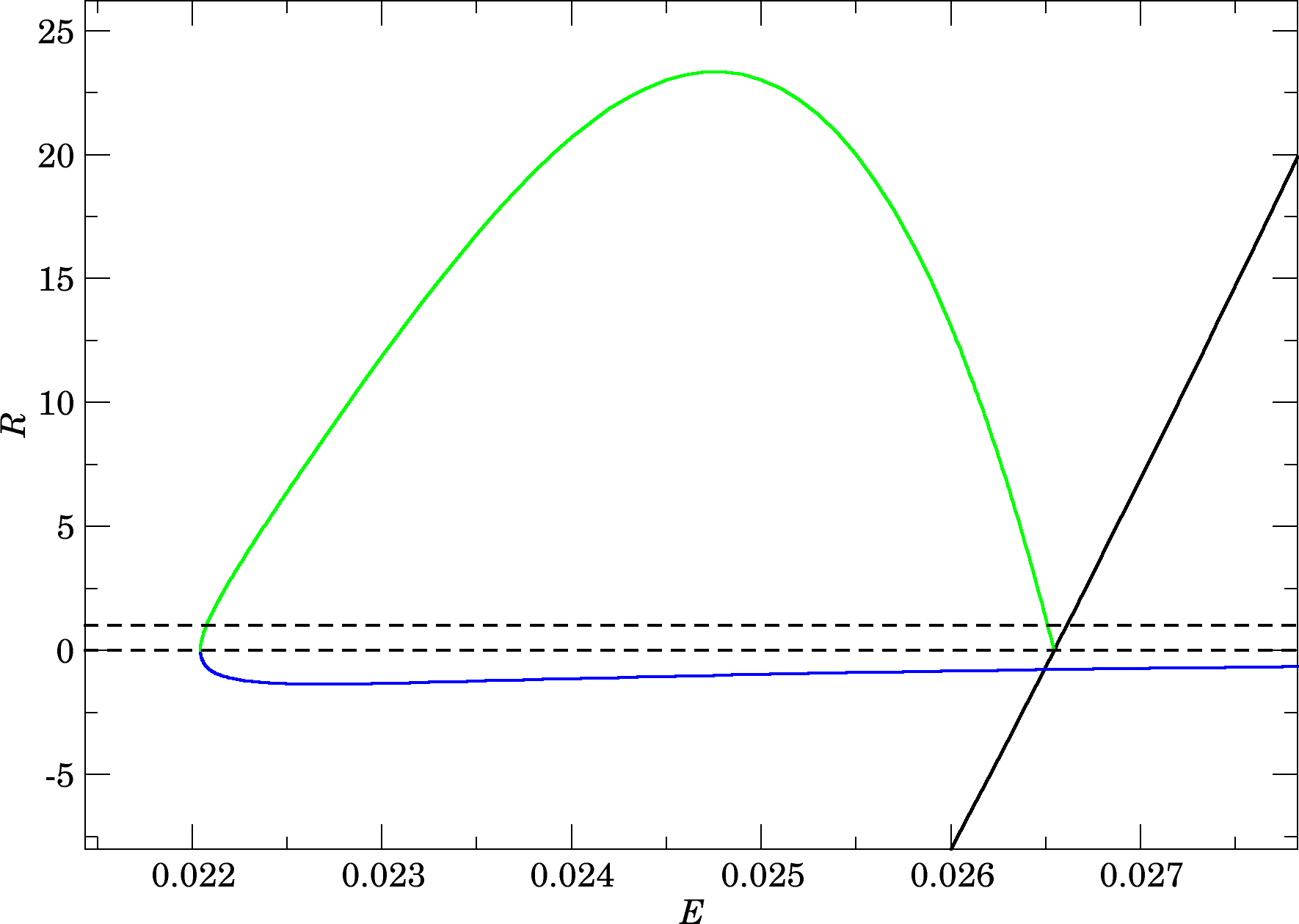}%
   \includegraphics[width=0.48\textwidth]{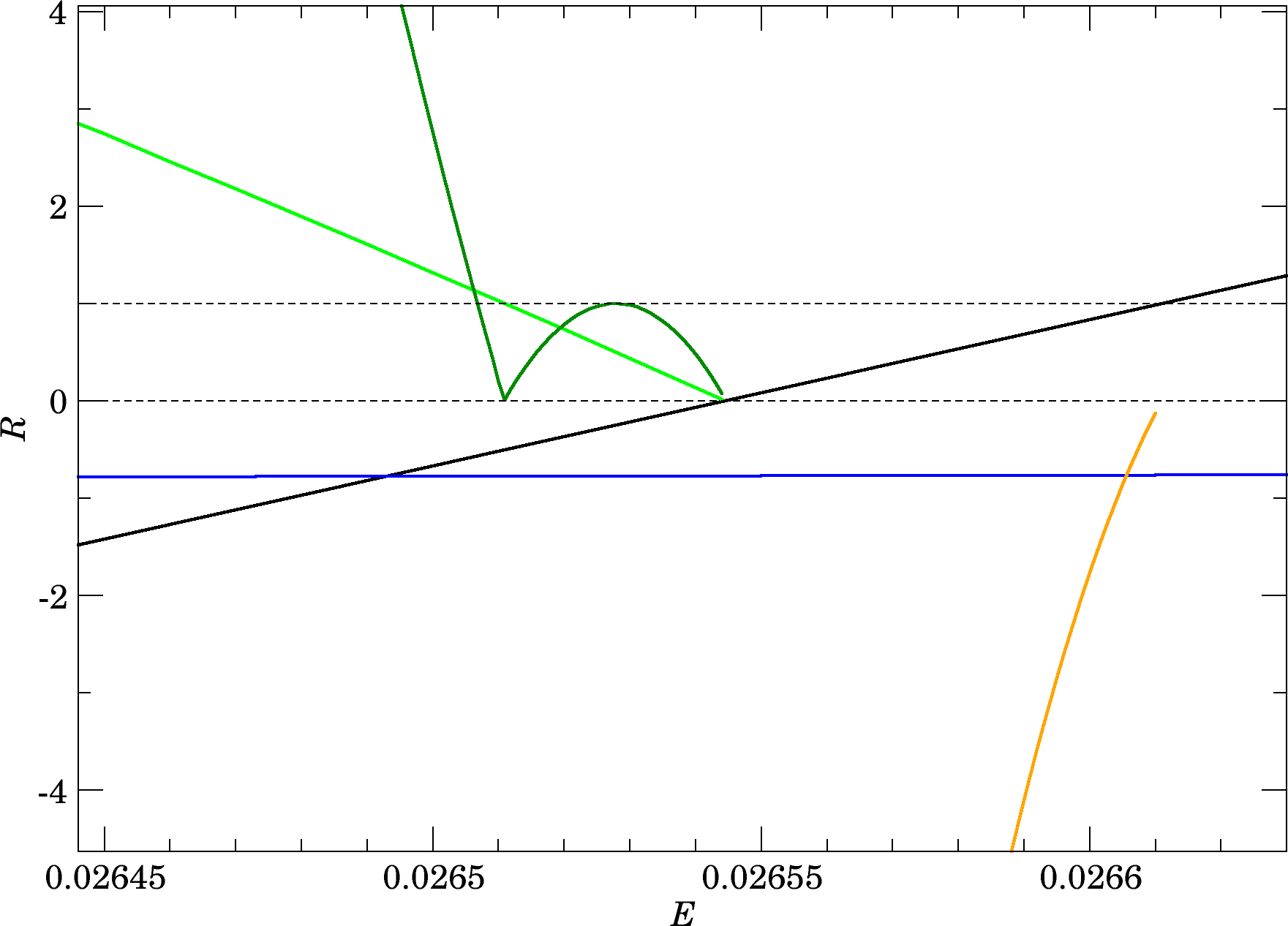}
   \caption{Details of the evolution of $F_0$ (black), $F_1$ (blue), $F_2$ (light green), $F_{21}$ (dark green), $F_{22}$ (red, identical with $F_{21}$) and $F_4$ (orange) on the energy-residue ($E,R$) plane.
   }\label{fig:bif2}
  \end{figure}
  
  Fig. \ref{fig:bif1} and \ref{fig:bif2} show bifurcation diagrams of most of the families on the energy-residue and the energy-action plane. By residue $R$ we mean the Greene residue as introduced by J. M. Greene in \cite{Greene68}, where $R<0$ means that the periodic orbit is hyperbolic, $0<R<1$ means it is elliptic and $R>1$ means it is inverse hyperbolic.
  
  The residue is derived from a matrix that describes the local dynamics near a periodic orbit - the monodromy matrix. Let $\Gamma$ be a periodic orbit with the parametrisation $\gamma(t)$ and period $T$, and $M(t)$ be the matrix satisfying the variational equation
  \begin{equation}
  \dot{M}(t)=JD^2 H(\gamma(t))M(t),
  \label{eq:Monod}
  \end{equation}
  where $J=\begin{pmatrix} 0 & Id\\ -Id & 0\end{pmatrix}$,
  with the initial condition $M(0)=Id.$
  The monodromy matrix is defined by $M=M(T)$ and it describes how a sufficiently small initial deviation $\delta$ from $\gamma(0)$ changes after a full period $T$:
  \begin{equation*}
  \Phi_H^T(\gamma(0)+\delta)=\gamma(T)+M\delta+O(\delta^2),
  \end{equation*}
  where $\Phi_H^t$ is the Hamiltonian flow.
  
  According to \cite{Eckhardt91}, if $\delta$ is an initial displacement along the periodic orbit $\delta\parallel J\nabla H$, then $\delta$ is preserved after a full period $T$, i.e. $M\delta=\delta$. Similarly an initial displacement perpendicular to the energy surface $\delta\parallel \nabla H$ is preserved. Consequently, two of the eigenvalues of $M$ are 
  \begin{equation}
   \lambda_1=\lambda_2=1.
   \label{eq:lambda1}
  \end{equation}

  As \eqref{eq:Monod} is Hamiltonian, the preservation of phase space volume following Liouville's theorem implies $\det M(t)=\det M(0)=1$ for all $t$. Therefore the two remaining eigenvalues must satisfy $\lambda_3\lambda_4=1$ and we can write them as $\lambda$ and $\frac{1}{\lambda}$. $\Gamma$ is hyperbolic if $\lambda>1$, it is elliptic if $|\lambda|=1$ and it is inverse hyperbolic if $\lambda<-1$.

  \begin{definition}
  The Greene residue of $\Gamma$ is defined as $R=\frac{1}{4}(4-Tr M),$
  where $M$ is the monodromy matrix corresponding to the periodic orbit $\Gamma$. 
  \end{definition}

  Using \eqref{eq:lambda1} we can write $R$ as
  $$R=\frac{1}{4}\left(2-\lambda-\frac{1}{\lambda}\right).$$
  By definition $R<0$ if $\Gamma$ is hyperbolic, $0<R<1$ if it is elliptic and $R>1$ if it is inverse hyperbolic.

  Davis \cite{Davis87} mostly focused on the energy interval below $0.02214$ and above $0.02655$, the interval where TST is exact and the interval where two TSs exist, respectively.
  
  In the light of normal form approximation described in Sec. \ref{subsec:local geometry}, we remark that the approximation breaks down completely when $F_0$ loses normal hyperbolicity at $0.02655$ at the latest. The loss of normal hyperbolicity is not the cause for the overestimation of the reaction rate by TST as it starts to deviate from the Monte Carlo rate well before $0.02300$.

  \subsection{Phase space regions}\label{subsec:regions}
  We would like to give up the binary partitioning of an energy surface into reactants and products in favour of defining an interaction region inbetween into which trajectories can only enter once.
  
  As explained in Sec. \ref{subsec:local geometry}, TSs give rise to bottlenecks in phase space. Because $F_1$ gives rise to the bottleneck the furthest away from the potential barrier, we use it to delimit regions as follows. Denote DS$_1$ and DS$_{\widehat{1}}$ the DSs constructed using $F_1$ and $\widehat{F}_1$ according to Sec. \ref{subsec:TST}. The interaction region is the region of the energy surface between the two DSs and it contains all other periodic orbits. Reactants and products are the regions on the $r_1>r_2$-side and the $r_1<r_2$-side of the interaction region respectively, see Figure \ref{fig:reg_config}.
  
  \begin{figure}
  \centering
   \includegraphics[width=0.7\textwidth]{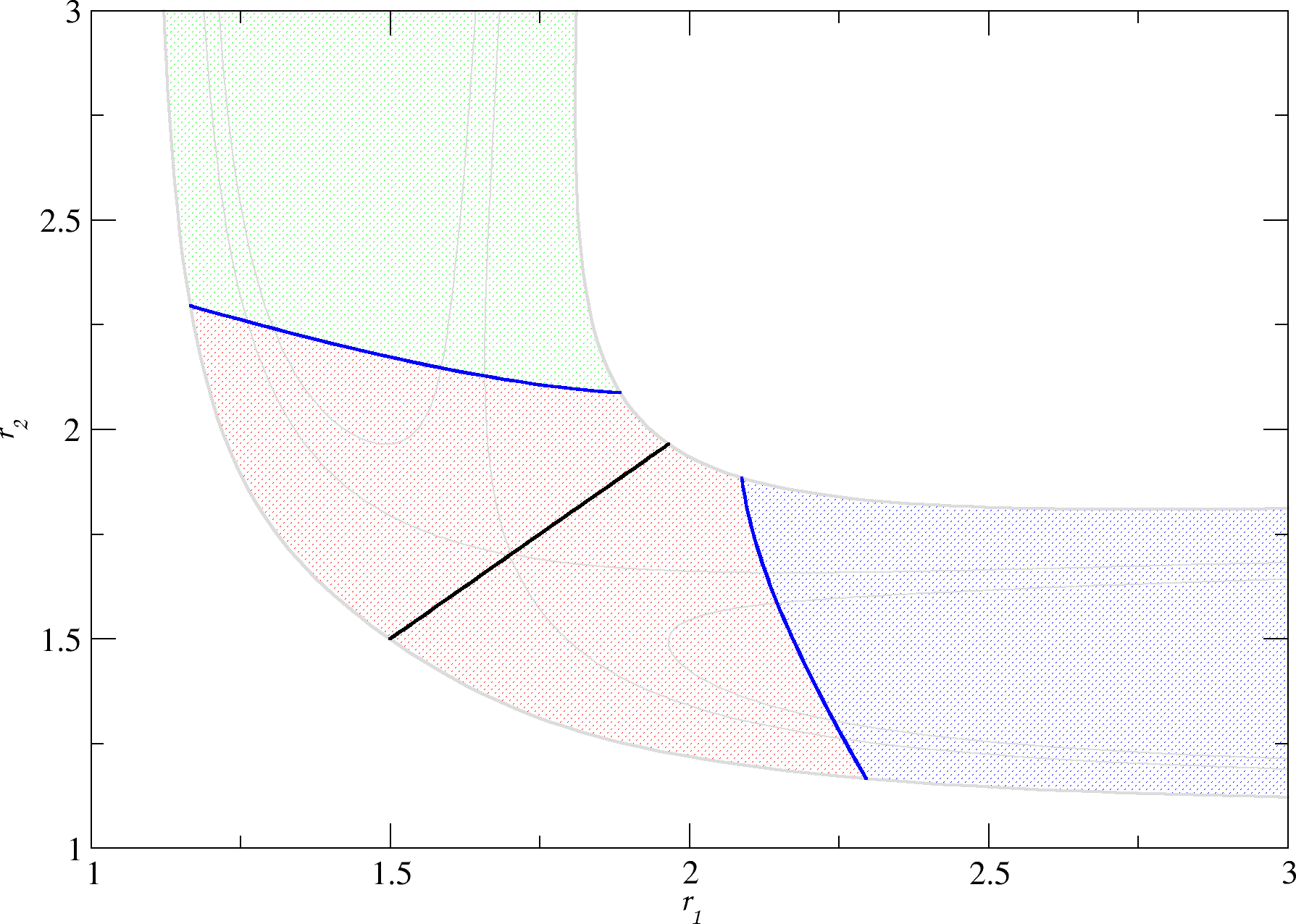}
   \caption{Regions in configuration space at energy $0.02400$. The interaction region (red) bounded by two orbit from the family $F_1$ (blue), the region of reactants (blue) and the region of products (green). The orbit $F_0$ (black) is also included.
   }\label{fig:reg_config}
  \end{figure}
  
  The advantages of this partition of space are immediate.
  \begin{itemize}
   \item All TSs and bottlenecks are in the interaction region or on its boundary. The dynamics in reactants and products has no influence on reactivity and to fully understand the hydrogen exchange reaction, it is enough to restrict the study to the interaction region.
   
   \item Trajectories that leave the interaction region never return. This is true in forward and backward time.
   
   \item It is impossible for a trajectory to enter reactants and products in the same time direction, unlike in the binary partitioning, where trajectories may oscillate between reactants and products.
  \end{itemize}

 \section{Definition of a Poincar\'e surface of section}\label{sec:define sos}
  Invariant manifolds are $2$ dimensional objects on the $3$ dimensional energy surface embedded in $4$ dimensional phase space. To facilitate the study of intersections of invariant manifolds, we define a $2$ dimensional surface of section on the energy surface that is transversal to the flow and intersects invariant manifolds in $1$ dimensional curves.
 
  \subsection{Reaction coordinate and minimum energy path}\label{subsec:mep}  
  \begin{figure}
   \centering
   \includegraphics[width=0.7\textwidth]{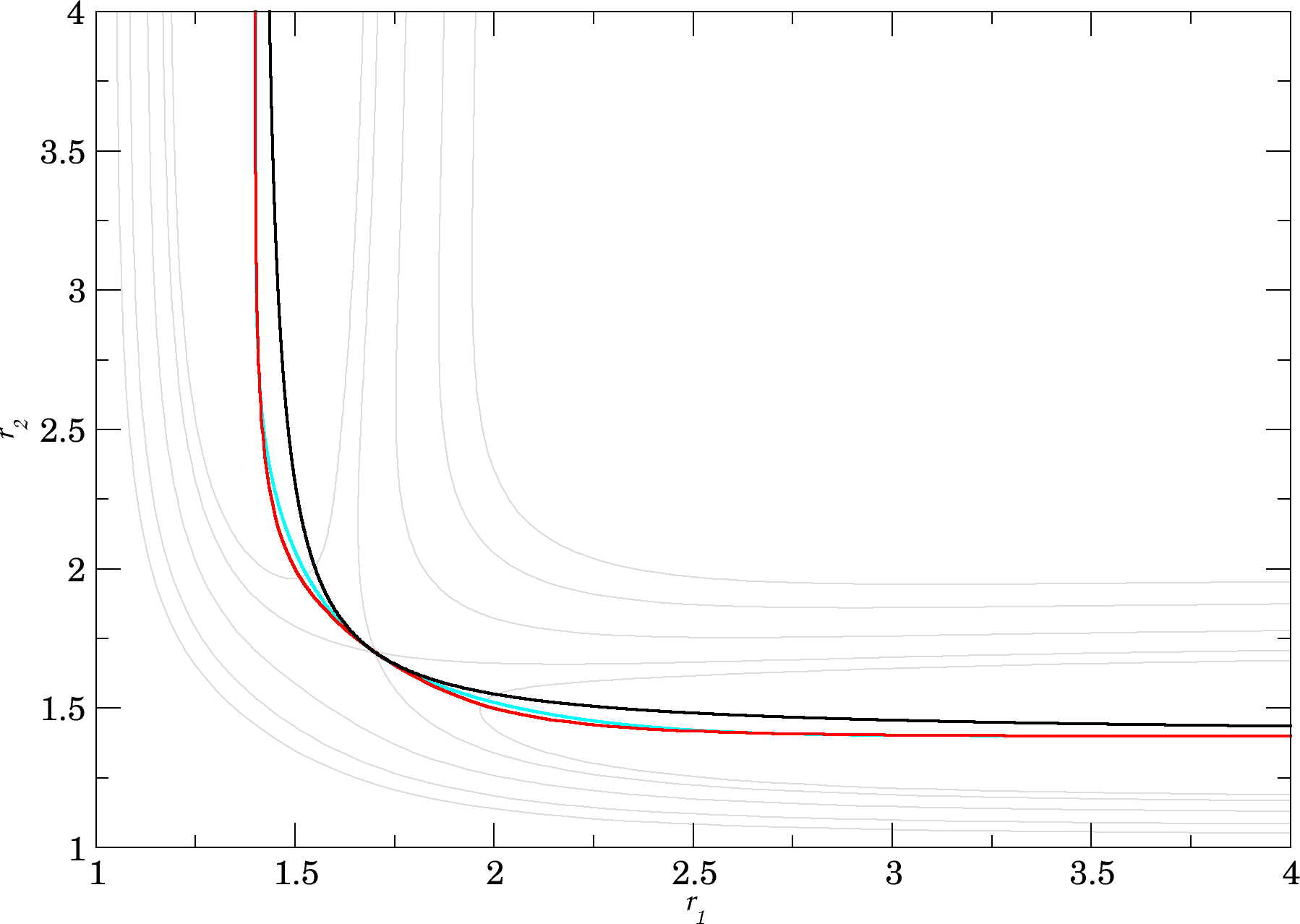}
   \caption{Comparison of the MEP (red), the coordinate line $q_1=0$ (black) and the coordinate line $\tilde{q}_1=0$ (cyan). Equipotential lines of the potential energy surface correspond to energies $0.01200$, $0.01456$, $0.02000$, $0.02800$ and $0.03500$
   }\label{fig:minE}
  \end{figure}
  
  Here we define a reaction coordinate, using which we can monitor the progress along a reaction pathway. Frequently a reaction coordinate is closely related to a \emph{minimum energy path} (MEP) connecting the potential wells of reactants and products via the potential saddle. The coordinate as such is not a solution of the Hamiltonian system and, as remarked in \cite{Pechukas1976mep}, is of no dynamical significance to the system.  

  A MEP can be defined as the union of two paths of steepest descend, the unique solutions of the gradient system
  \begin{equation*}
   \dot{r}_1 = -\frac{\partial U}{\partial r_1},\qquad \dot{r}_2 = -\frac{\partial U}{\partial r_2}, 
  \end{equation*}
  one connecting the saddle $(R_s,R_s)$ to the potential well $(\infty,R_{min})$, the other connecting $(R_s,R_s)$ to $(R_{min},\infty)$. Fig. \ref{fig:minE} shows the MEP on a contour plot of $U$.

  \subsection{Surface of section}\label{subsec:sos}
  The MEP as defined above does not have an analytic expressing, but can be approximated using $q_1=0$, where
  \begin{equation*}
   q_1=(r_1-R_{min})(r_2-R_{min})-(R_s-R_{min})^2,
  \end{equation*}
  as used by \cite{Davis87} and shown in Figure \ref{fig:minE}. Invariant manifolds are always transversal to the MEP and transversal to $q_1=0$ for the energy interval considered in this work. At higher energies Davis used $q_1=-0.04$, $q_1=-0.07$ and $q_1=-0.084$ to avoid tangencies.
  
  We found that
  \begin{equation}
   \tilde{q}_1=(r_1-R_{min})(r_2-R_{min})-(R_s-R_{min})^2e^{-2((r_1-R_s)^2+(r_2-R_s)^2)},
  \end{equation}
  approximates the MEP significantly better, but a coordinate system involving $\tilde{q}_1$ is rather challenging to work with.
  
  Throughout this work we use the surface of section $\Sigma_0$ defined by $q_1=0$, $\dot{q}_1>0$. The condition $\dot{q}_1>0$ determines the sign of the momenta and guarantees that each point on $\Sigma_0$ corresponds to a unique trajectory. We remark that the boundary of $\Sigma_0$ does not consist of invariant manifolds and therefore it is not a surface of section in the sense of Birkhoff \cite[Chapter 5]{Birkhoff27}.
  
  For the sake of utility, we define the other coordinate $q_2$ such that $(q_1,q_2)$ is an orthogonal coordinate system on $\mathbb{R}^2$ and the coordinate lines of $q_2$ are symmetric with respect to $r_1=r_2$. These conditions are satisfied by
  \begin{equation}
  q_2=\frac{1}{2}(r_1-R_{min})^2-\frac{1}{2}(r_2-R_{min})^2.
  \end{equation}
  Note that $q_2=0$ is equivalent to $r_1=r_2$ and $q_2$ is a reaction coordinate - it captures progress along $q_1=0$ and $q_2>0$ contains reactants, while $q_2<0$ contains products. We remark that $q_1$ can locally considered a \emph{bath coordinate} capturing oscillatory motion near the potential barrier. For a fixed energy, the energy surface is bounded in $q_1$ and unbounded in $q_2$.
  
  \subsection{Symplectic coordinate transformation}\label{subsec:momenta}
  Here we define a coordinate system in phase space, such that the coordinate transformation is symplectic.
  This requires finding the conjugate momenta $p_1$, $p_2$ corresponding to $q_1$, $q_2$. For this purpose we use the following generating function (type 2 in \cite{Arnold76}):
  \begin{multline*}
   G(r_1,r_2,p_1,p_2)= \big((r_1-R_{min})(r_2-R_{min})-(R_s-R_{min})^2\big)p_1\\ +\frac{1}{2}\big((r_1-R_{min})^2-(r_2-R_{min})^2\big)p_2.
  \end{multline*}
  Then
  \begin{equation*}
  \frac{\partial G}{\partial r_i}=p_{r_i},\qquad \frac{\partial G}{\partial p_i}=q_i.
  \end{equation*}
  One finds that
  \begin{equation*}
  \begin{split}
   p_{r_1}&=\frac{\partial G}{\partial r_1}=(r_2-R_{min})p_1+(r_1-R_{min})p_2,\\
   p_{r_2}&=\frac{\partial G}{\partial r_2}=(r_1-R_{min})p_1-(r_2-R_{min})p_2.
  \end{split}
  \end{equation*}
  From this we obtain
  \begin{equation*}
  \begin{split}
   p_1&=\frac{(r_2-R_{min})p_{r_1}+(r_1-R_{min})p_{r_2}}{(r_1-R_{min})^2+(r_2-R_{min})^2},\\
   p_2&=\frac{(r_1-R_{min})p_{r_1}-(r_2-R_{min})p_{r_2}}{(r_1-R_{min})^2+(r_2-R_{min})^2}.
  \end{split}
  \end{equation*}
  
  This transformation has a singularity at $r_1=r_2=R_{min}$, but $U(R_{min},R_{min})=0.03845$ is inaccessible at energies we consider. By straightforward calculation one finds that the symplectic $2$-form $\omega_2$ is indeed preserved:
  $$\omega_2=\dif p_{r_1}\wedge \dif r_1 + \dif p_{r_2}\wedge \dif r_2 = \dif p_1\wedge \dif q_1 + \dif p_2 \wedge \dif q_2.$$
  We remark that $(q_2,p_2)$ as defined above are the canonical coordinates on $\Sigma_0$.
    
  
  \section{Transport and barriers}\label{sec:transport barriers} 
  In this section we discuss the dynamics on the surface of section $q_1=0$ under the return map. This involves investigating structures formed by invariant manifolds via lobe dynamics due to \cite{Rom-Kedar90}.
  
  \subsection{Structures on the surface of section} 
  The return map $P$ associated with $\Sigma_0$ is defined as follows. Every point $(q^0,p^0)$ on $\Sigma_0$ is mapped to 
  $$P(q^0,p^0)=(q_2(T),p_2(T)),$$
  where $T>0$ is the smallest for which $q_1(T)=0$ along the solution $$(q_1(t),p_1(t),q_2(t),p_2(t)),$$ with the initial condition
  $$(q_1(0),p_1(0),q_2(0),p_2(0))=(0,p_1,q^0,p^0),$$
  where $p_1$ is given implicitly by the fixed energy $E$.
  $P$ is symplectic because it preserves the canonical $2$-form restricted to $\Sigma_0$,
  \begin{equation}
   \omega_2\bigr\vert_{\Sigma_0}=\dif p_2\wedge\dif q_2,
  \label{eq:omega2sigma}
  \end{equation}
  see \cite{Binney85}. Because the Hamiltonian flow is reversible, $P^{-1}$ is well defined.
  
  Each periodic orbit intersects $\Sigma_0$ in a single point that is a fixed point of $P$. Its stability follows from the eigenvalues of the monodromy matrix, as explained in Sec. \ref{subsec:po}. Due to conservation laws, the eigenvalues can be written as $\lambda$, $\frac{1}{\lambda}$, $1$, $1$, see \cite{Eckhardt91}. For TSs, the eigenvectors corresponding to $\lambda$, $\frac{1}{\lambda}$ define stable and unstable invariant manifolds under the linearisation of $P$ near a fixed point.
  
  \subsection{Barriers formed by invariant manifolds}\label{subsec:barrier}
    \begin{figure}
  \centering
  \includegraphics[width=0.49\textwidth]{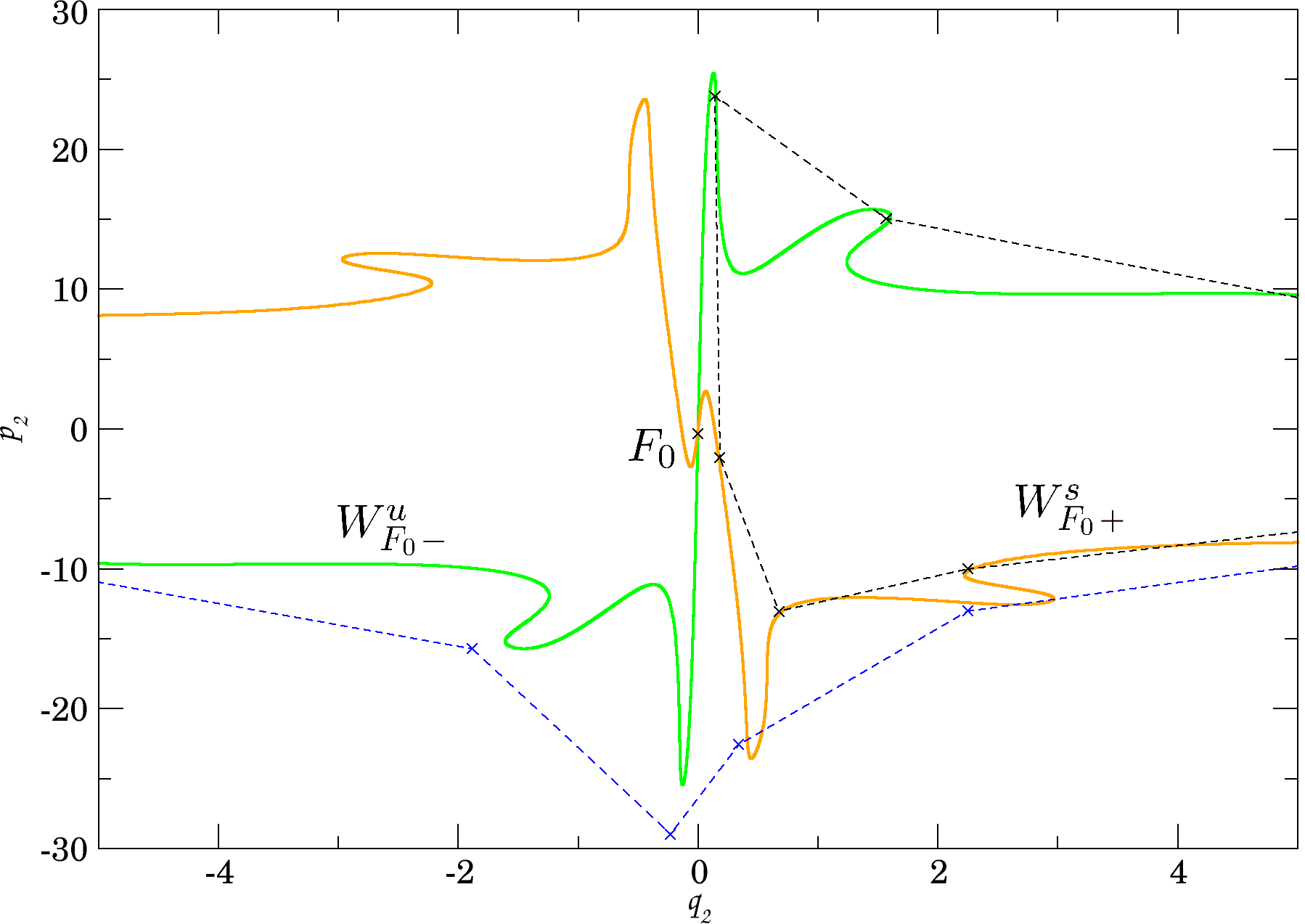}
  \includegraphics[width=0.49\textwidth]{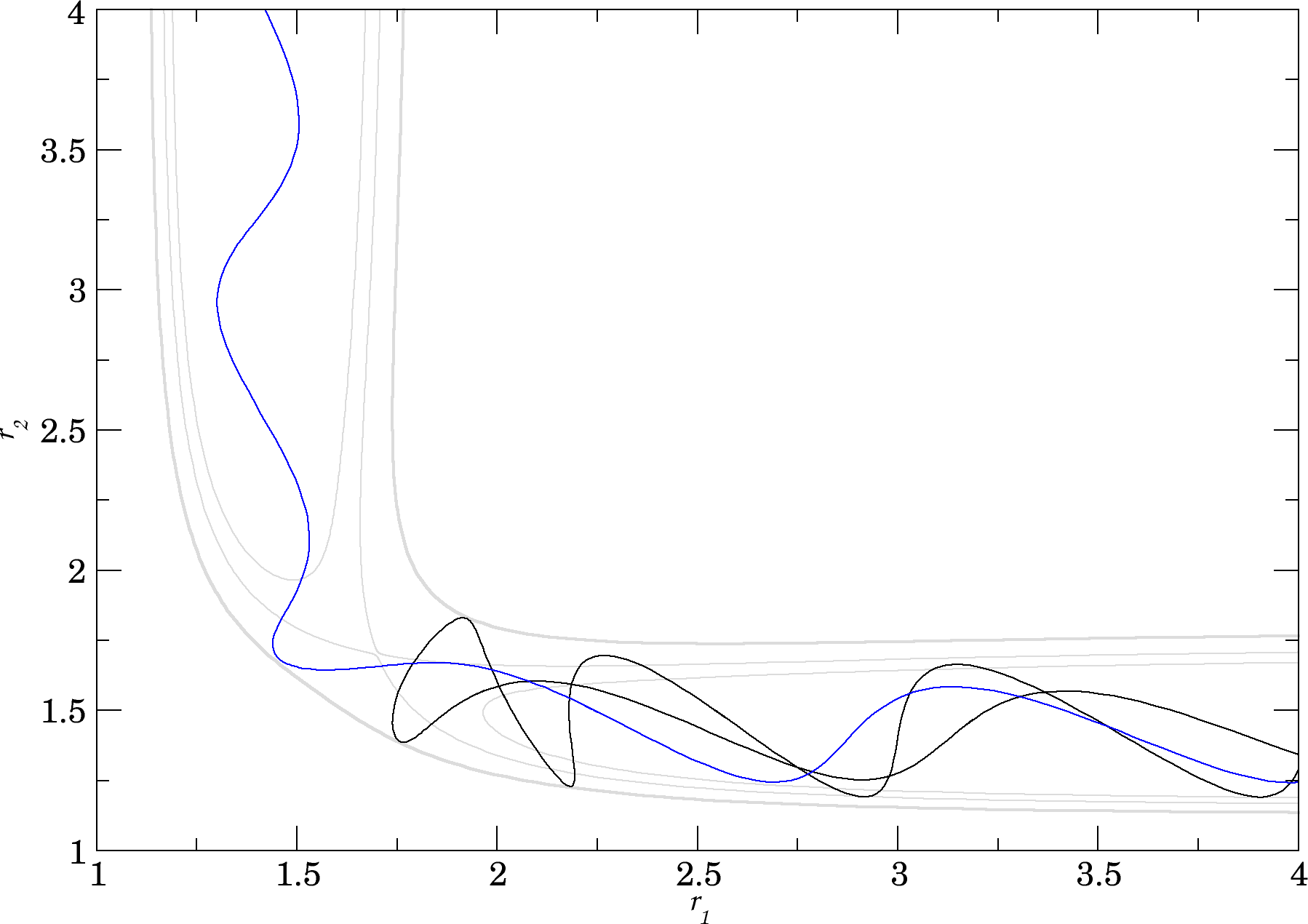}
  \caption{Disjoint invariant manifolds of $F_0$ forming a barrier on $\Sigma_0$ at $0.01900$ and examples of a nonreactive (black) and a reactive (blue) trajectory on $\Sigma_0$ and in configuration space.
  }\label{fig:1900}
  \end{figure}
  
  In the following we discuss invariant manifolds of TSs and their impact on dynamics with increasing energy. Let $F_i$ be a TS, we denote $W_{F_i}$ its invariant manifolds as a whole, stable and unstable invariant manifolds are denoted $W^s_{F_i}$ and $W^u_{F_i}$ respectively. An additional $+/-$ subscript indicates the branch of the invariant manifold with larger/smaller $q_2$ coordinate in the neighbourhood of $F_i$, for example $W^s_{F_i+}$ and $W^s_{F_i-}$. Recall from Sec. \ref{subsec:local geometry} that invariant manifolds of unstable brake orbits are cylinders of codimension-$1$ on the energy surface and they intersect $\Sigma_0$ in curves that divide $\Sigma_0$ into two disjoint parts each.
  
  As mentioned in Sec. \ref{subsec:po}, the system has a single periodic orbit $F_0$ between $0.01456$ and $0.02204$. Its invariant manifolds do not intersect and act as separatrices or \emph{barriers} between reactive and nonreactive trajectories, as shown at $0.01900$ in Fig. \ref{fig:1900}. Reactive trajectories are characterised by a large $|p_2|$ momentum and are located above and below $W_{F_0}$. Nonreactive ones have a smaller $|p_2|$ momentum and are located between $W^s_{F_0}$ and $W^u_{F_0}$. Consequently DS$_0$, the DS associated with $F_0$, has the no-return property and TST is exact (\cite{Davis87}).
  
  $F_1$ and $F_2$ come into existence at $0.02204$, but the reaction mechanism is governed entirely by $W_{F_0}$. $W_{F_1}$ form a homoclinic tangle, but it only contains nonreactive trajectories. TST remains exact until $0.02215$, when a heteroclinic intersection of $W_{F_0}$ and $W_{F_1}$ first appears. In the following we introduce the notation for homoclinic and heteroclinic tangles and subsequently introduce lobe dynamics due to \cite{Rom-Kedar90} on the example of the homoclinic tangle formed by $W_{F_1}$, the $F_1$ tangle.

  \subsection{Definitions and notations}\label{subsec:def manif}
  Let $F_i$ and $F_j$ be fixed points and assume $W^s_{F_i}$ and $W^u_{F_j}$ intersect transversally, as is the case in this system. The \emph{heteroclinic point} $Q\in W^s_{F_i}\cap W^u_{F_j}$ converges to $F_i$ as $t\rightarrow\infty$ and to $F_j$ as $t\rightarrow-\infty$. The images and preimages of $Q$ under $P$ are also heteroclinic points and therefore $W^s_{F_i}$ and $W^u_{F_j}$ intersect infinitely many times creating a \emph{heteroclinic tangle}. If $i=j$, we speak of homoclinic points and homoclinic tangles.
  
  Homoclinic and heteroclinic tangles are chaotic, since dynamics near its fixed points is locally conjugate to Smale's horseshoe dynamics (see \cite{Hirsch04}).
  
  Denote the segment of $W^s_{F_i}$ between $F_i$ and $Q$ by $S[F_i,Q]$ and the segment of $W^u_{F_j}$ between $F_j$ and $Q$ by $U[F_j,Q]$.  
  \begin{definition}
   If $S[F_i,Q]$ and $U[F_j,Q]$ only intersect at $Q$ (and $F_i$ if $i=j$), then $Q$ is a \emph{primary intersection point} (pip).
  \end{definition}

  It should be clear that every tangle necessarily has pips. If $Q$ is a pip, then $PQ_0$ is a pip too, because if $S[F_i,Q]\cap U[F_j,Q]=\{Q\}$, then $S[F_i,PQ]\cap U[F_j,PQ]=\{PQ\}$. Similarly $P^{-1}Q$ is a pip. We remark that by definition all pips lie on $S[F_i,Q]\cup U[F_j,Q]$.
    
  \begin{definition}
   Let $Q_0$ and $Q_1$ be pips such that $S[Q_1,Q_0]$ and $U[Q_0,Q_1]$ do not intersect in pips except for their end points. The set bounded by $S[Q_1,Q_0]$ and $U[Q_0,Q_1]$ is called a \emph{lobe}.
  \end{definition}
  
  Note that the end points of the segments are ordered, the first being closer to the fixed point along corresponding the manifold in terms of arclength on $\Sigma_0$. Clearly $P$ preserves this ordering. It follows that if $S[Q_1,Q_0]$ and $U[Q_0,Q_1]$ do not intersect in pips except for the endpoints, $S[PQ_1,PQ_0]$ and $U[PQ_0,PQ_1]$ cannot intersect in pips other than the end points. Therefore $P$ always maps lobes to lobes.
  
  \subsection{A partial barrier}\label{subsec:partial barrier}
  Without knowing about invariant manifolds, the influence of a tangle on transport between regions of a Hamiltonian system may seem unpredictable and random. The role of invariant manifolds is well known and the transport mechanism may be intricate, yet understandable.
  
  We explain this mechanism on the example of the $F_1$ tangle. The analogue in heteroclinic tangles will be apparent. The choice of the $F_1$ tangle at $0.02206$ is due to the logical order in terms of increasing energy and its relative simplicity. Of the invariant manifolds, $W^s_{F_1+}$ and $W^u_{F_1+}$ form barriers similar to those discussed in Sec. \ref{subsec:barrier} at all energies, while $W^s_{F_1-}$ and $W^u_{F_1-}$ form a homoclinic tangle. All branches of $W_{F_1}$ lie in the region of nonreactive trajectories on the reactant side of $F_0$, see Figure \ref{fig:2206_all}.
  
  \begin{figure}
  \centering
   \includegraphics[width=0.7\textwidth]{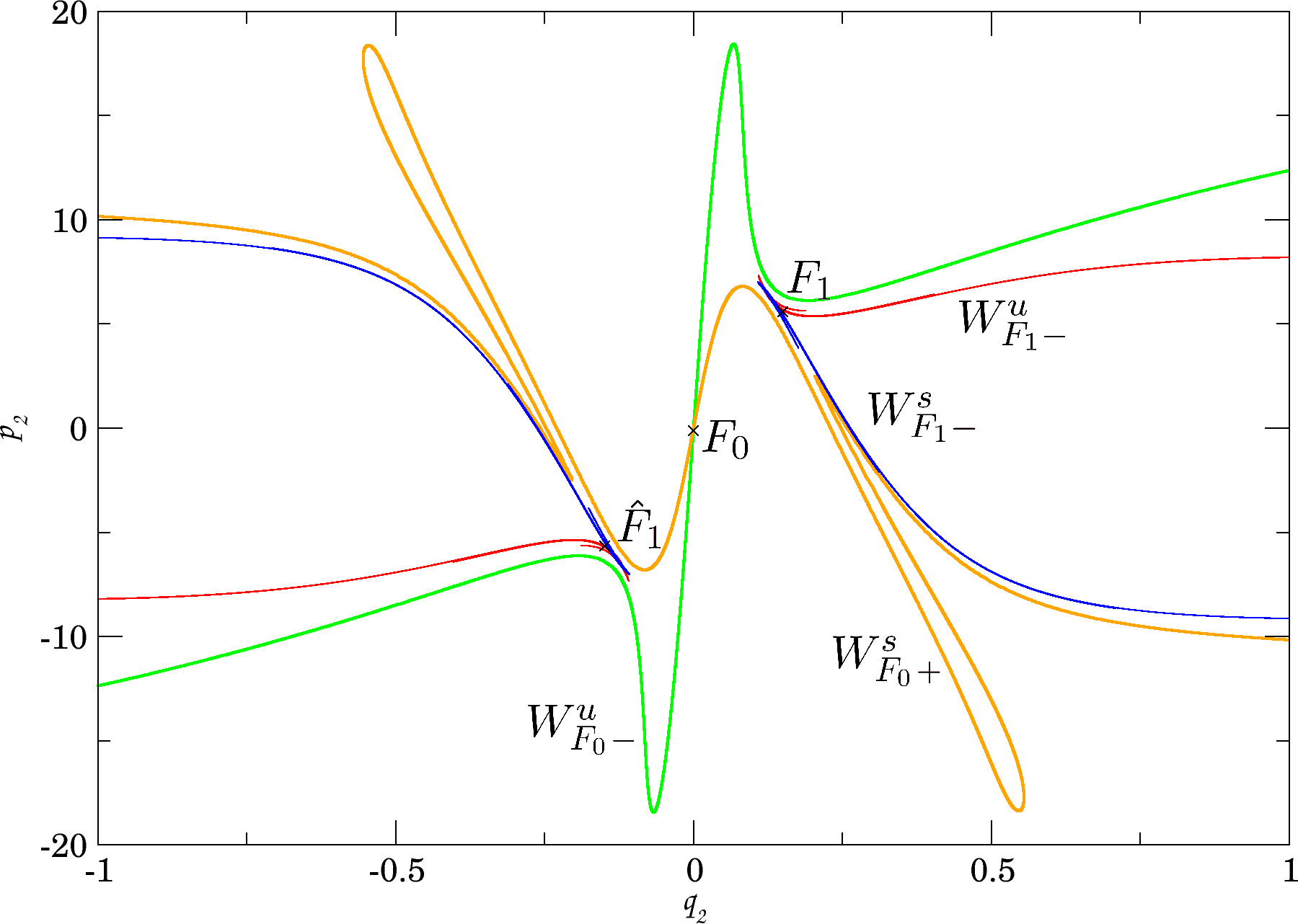}
  \caption{Invariant manifolds of $F_0$, $F_1$ and $\widehat{F}_1$ at $0.02206$.}\label{fig:2206_all}
  \end{figure}  
  
  Choose a pip $Q_0\in W^s_{F_1-}\cap W^u_{F_1-}$, we will comment on the negligible consequences of choice later. The segments $S[F_1,Q_0]$ and $U[F_1,Q_0]$ delimit a region that we denote in reference to $F_1$ by $R_1$. The complement to $R_1$ in the region bounded by $W^s_{F_0+}$ and $W^u_{F_0+}$ is denoted $R_0$, see Figure \ref{fig:2206}.
  
  \begin{figure}
  \centering
  \includegraphics[width=0.49\textwidth]{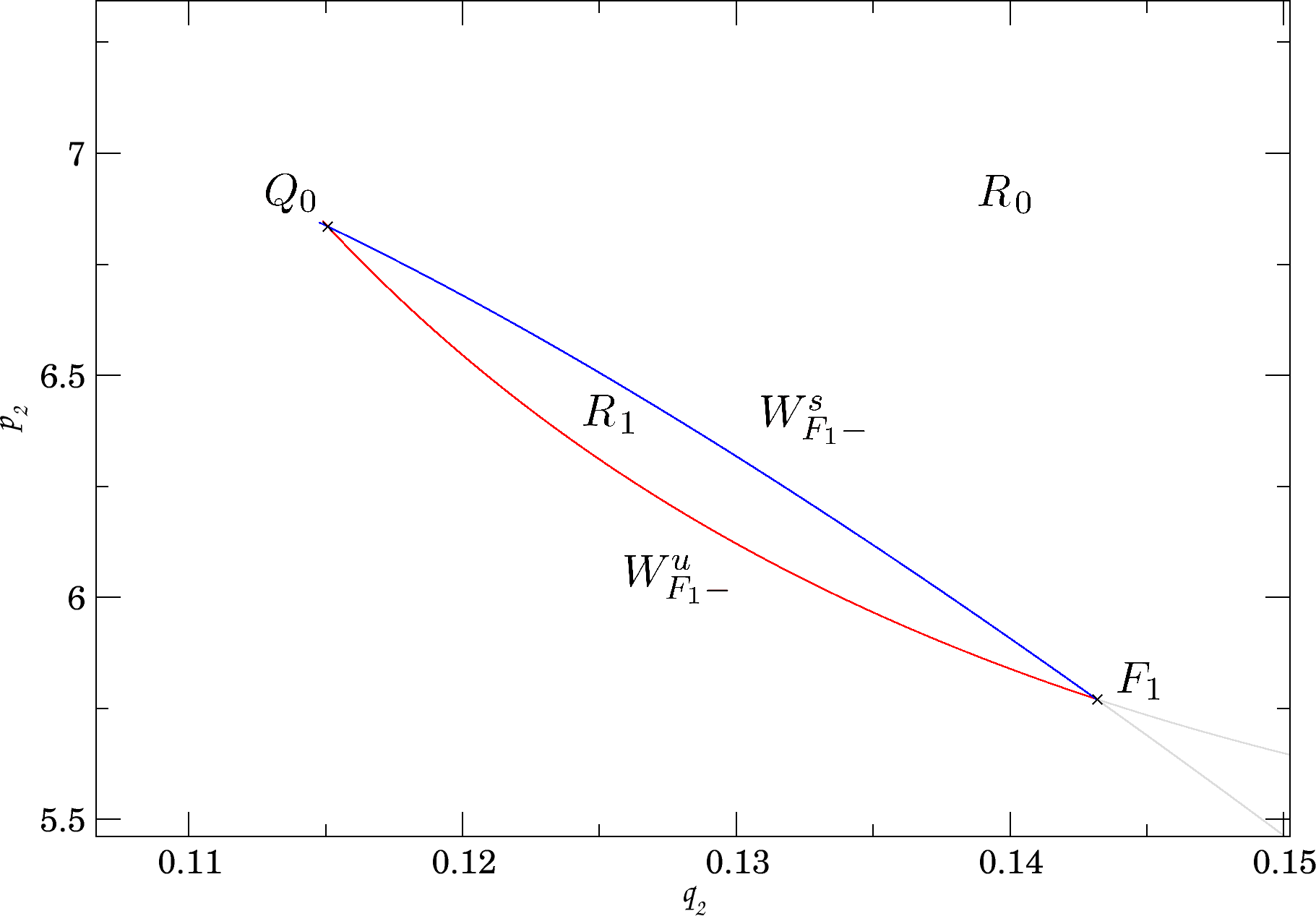}
  \includegraphics[width=0.49\textwidth]{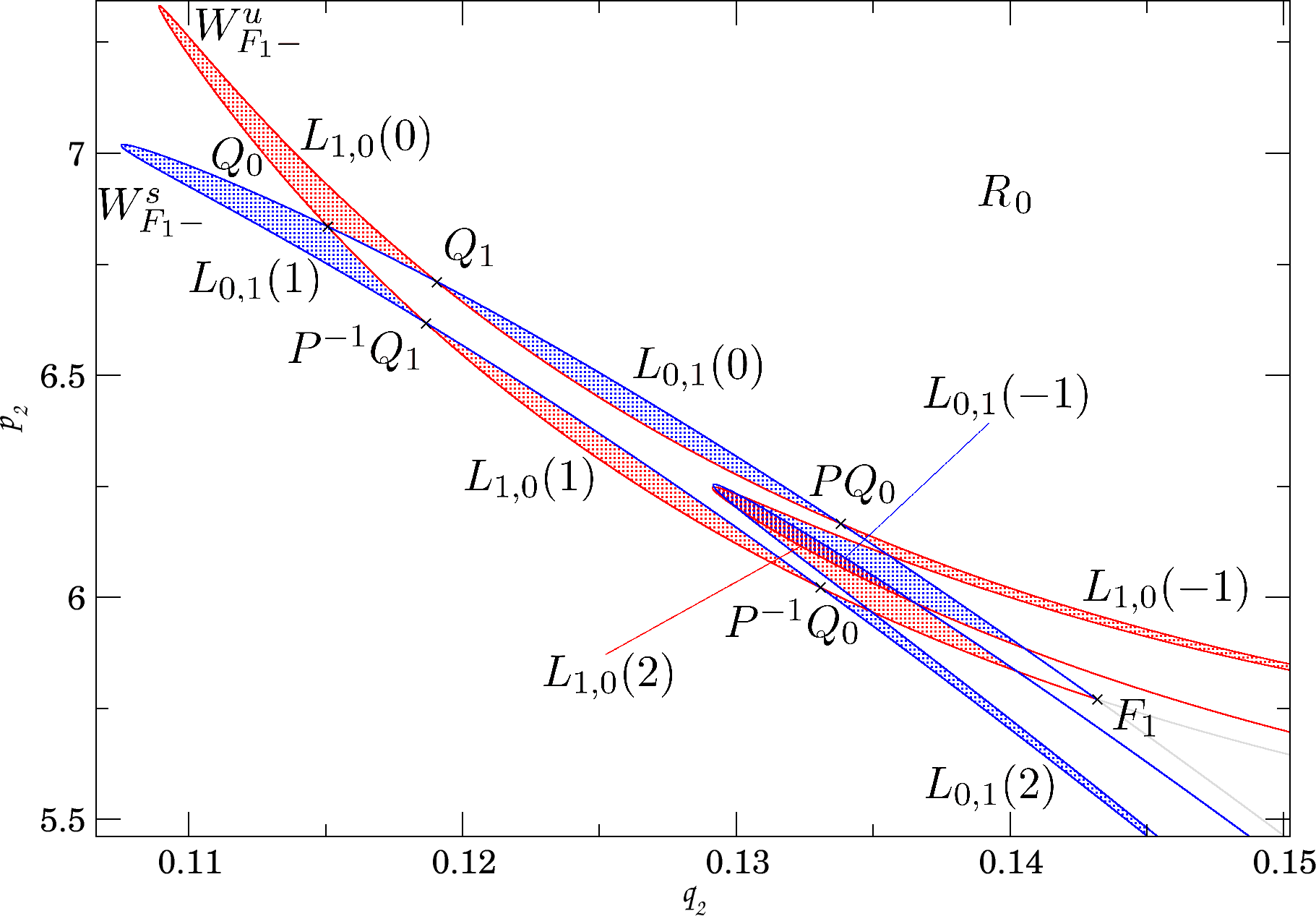}
  \caption{Definition of a region and highlighted lobes in the $F_1$ tangle at $0.02206$.}\label{fig:2206}
  \end{figure}
  
  There is only one pip between $Q_0$ and $PQ_0$, denote it $Q_1$. In general the number of pips between $Q_0$ and $PQ_0$ is always odd (see \cite{Rom-Kedar90}.
  
  We define lobes using $Q_0$, $Q_1$ and all of their (pre-)images. The way lobes guide trajectories in and out of regions can be seen on the lobe bounded by $S[Q_1,Q_0]$ and $U[Q_0,Q_1]$. The lobe is located in $R_0$, but its preimage bounded by $S[P^{-1}Q_1,P^{-1}Q_0]$ and $U[P^{-1}Q_0,P^{-1}Q_1]$ lies in $R_1$. This area escapes from $R_1$ to $R_0$ after $0$ iterations of the map $P$, we denote the lobe by $L_{1,0}(0)$. Analogously, by $L_{0,1}(0)$ we denote the lobe that is captured in $R_1$ from $R_0$ after $0$ iterations and is bounded by $S[PQ_0,Q_1]$ and $U[Q_1,PQ_0]$. We refer to images and preimages of $L_{1,0}(0)$ and $L_{0,1}(0)$ as \emph{escape lobes} and \emph{capture lobes} respectively.
  Note that due to the no-return property of the interaction region, escape and capture lobes cannot intersect beyond DS$_1$.
   
  Denote the lobe that leaves $R_i$ for $R_j$, $i\neq j$, immediately after $n$ iterations of the map $P$ by
  $$L_{i,j}(n).$$
  In this notation we have for all $k,n\in\mathbb{Z}$ the relation 
  \begin{equation}\label{eq:map lobes}
   P^kL_{i,j}(n)=L_{i,j}(n-k).
  \end{equation}
    
  Transition between $R_0$ and $R_1$ is closely connected to $Q_0$ and the transition from $L_{i,j}(1)$ to $L_{i,j}(0)$. All other lobes are confined by the barrier consisting of invariant manifolds to their respective regions. Near $Q_0$, however, the barrier has a gap through which trajectories can pass. MacKay, Meiss and Percival \cite{MacKay84} described this mechanism by saying that it ``acts like a revolving door or turnstile.'' The term \emph{turnstile} was born and lives on, see \cite{Meiss15}.
  
  While $W^s_{F_1-}$ contracts exponentially near the $F_1$, $W^u_{F_1-}$ stretches out. It is easy to see that $S[F_1,Q_0]$ is a rigid barrier - nearly linear and guiding all trajectories in its vicinity. $W^u_{F_1-}$ is a more flexible barrier in forward time - the manifold itself twists and stretches, alternately lying in $R_0$ and $R_1$. The fluid shape of $W^u_{F_1-}$ is the result of complicated dynamics and the influence of $S[F_1,Q_0]$. Stable manifolds behave similarly in backward time and the transition from rigid to flexible results in the turnstile mechanism.
  
  The same is true for heteroclinic tangles. These imperfect barriers are responsible for nonreactive trajectories with high translational energy and reactive trajectories with surprisingly low translational energy. Due to this strangely selective mechanism we speak of a \emph{partial barrier}.
  
  Choosing any other pip than $Q_0$ for the definition of the regions merely affects the time in which lobes escape. Compared to definitions based on $Q_0$, if we chose $PQ_0$ instead, escape/capture of lobes would be delayed by $P$, if we chose $Q_1$, only escape lobes would be affected. This has implications for notation, not for dynamics or its understanding.
  
  \subsection{Properties of lobes}
  Here we state some of the basic properties of lobes that will be relevant in the following sections. The following statements assume that we study transport between two regions that are separated by a homoclinic tangle or a heteroclinic tangle and involves no other invariant manifolds. This provides useful insight into the complex dynamics of homoclinic and heteroclinic tangles.
  
  If the intersection $L_{i,j}(0)\cap L_{j,i}(0)$ is non-empty, it does not leave the respective region and is not subject to transport. In this case we may redefine lobes to be $$\tilde{L}_{i,j}(k):=L_{i,j}(k)\setminus \left(L_{i,j}(k)\cap L_{j,i}(k)\right),$$
  where $\tilde{L}_{i,j}(k)\cap\tilde{L}_{j,i}(k)=\emptyset$. This justifies the following assumption.
  \begin{assumption}\label{assum:lobe}
  We assume that the lobes $L_{i,j}(0)$ and $L_{j,i}(0)$ are disjoint.
  \end{assumption} 
  Equivalently we could assume $L_{i,j}(1)\subset R_i$ and $L_{i,j}(0)\subset R_j$. In case of transport between several regions, we can only make statements based on the two regions that are separated by manifolds of the given tangle.

  Each homoclinic and heteroclinic tangle involves a region bounded by segments of invariant manifolds, such as $R_1$ in Sec. \ref{subsec:partial barrier}. Since $P$ is symplectic, almost all trajectories that enter the bounded region must eventually leave it. This can be formulated as
  \begin{lemma}\label{lemma:partition}
  Let at least one of $R_i$ and $R_j$ be bounded. Then $L_{i,j}(0)$ can be partitioned, except for a set of measure zero $O$, as
  $$L_{i,j}(0)\setminus O=\bigcup_{n\in\Zbb} L_{i,j}(0)\cap L_{j,i}(n).$$
  \end{lemma}
  
  \begin{remark}
  The region $R_j$ has the no-return property iff escape lobes ($L_{j,i}$) are disjoint, or equivalently iff capture lobes are disjoint. Automatically then for all $n>0$
  $$L_{i,j}(0)\cap L_{j,i}(-n) = \emptyset.$$
  \end{remark}
  
  Some of the intersections in Lemma \ref{lemma:partition} $L_{i,j}(0)\cap L_{j,i}(n)$ for $n>0$ are empty sets. We are going to show that finitely many are empty at most.
  
  \begin{lemma}\label{lemma:lobes intersect}
   For all $n_0>0$
   $$L_{i,j}(0)\cap L_{j,i}(n_0)\neq \emptyset \Rightarrow L_{i,j}(0)\cap L_{j,i}(n_0+1)\neq \emptyset.$$
  \end{lemma}
  
  Using Fig. \ref{fig:2206} as an example,
  $$L_{0,1}(-1)\cap L_{1,0}(2)\neq \emptyset \Rightarrow L_{0,1}(-1)\cap L_{1,0}(3)\neq \emptyset,$$
  because $L_{0,1}(0)\cap L_{1,0}(3)\neq \emptyset$ and $W^s_{F_1-}$ can only reach $L_{0,1}(0)$ by passing through $L_{0,1}(-1)$.
  
  \begin{proof}
   Without loss of generality assume $R_j$ is bounded and fix $n_0>0$. If
   $$L_{i,j}(0)\cap L_{j,i}(n_0)\neq \emptyset,$$
   then its image under $P$
   $$L_{i,j}(-1)\cap L_{j,i}(n_0-1)\neq \emptyset.$$
   We are going to argue that the only way for $L_{i,j}(-1)$ to reach $L_{j,i}(n_0-1)$ is by intersecting $L_{j,i}(n_0)$.
   
   Denote $Q_1$ and $Q_2$ the pips that define $L_{i,j}(0)$ and $P^{-n_0}Q_0$ and $P^{-n_0}Q_1$ the pips that define $L_{j,i}(n_0)$.
   Let $\widetilde{Q}\in U[Q_1,Q_2]\cap S[P^{-n_0}Q_1,P^{-n_0}Q_0]$.
   
   $L_{i,j}(-1)$ lies inside $R_j$ (possibly partially in $R_i$ via another escape lobe) and so does $U[PQ_1,PQ_2]$, the part of $\partial L_{i,j}(-1)$ that does not coincide with $\partial R_j$. Note that as all pips, $PQ_1,PQ_2\in\partial R_j$. The intersection point $\widetilde{Q}$ lies in the interior of the region bounded by $U[P^{-n_0}Q_1,\widetilde{Q}]$ and $S[P^{-n_0}Q_1,\widetilde{Q}]$, while $PQ_1$ is located outside. Because a invariant manifold cannot reintersect itself, $U[PQ_1,P\widetilde{Q}]$ has to cross $S[P^{-n_0}Q_1,\widetilde{Q}]$, which is part of $\partial L_{j,i}(-n_0)$. Therefore
   $$L_{i,j}(-1)\cap L_{j,i}(n_0)\neq \emptyset,$$
   and when mapped backward,
   $$L_{i,j}(0)\cap L_{j,i}(n_0+1)\neq \emptyset.$$
  \end{proof}
  
  Note for $n_0<0$, time reversal yields using a similar argument
  $$L_{i,j}(0)\cap L_{j,i}(n_0)\neq \emptyset \Rightarrow L_{i,j}(0)\cap L_{j,i}(n_0-1)\neq \emptyset.$$

  Following Lemma \ref{lemma:partition} and Lemma \ref{lemma:lobes intersect}, for $k$ large enough $L_{i,j}(k)$ lies simultaneously in both regions forming a complicated structure. Since pips are mapped exclusively on $\partial R_j$, they aid identification of parts of lobes.
    
  Due to (\ref{eq:map lobes}), for $n$ small we may study lobe intersections of the form
  $$L_{0,1}(k)\cap L_{1,0}(k+n),$$
  that tend to be heavily distorted by the flow simply by mapping them forward or backward to less distorted intersections. However this does not work for
  $$L_{0,1}(-k)\cap L_{1,0}(k),$$
  for large $k$. On the other hand, we can expect the area of this intersection to shrink considerably with $k$, so their quantitative impact is limited.
  
  We remark that while almost the entire area of a capture lobe must escape at some point, this does not apply to entire regions. Regions may contain stable fixed points surrounded by KAM curves (sections of KAM tori) that never escape.
  
  The picture of a heteroclinic tangle as a structure consisting of only two manifolds is oversimplified. In general heteroclinic tangles in a Hamiltonian system with $2$ degrees of freedom can be expected to involve four branches of invariant manifolds. It takes four segments and two pips to define a region and consequently there will always be two turnstiles. The oversimplification is justified for tangles where the two turnstiles are made up of mutually disjoint lobes. Tangles with two intersecting turnstiles admit transport between non-neighbouring regions and we approach them differently.
  
  \subsection{Content of a lobe}
  In this section we use show how lobes guide trajectories in their interior.
  
  Denote by $\mu$ the measure on $\Sigma_0$, that is proportional to $\omega_2\bigr\vert_{\Sigma_0}$ \eqref{eq:omega2sigma}. Under area preservation we understand that for any set $A$ and for all $k\in\mathbb{Z}$ $$\mu\left(A\right)=\mu\left(P^kA\right).$$
  As a direct consequence of area preservation of a region we have for all $k,n\in\mathbb{Z}$
  $$\mu\left(L_{i,j}(n)\right)=\mu\left(L_{j,i}(k)\right).$$
  
  \begin{assumption}\label{assum:nonzero}
  Throughout this work we assume that $\mu(L_{i,j}(0))\neq0$.
  \end{assumption}
    
  Combining Assumptions \ref{assum:lobe} and \ref{assum:nonzero} implies that $L_{i,j}(0),L_{j,i}(1)\subset R_j$ and if 
  $$2\mu\left(L_{i,j}(0)\right)>\mu\left(R_j\right),$$
  then necessarily $L_{i,j}(0)\cap L_{j,i}(1)\neq \emptyset$.
  
  All other lobes may partially lie in both $R_i$ and $R_j$, depending on the intersections of escape and capture lobes.
  \begin{definition}
  Assume $R_j$ is bounded. The \emph{shortest residence time} in a tangle is a number $k_{srt}\in\mathbb{N}$, such that
  $$L_{i,j}(0)\cap L_{j,i}(k) = \emptyset,$$
  for $0<k<k_{srt}$ and
  $$L_{i,j}(0)\cap L_{j,i}(k_{srt})\neq \emptyset.$$
  \end{definition}
  
  \begin{remark}\label{rem:intersections}
  The first lobe to lie partially outside $R_j$ is $L_{i,j}(-k_{srt})$, because it intersects $L_{j,i}(0)\subset R_i$. The lobes $L_{i,j}(-k)$ and $L_{j,i}(k)$ are entirely contained in $R_j$ for $0\leq k<k_{srt}$.
  \end{remark}
  
  Note that in a homoclinic tangle, since $L_{i,j}(-k)$ for $0\leq k<k_{srt}$ must be mutually disjoint and all contained in $R_j$, necessarily
  $$\mu\left(R_j\right)>k_{srt}\mu\left(L_{i,j}(0)\right).$$

  Once $L_{i,j}(-k_{srt})$ where $k_{srt}>0$ lies partially in $R_i$ by Lemma \ref{lemma:lobes intersect}
  $$L_{i,j}(-k_{srt})\cap L_{j,i}(n)\neq \emptyset,$$
  for all $n>0$ and therefore $L_{i,j}(-k)$ intersects $L_{j,i}(0)\subset R_i$ for all $k>k_{srt}$. Due to reentries and Assumption \ref{assum:lobe}, the statement is not true for $L_{i,j}(k)$ with $k>0$, but an analogue holds in reverse time.

  Reentries are possible in tangles where escape (and capture) lobes are not mutually disjoint, hence the following Lemma.
  \begin{lemma}
  Let $k_1<k_3$ be such that $L_{i,j}(k_1)\cap L_{i,j}(k_3) \neq\emptyset$ with $i=0,1$ and $j=1-i$. Then
  \begin{equation*}
  L_{i,j}(k_1)\cap L_{i,j}(k_3) = \bigcup\limits_{k_2=k_1+1}^{k_3-1} L_{i,j}(k_1)\cap L_{j,i}(k_2)\cap L_{i,j}(k_3). 
  \end{equation*} 
  \end{lemma}
  
  \begin{proof}
  Let $p\in L_{i,j}(k_1)\cap L_{i,j}(k_3)$, $P^{k_1}p\in R_j$ and $P^{k_3-1}p\in R_i$ follow from Assumption \ref{assum:lobe}. Necessarily there exists $k_2$, such that $k_1<k_2<k_3$ and $p\in L_{j,i}(k_2)$. Since $k_2$ may be different for every $p$, the union over $k_2$ follows.
  \end{proof}
  
  The argument can be easily generalised for tangles that govern transport between multiple regions. One only needs to observe that $p$ can return to $R_i$ from any region.

  In the $F_1$ tangle at $0.02215$, reentries can be deduced from the intersection $L_{0,1}(1)\cap L_{1,0}(0)$ that lies completely in $R_0$. See Figure \ref{fig:tangle intersection} for comparison of a tangle at $0.02215$ with reentries and at $0.02210$ without. Note that both tangles have $k_{srt}=1$.
  \begin{figure}
  \centering
  \includegraphics{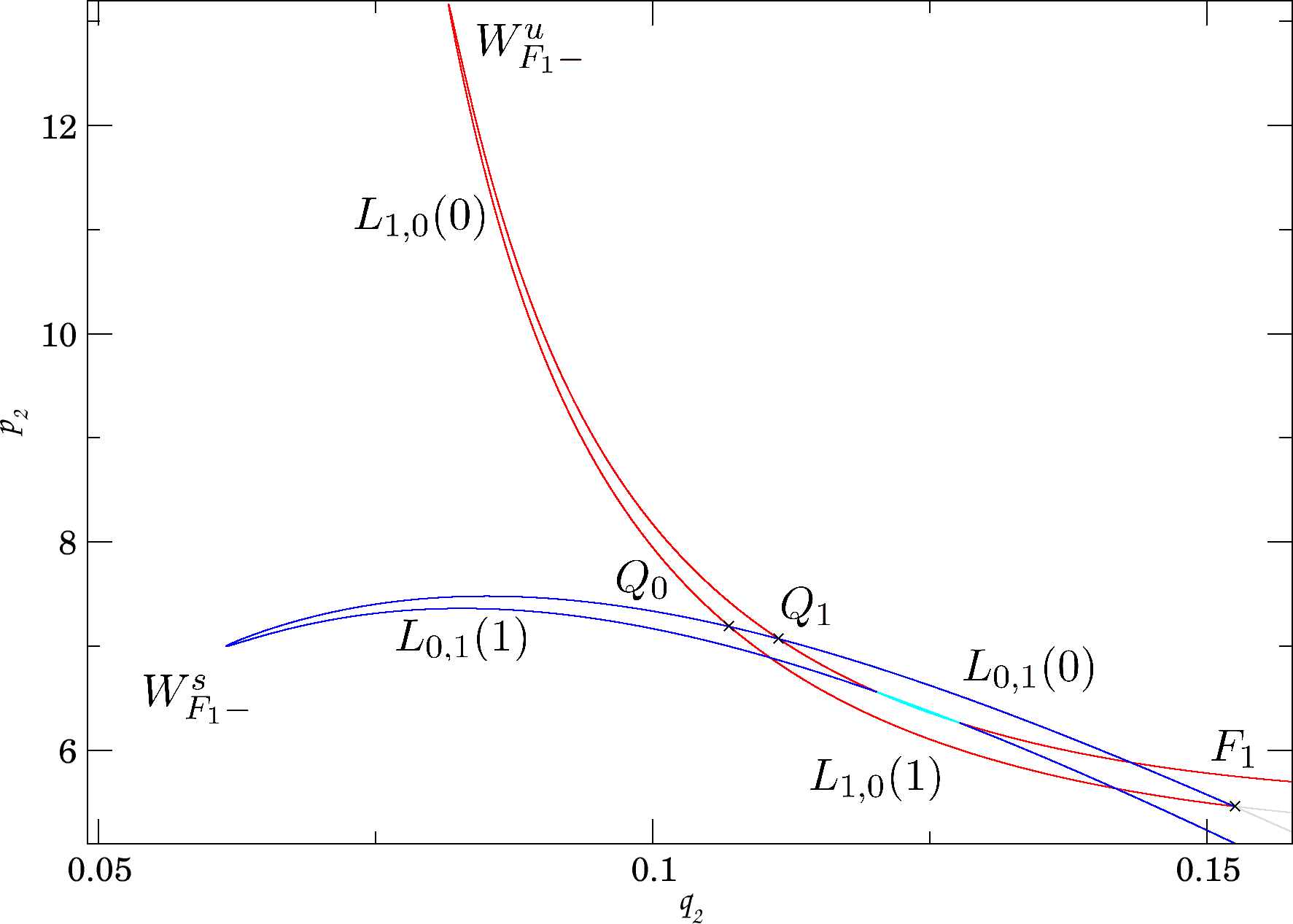}\\
  \includegraphics{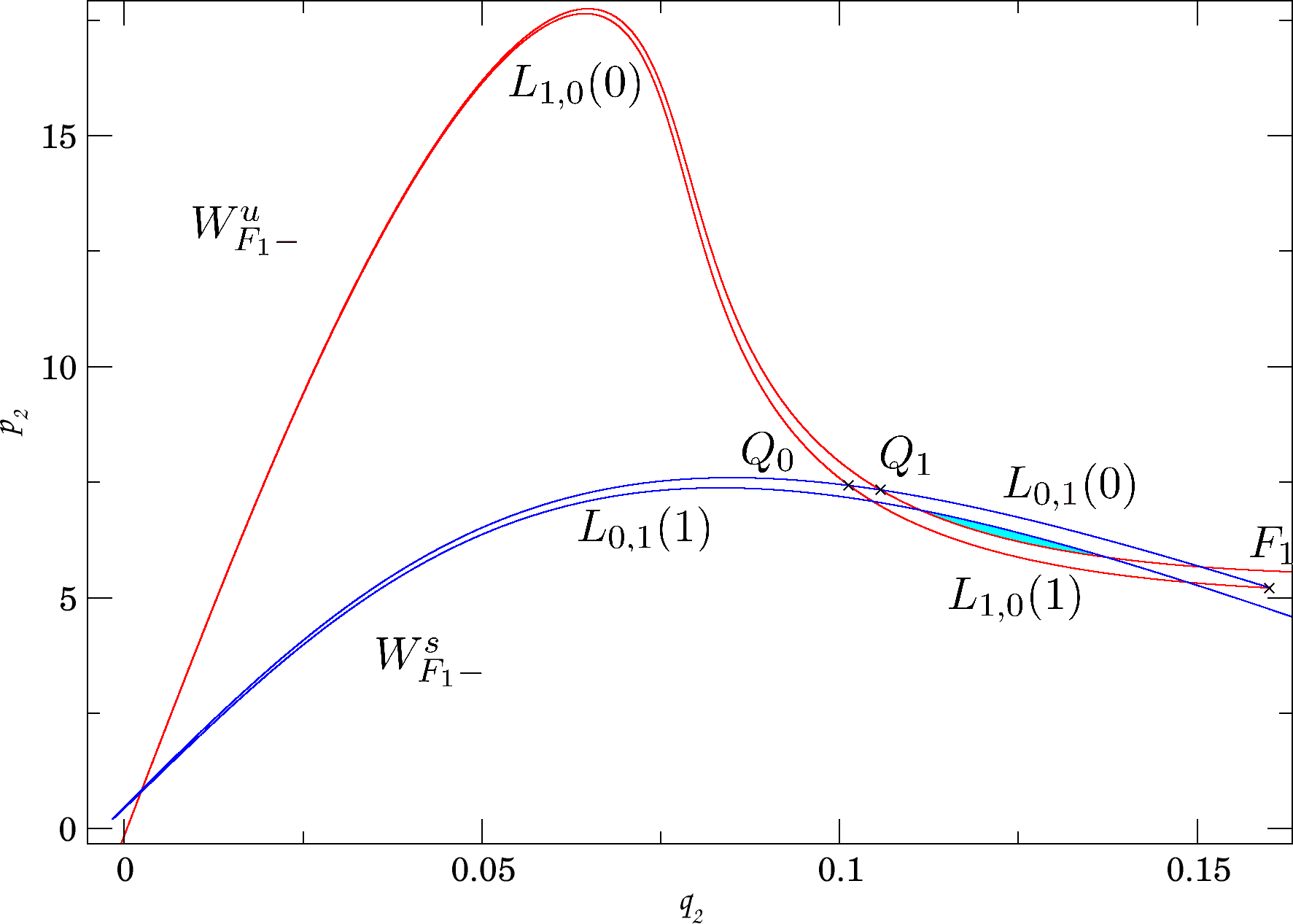}
  \caption{The $F_1$ tangle at $0.02210$ (above) and at $0.02215$ (below). Both homoclinic tangles have $k_{srt}=1$, that can be seen by $L_{0,1}(0)\cap L_{1,0}(1)\neq\emptyset$ shown in cyan. At $0.02215$ the tangle admits reentries.}\label{fig:tangle intersection}
  \end{figure}
  
  Instantaneous transport between regions is described by the turnstile mechanism. Transport on a larger time scale can be studied using a measureless and weightless entity (species, passive scalars or contaminants \cite{Sreenivasan91}, \cite{Sreenivasan97}) that is initially contained and uniformly distributed in a region, as done in \cite{Rom-Kedar90}. Its role is to retain information about the initial state without influencing dynamics indicate escapes and reentries via lobes.
  
  The challenge of studying lobes over large timescales is to determine which regions a lobe lies in and correctly identifying the interior of a lobe. For this we propose a partitioning of heteroclinic tangles into regions of no return outside of which the evolution of lobes is of no interest.
  
 \section{Influence of tangles on the reaction rate}\label{sec:tangles reaction}
 In this section we discuss the evolution of homoclinic and heteroclinic tangles in the entire energy interval $0<E\leq0.03000$ and their influence on dynamics in the interaction region. The dynamics for higher energies is due to the lack of bifurcations analogous. The study of invariant manifolds employs lobe dynamics and a new partitioning based on dynamical properties. An in-depth review of invariant manifolds in a chemical system and structural changes in tangles caused by bifurcations has to our knowledge not been done before.
 
 \subsection{Energy interval where TST is exact}
 TST is exact in the presence of a single TS (due to \cite{PechukasPollak79TST}) and remains exact in case of multiple TSs provided their invariant manifolds do not intersect (due to \cite{Davis87}). Therefore results of TST and Monte Carlo agree on the interval from $0$ to $0.02215$. $W_{F_0}$ separate reactive and nonreactive trajectories, see Sec. \ref{subsec:barrier}, while the $F_1$ tangle captures nonreactive trajectories only.
  
 \begin{figure}
 \centering
     \includegraphics[width=\textwidth]{2206_all.png}\\
     \includegraphics[width=\textwidth]{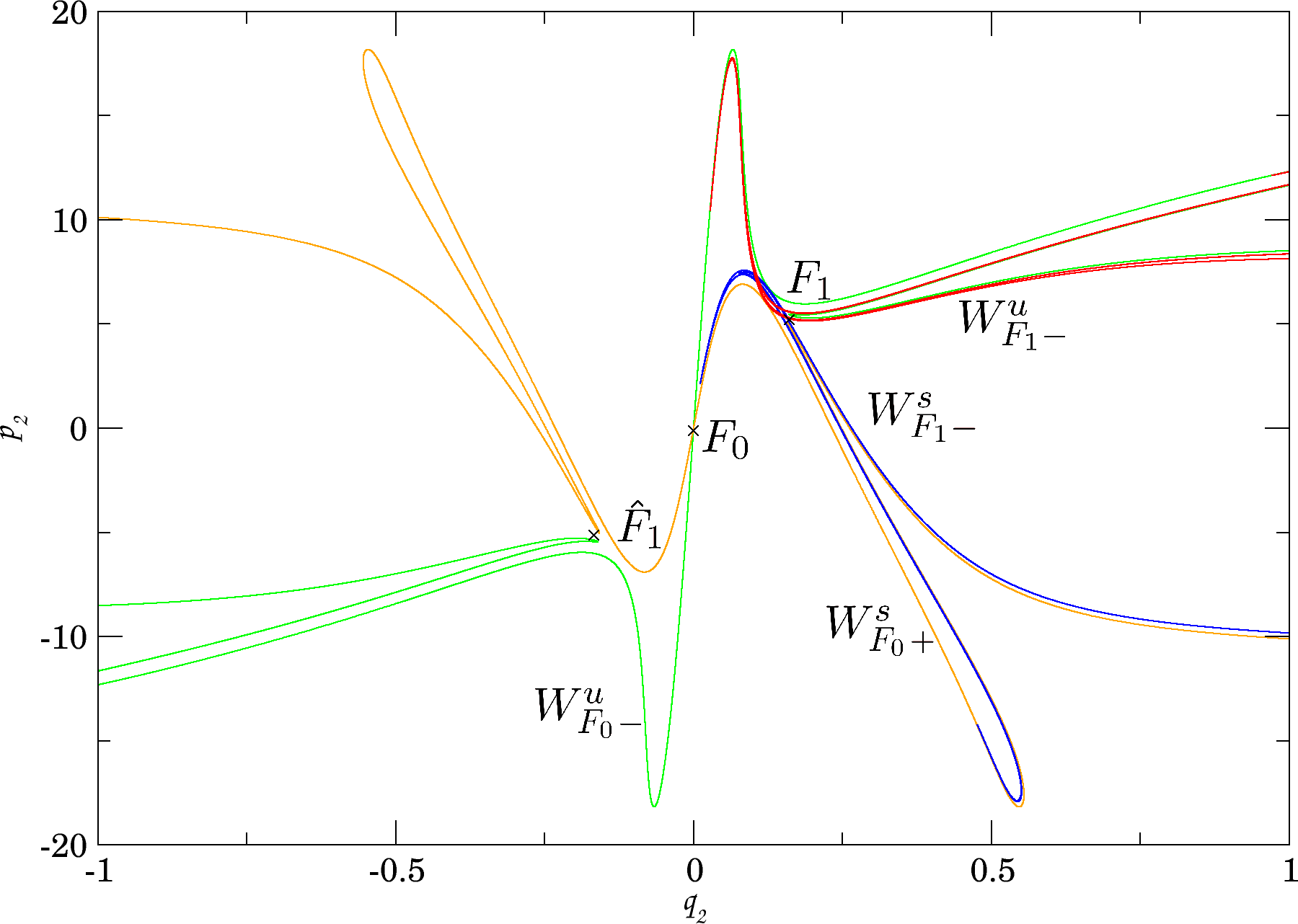}
         \caption{Invariant manifolds at $0.02206$ and $0.02214$.}\label{fig:2206 and 2214}
 \end{figure}
 
 \begin{figure}
 \centering
     \includegraphics[width=\textwidth]{2206_detail.png}\\
     \includegraphics[width=\textwidth]{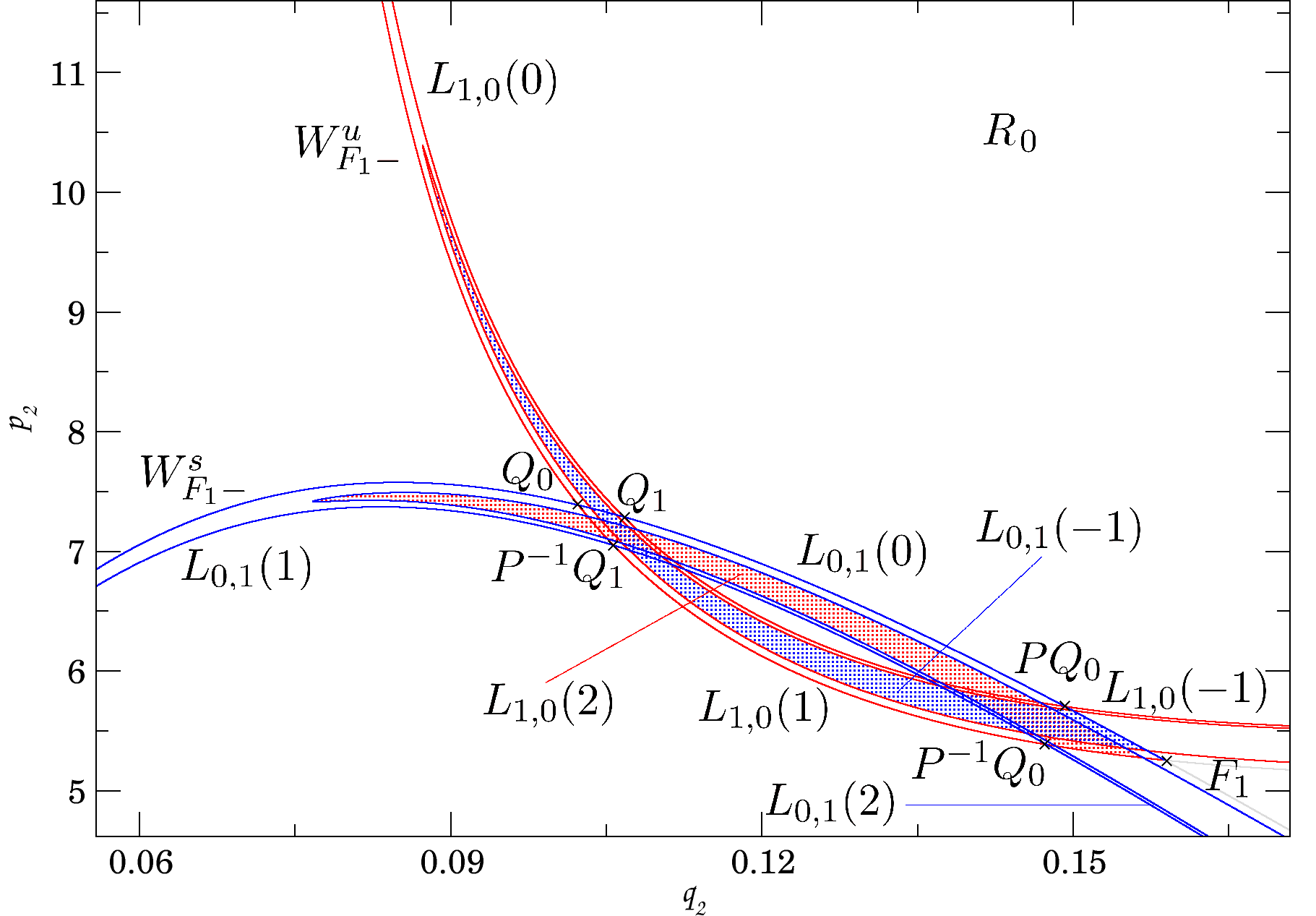}
         \caption{Lobe structure of the $F_1$ tangle at $0.02206$ and $0.02214$.}\label{fig:homoclinic 2206 and 2214}
 \end{figure}
 
 Some properties of the $F_1$ tangle are carried over to higher energies, such as shape of lobes or $k_{srt}$. Fig. \ref{fig:2206 and 2214} shows $W_{F_0}$ and $W_{F_1}$ approaching prior to the intersection at $0.02215$ and the failure of TST.
 
 Each change of structure seems to coincide with a bifurcation of a periodic orbit. The decrease $k_{srt}$ from $3$ to $1$ over the energy interval, shown in Fig. \ref{fig:homoclinic 2206 and 2214}, coincides with the period doubling of $F_2$ at $0.02208$ and the period doubling of $F_{21}$ before $0.02209$. From a quantitative perspective, the tangle and its lobes grow larger in area.

 \subsection{Point where TST fails}
 At $0.02215$, $W_{F_0}$ and $W_{F_1}$ interact through heteroclinic intersections. Instead of minor changes in the overall topology of the invariant manifolds, we come across something that is better described as a chain reaction.
 
 Firstly, we observe that $W_{F_0+}$ and $W_{F_1-}$ intersect forming a heteroclinic tangle, see Fig. \ref{fig:2230 R0}. Consequently, TST starts to fail (see \cite{Davis87}) and the Monte Carlo reaction rate is lower than TST. $W^s_{F_0}$ and $W^u_{F_0}$ form a partial barrier and this enables the $F_1$ tangle to capture reactive trajectories. We also find heteroclinic intersections of $W_{F_1-}$ and $W_{\widehat{F}_1+}$ as shown in Fig. \ref{fig:2230 F1F1}, as well as $W_{F_1-}$ and $W_{F_0-}$. Recall that statements for $F_1$ also hold for $\widehat{F}_1$.
 
 \begin{figure}
    \centering
     \includegraphics{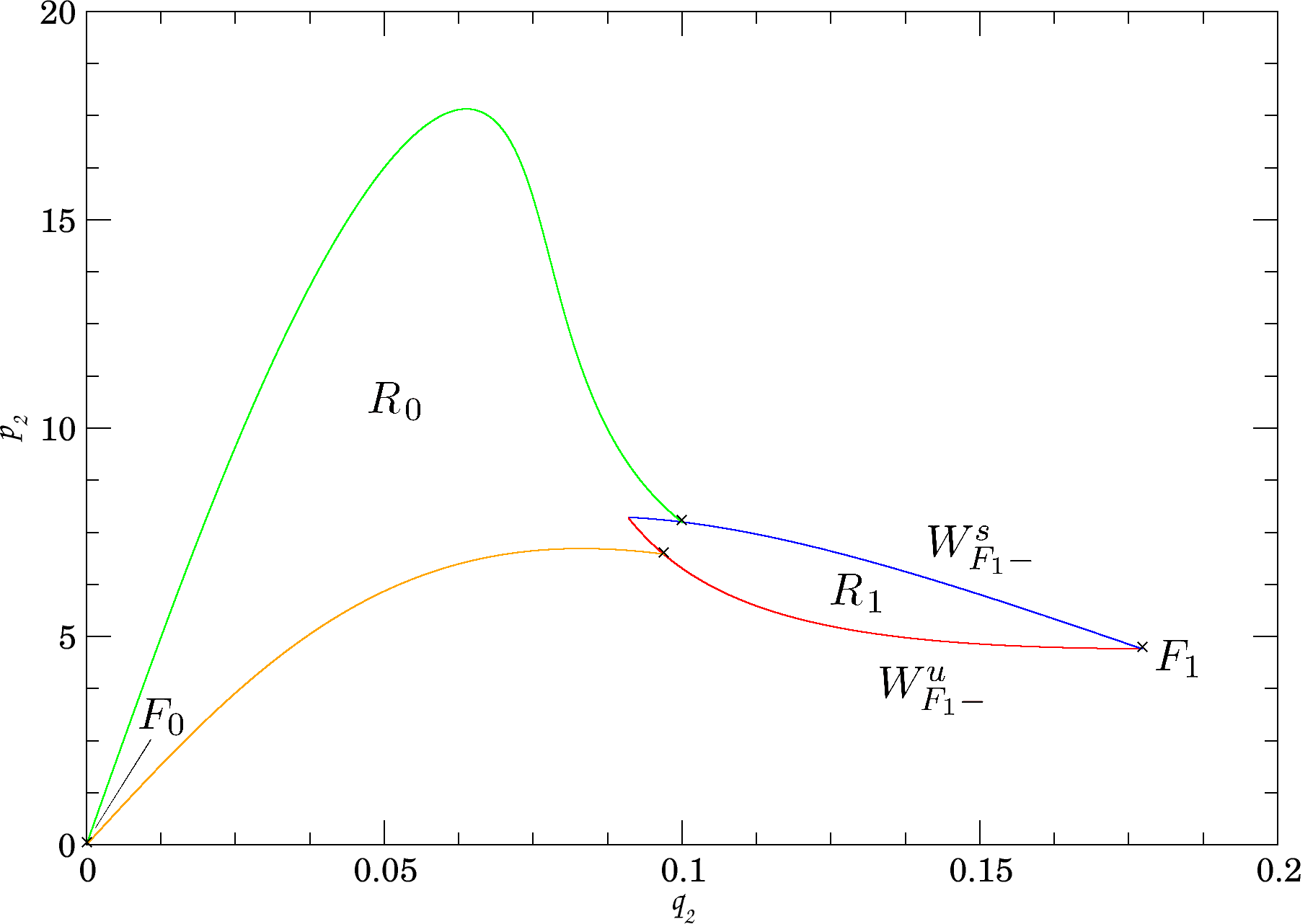}
         \caption{The regions $R_0$ and $R_1$ at $0.02230$.}
    \label{fig:2230 R0}
 \end{figure}
  
 Choose two pips in the $F_0$-$F_1$ tangle, so that the region bounded by $W_{F_0+}$ and $W_{F_1-}$ denoted $R_0$ satisfies $R_1\subset R_0$ (Fig. \ref{fig:2230 R0}) and define $\widehat{R}_0$ using symmetry. 
  
 As $L_{0,1}(0)$ and $L_{1,0}(1)$ in the $F_1$ tangle contain heteroclinic points that converge towards $F_0$ (forward or backward time), they necessarily intersect in $R_0$ (see Fig. \ref{fig:tangle intersection}). By definition, $L_{0,1}(0)\cap L_{1,0}(1)$ contains trajectories that reenter $R_1$ after they have escaped and consequently $R_1$ (and $\widehat{R}_1$) loses its no-return property.
 In particular, trajectories that periodically reenter $R_1$ may exist and if they do, they will be located in $L_{0,1}(0)\cap L_{1,0}(\tilde{k}) \cap \dots$ for some $\tilde{k}$.
 
 \begin{figure}
    \centering
     \includegraphics[width=\textwidth]{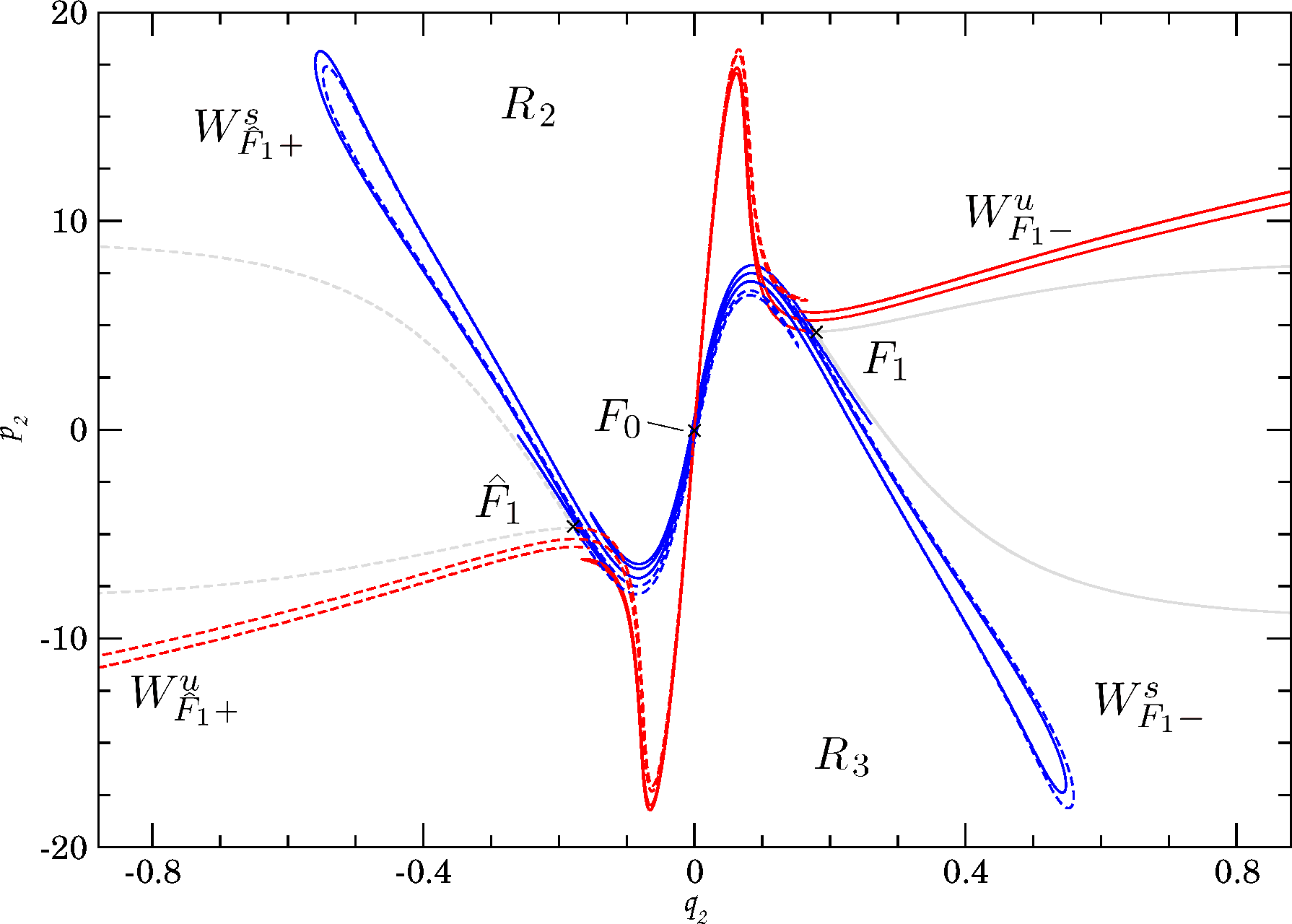}
         \caption{The $F_1$-$\widehat{F}_1$ tangle at $0.02230$, $W_{F_1}$ are shown as solid lines, $W_{\widehat{F}_1}$ as dashed.}
     \label{fig:2230 F1F1}
 \end{figure}
 
 By symmetry $L_{\hat{0},\hat{1}}(0)$ and $L_{\hat{0},\hat{1}}(1)$ also contain heteroclinic points that converge towards $F_0$ and they cannot avoid intersecting $L_{1,0}(1)$ and $L_{0,1}(0)$ respectively. Figure \ref{fig:2230 F1F1} portraits the intersecting invariant manifolds. These intersections guide trajectories that may cross DS$_0$ multiple times and result in an overestimation of the reaction rate by TST. Due to the size of the lobe intersections, the overestimation is small but increases with energy. VTST suffers from recrossings too as it estimates the rate using the DS with lowest flux, but none of the DSs is recrossing-free.
 
 Due to a high $k_{srt}$ and small area of lobes, we avoid details of the $F_1$-$\widehat{F}_1$ tangle until higher energies. We remark that lobes in the $F_1$-$\widehat{F}_1$ tangle do not intersect outside of the bounded region.

 \subsection{Definitions of important regions}
 We have established that TST fails at $0.02215$ due to recrossings. In this section we give a detailed description of homoclinic and heteroclinic tangles at $0.02230$ and explain the transport mechanism in these tangles using lobes. The energy $0.02230$ is representative for the interval between TST failure at $0.02215$ and one of several period doubling bifurcations of $F_{21}$ at $0.02232$. Moreover, lobes at $0.02230$ are sufficiently large to study.
 
 For the sake of simple notation, in what follows $Q_0$, $Q_1$, $Q_2$ and $Q_3$ denote pips that differ from tangle to tangle. To avoid confusion, we always clearly state which tangle is discussed.

 First we discuss the homoclinic tangles of $F_0$, $F_1$ and $\widehat{F}_1$ at $0.02230$. We define regions relevant to these homoclinic tangles shown in Figure \ref{fig:2230 regions} as follows.
 
 Denote $R_0$, the region bounded by $W_{F_0+}$ and $W_{F_1-}$. The $F_0$-$F_1$ tangle is responsible for most of the complicated evolution of reactive trajectories at $0.02230$. The regions above and below the $F_0$-$F_1$ tangle are $R_2$ and $R_3$ respectively.
 
 The region inside the $F_1$ tangle bounded by $W_{F_1-}$ is denoted $R_1$. Further we denote $R_4$ the region bounded by $W_{F_0+}$ that is relevant for the $F_0$ tangle. A near-intersection of $W_{F_0+}$ in $R_1$ suggests that $R_4$ is smaller after the period doubling bifurcation of $F_{21}$ at $0.02232$.
  
 \begin{figure}
  \centering
     \includegraphics[width=\textwidth]{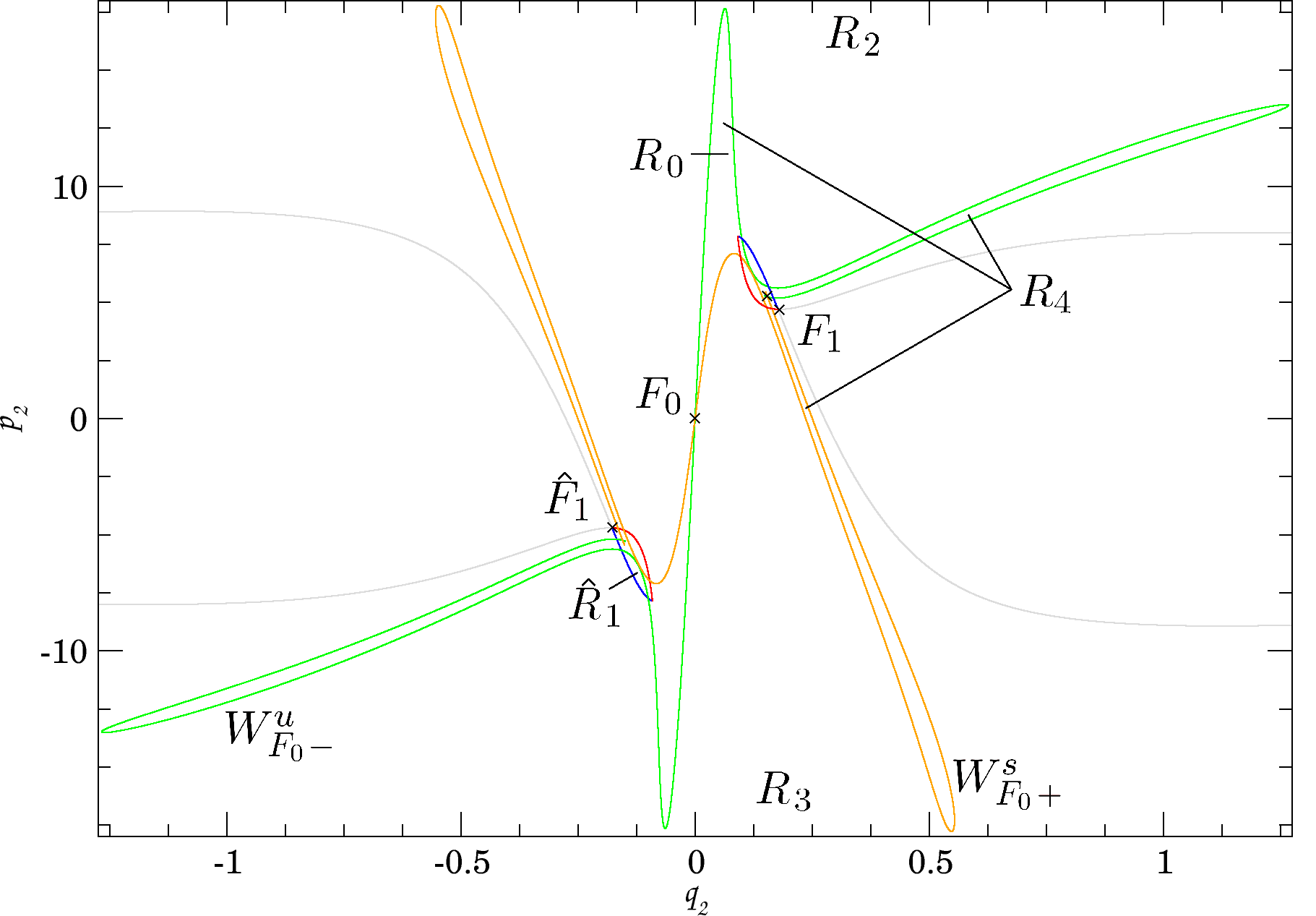}
         \caption{Various region at $0.02230$.}\label{fig:2230 regions}
 \end{figure}
  
 \subsection{Homoclinic tangles}\label{subsec:homoclinic}
 First we concentrate on the $F_0$ tangle at $0.02230$, followed by the $F_1$ tangle, both depicted in Fig. \ref{fig:homoclinic labeled}. In both it is possible to identify a number of lobes that explain the dynamics within.
 
 The $F_0$ tangle govern transport from $R_3$ to $R_4$ and from $R_4$ to $R_2$. The lobes in this tangle consist of two disjoint parts. $L_{3,4}(0)$, for example, is bounded by $S[Q_1,Q_0]\cup U[Q_0,Q_1]$ and $S[Q_3,Q_2]\cup U[Q_2,Q_3]$. Note that $L_{4,2}(1)$ and $L_{3,4}(1)$ intersect near $Q_0$ and recall that $L_{4,2}(1) \cap L_{3,4}(1)$ does not leave $R_4$. $L_{3,4}(0) \cap L_{4,2}(1)$ near $Q_3$ implies $k_{srt}=1$.
 
 By far the largest intersection in the $F_0$ tangle is $L_{3,4}(-1) \cap L_{4,2}(2)$. It comprises most of the white area in $R_4$ occupied by nonreactive trajectories and we can deduce the structure of the intersection from $L_{3,4}(0)$ and $L_{4,2}(1)$ as follows. As an image of $L_{3,4}(0)$, the larger part of $L_{3,4}(-1)$ is bounded by $S[PQ_1,PQ_0]\cup U[PQ_0,PQ_1]$ with pips indicated in Fig. \ref{fig:homoclinic labeled}. This is nearly a third of the entire region $R_4$.
 Similarly the larger part of $L_{4,2}(1)$ is bounded by $S[Q_0,P^{-1}Q_3]\cup U[P^{-1}Q_3,Q_0]$. Its preimage, the larger part of $L_{4,2}(2)$, is bounded by $S[P^{-1}Q_0,P^{-2}Q_3]\cup U[P^{-2}Q_3,P^{-1}Q_0]$. Thanks to pips we are able to deduce that the majority of trajectories in the $F_0$ tangle is due to the intersection of these two lobes.
 
 Note that part of an escape lobe extends to the product side of $F_0$ and contains reactive trajectories. This part of the lobe enters $R_4$ via $L_{3,4}(1)$, most of which is mapped to $L_{3,4}(0)\cap L_{4,2}(2)$ and escapes into $R_2$ via $L_{4,2}(1)$. Using an analogous argument we find that the part of a capture lobe lies on the product side of $F_0$ and carries reactive trajectories that escaped from $R_4$.
 
 \begin{figure}
  \centering
     \includegraphics[width=\textwidth]{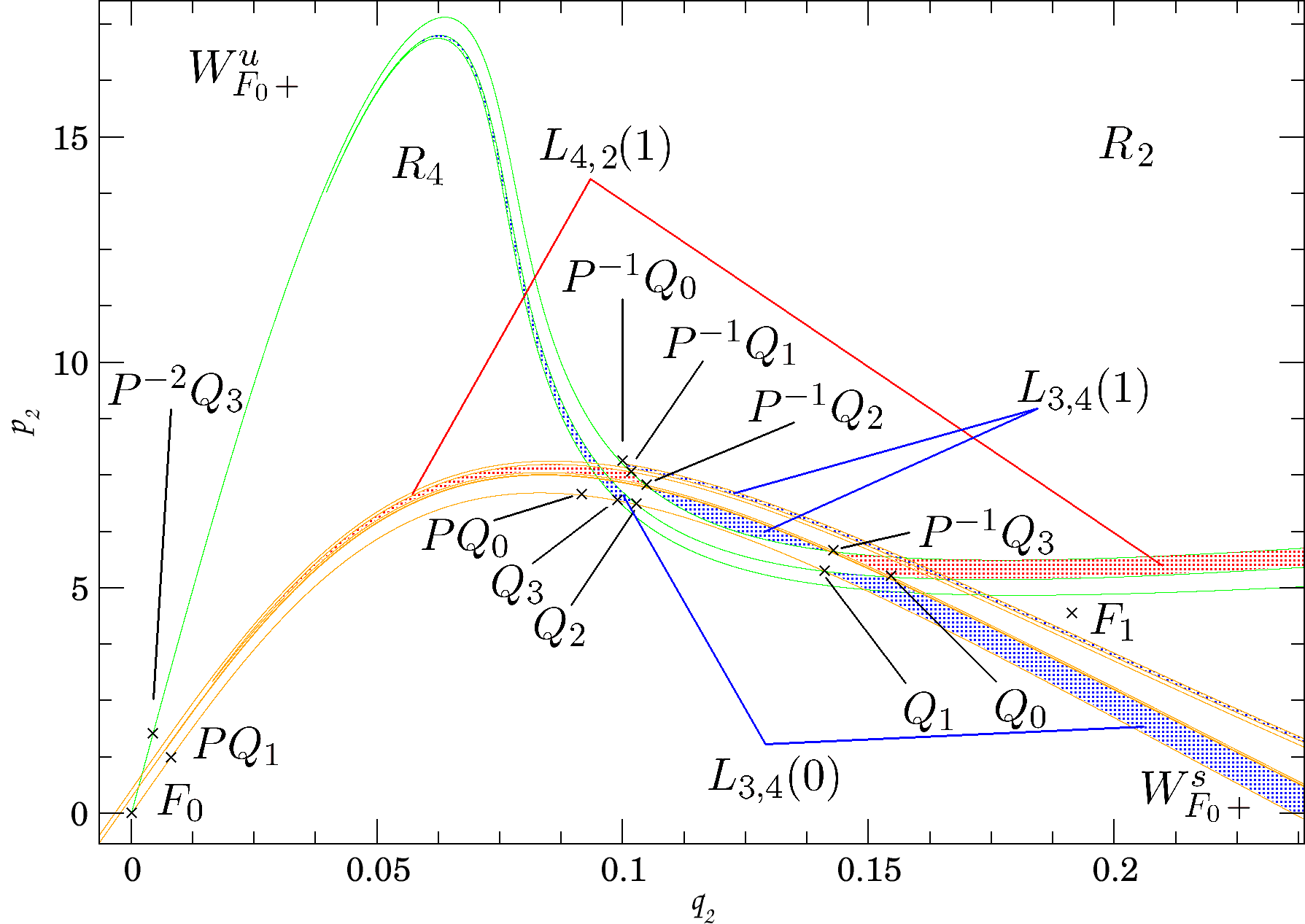}\\
     \includegraphics[width=\textwidth]{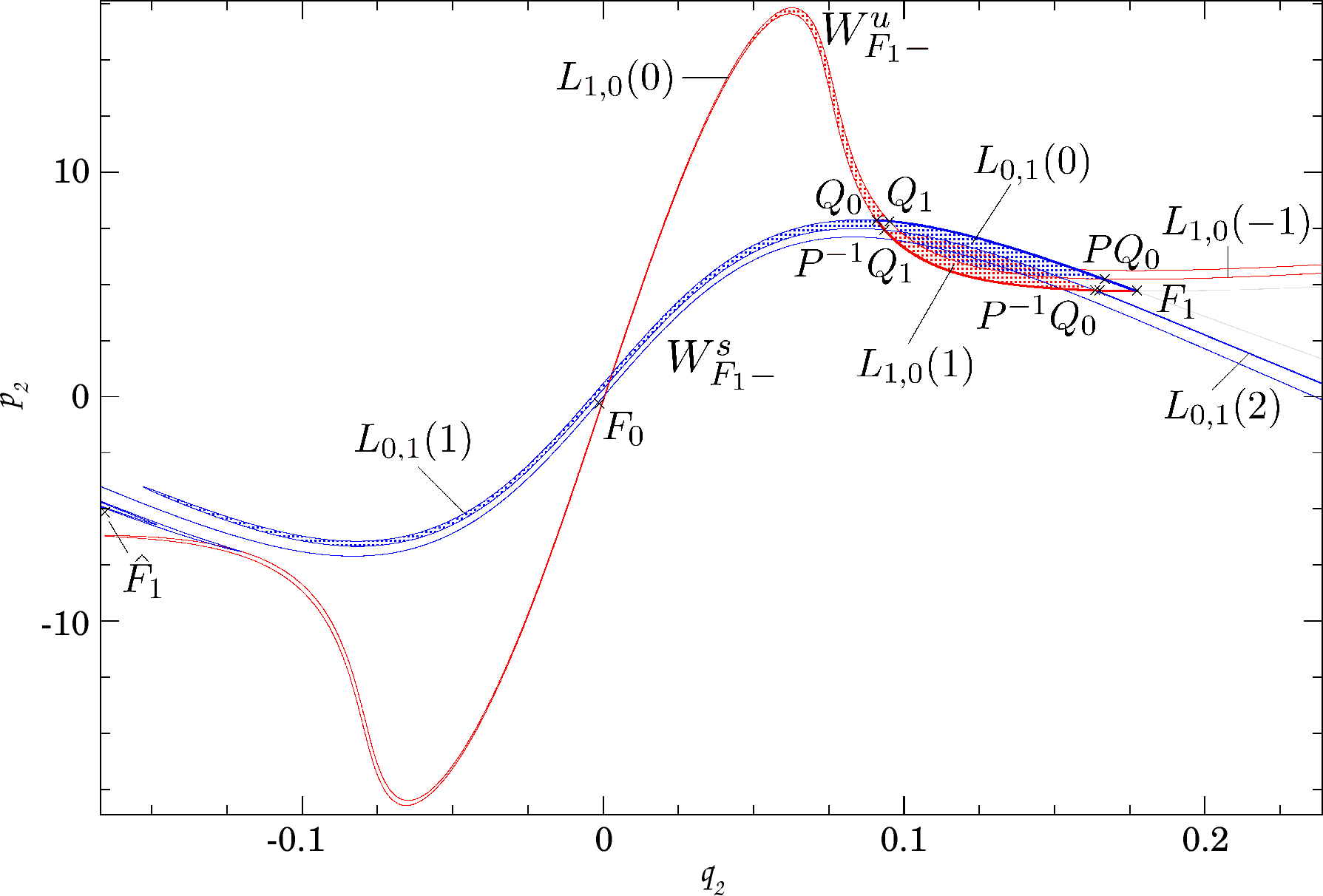}
         \caption{Homoclinic tangles associated with $F_0$ and $F_1$ respectively at $0.02230$.}\label{fig:homoclinic labeled}
 \end{figure}
  
 The $F_1$ tangle has only one pip between $Q_0$ and $PQ_0$ and therefore a simpler structure. $L_{0,1}(0) \cap L_{1,0}(1)$ implies $k_{srt}=1$, therefore trajectories pass through this tangle quickly. Most nonreactive trajectories of the $F_0$ tangle pass inbetween $L_{1,0}(0)$ and $L_{0,1}(1)$ and avoid the $F_1$ tangle. This follows from its adjacency to $Q_0$, which is only mapped along the boundary of $R_1$ always on the reactant side of $F_0$. Similarly we can follow the area between $L_{1,0}(0)$ and $L_{0,1}(2)$ on the product side of $F_0$ using the $\widehat{F}_1$ tangle and symmetry.
 
 The considerable size of lobes on the product side of $F_0$ carries information about nonreactive trajectories. The part of $L_{0,1}(1)$ on the product side of $F_0$ enters $R_1$ via the upper part of $L_{0,1}(0)$, just above the indicated intersection with $L_{1,0}(-1)$. Since this area does not lie in $L_{1,0}(1)$, it is has to be mapped to $L_{0,1}(-1)\setminus L_{1,0}(0)$ that remains in $R_1$ and is defined by the pips $PQ_1$ and $P^2Q_0$ located on $S[F_1,PQ_0]$. Further this area will be mapped in $L_{1,0}(1)\setminus L_{0,1}(0)$ and, unlike the part of $L_{1,0}(1)$ bordering $S[P^{-1}Q_1,P^{-1}Q_0]$, back into products.
 
 In contrast, we can follow the part of $L_{0,1}(2)$ near its boundary $U[P^{-1}Q_0,P^{-2}Q_1]$ in reactants being mapped to $L_{0,1}(1)$ near its boundary $U[Q_0,P^{-1}Q_1]$ and via $L_{0,1}(0)$ near its boundary $U[PQ_0,Q_1]$ into products.
  
 As energy increases, we observe that the nonreactive mechanism of the $F_0$ tangle grows slower than the nonreactive mechanism in the $F_1$ tangle or even shrinks. The later involves crossing the axis $q_2=0$, which on $\Sigma_0$ coincides DS$_0$. Due to symmetry the same happens in the $\widehat{F}_1$ tangle. Therefore the flux across DS$_0$ grows twice as quickly as across DS$_1$. Therefore eventually DS$_1$ becomes the surface of minimal flux.
 
 \subsection{Heteroclinic tangles}\label{subsec:heteroclinic}
 Heteroclinic tangles partially share shapes, lobes and boundaries with homoclinic tangles and their description of transport must agree. Recall heteroclinic tangles have two turnstiles and two sets of escape and capture lobes.
  
 \begin{figure}
  \centering
     \includegraphics[width=\textwidth]{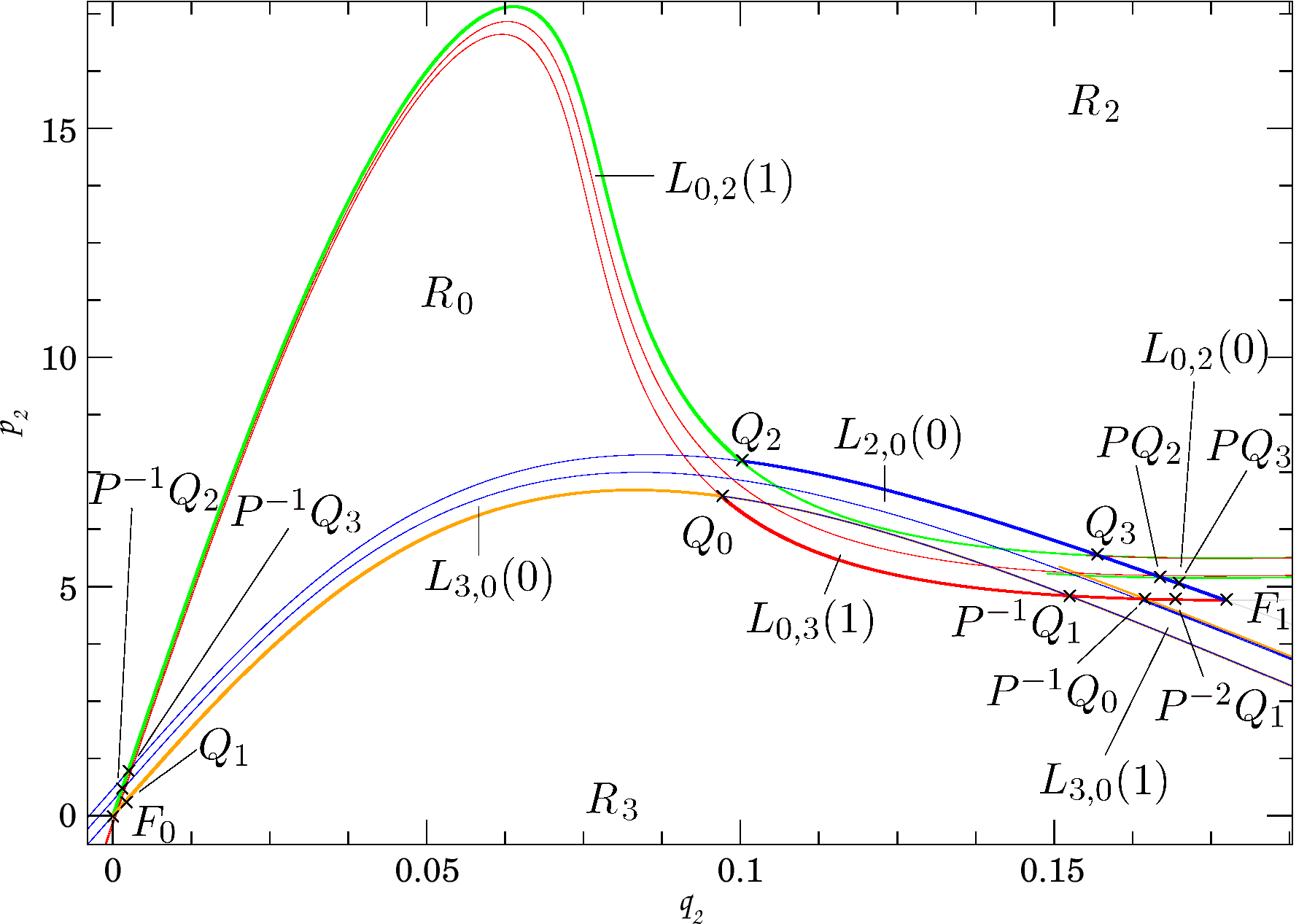}\\
     \includegraphics[width=\textwidth]{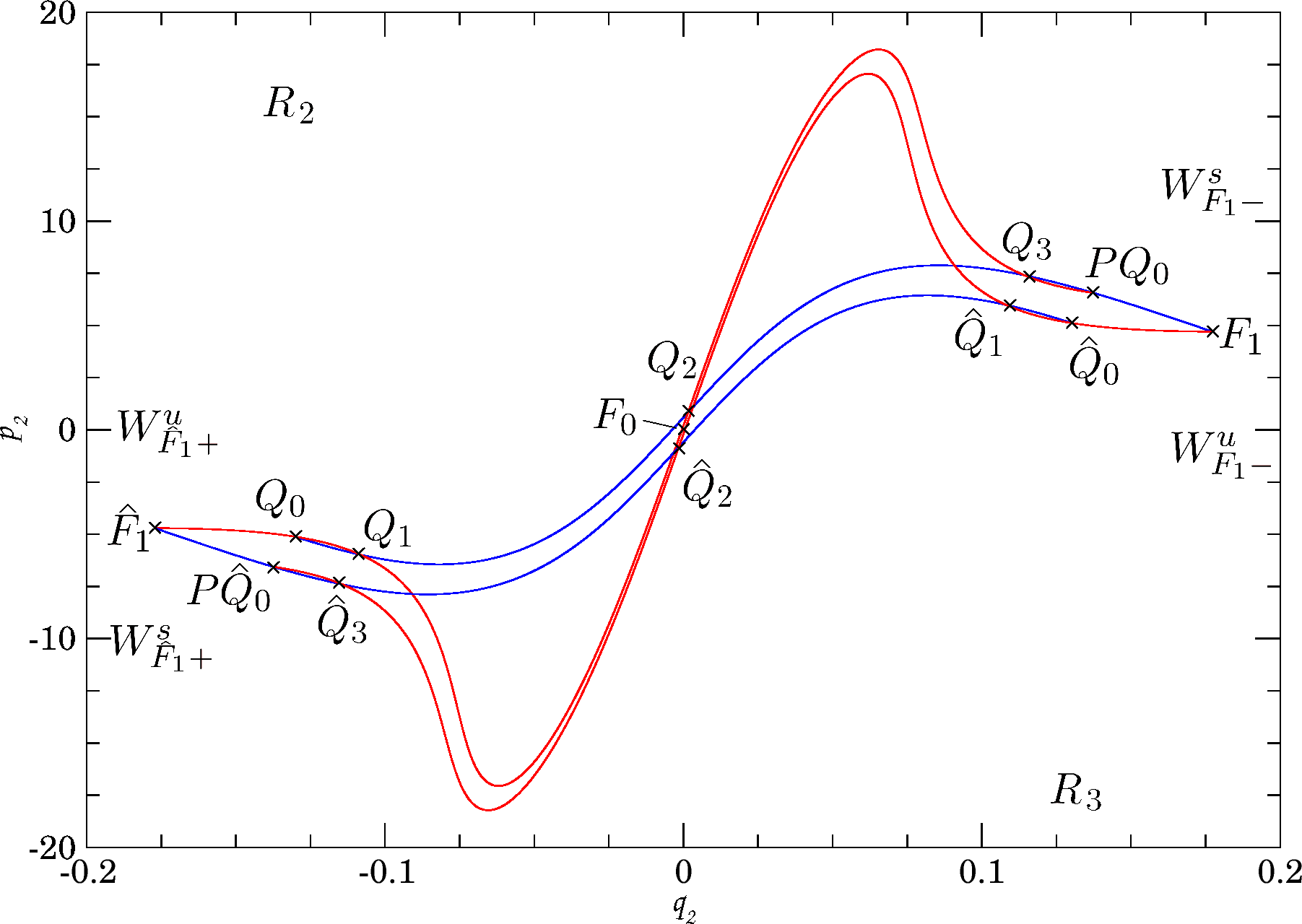}
         \caption{The $F_0$-$F_1$ tangle and the outline of $F_1$-$\widehat{F}_1$ tangle at $0.02230$.}\label{fig:heteroclinic labeled}
 \end{figure}
 
 For the sake of simplicity, we rely on pips and prior knowledge from Sec. \ref{subsec:homoclinic} to interpret Fig. \ref{fig:heteroclinic labeled}. Define $R_0$ in the $F_0$-$F_1$ tangle using $W_{F_0+}$ and $W_{F_1-}$ and the pips $Q_0$ and $Q_2$. A single pip is located on $\partial R_0$ between $Q_0$ and its image, the same is true for $Q_2$.
 
 $L_{3,0}(0)$ bounded by $S[Q_1,Q_0]\cup U[Q_0,Q_1]$ is significantly larger than $L_{0,3}(1)$ bounded by $S[Q_0,P^{-1}Q_1]\cup U[P^{-1}Q_1,Q_0]$. Similarly $L_{0,2}(1)$ is larger than $L_{2,0}(0)$. Also note that $L_{3,0}(0)\cap L_{0,2}(1)$ takes up most of $R_0$. Hence most of $R_0$ originates in $R_3$ and escapes into $R_2$ after $1$ iteration. The trajectories contained therein are nonreactive.
 
 It is worth mentioning that the lobes governing transport from $R_2$ to $R_3$, $L_{0,3}(1)$ and $L_{2,0}(0)$, are disjoint. Nonreactive trajectories originating in $R_2$ spend some time in $R_0$. This agrees with our conclusions on the nonreactive mechanism in the $F_1$ tangle.
  
 The reactive mechanism in the $F_0$-$F_1$ tangle involves the capture lobe $L_{3,0}(1)$ part of which is mapped to $L_{3,0}(0)\setminus L_{0,2}(1)$ and on to $L_{3,0}(1)\cap R_0$, part of which lies in $L_{0,3}(1)$. The area of this intersection is small in $R_0$.
  
 Understanding the $F_1$-$\widehat{F}_1$ tangle is very involved, as the boundary of the tangle requires several segments of $W_{F_1-}$ and $W_{\widehat{F}_1+}$. We propose a different point of view. In all tangles above, we have found that escape from the bounded region in a tangle, all area above the uppermost and below the lowermost stable invariant manifold escapes without further delay. For example in the $F_0$-$F_1$ tangle, $L_{0,2}(1)$ located above $W^s_{F_1-}$ and $L_{0,3}(1)$ located below $W^s_{F_0+}$ escape to reactants and products respectively, because as the stable manifold bounding the lobe contracts, the unstable manifold is unobstructed to leave the interaction region. In this sense that we propose only stable invariant manifolds to be considered a barrier in forward time.
 
 Using this reasoning, concentrate on the area between $S[\widehat{F}_1,\widehat{Q}_0]$ and $S[F_1,Q_0]$ in the $F_1$-$\widehat{F}_1$ tangle. Everything above $S[Q_3,Q_2]$ and below $S[\widehat{Q}_3,\widehat{Q}_2]$ may pass through the tangle, but evolves in a regular and predictable manner from $R_3$ to $R_2$ or vice versa. We remark that this area is the intersection of two turnstiles. The same argument applies to the areas above $S[Q_1,Q_0]$ and below $S[\widehat{Q}_1,\widehat{Q}_0]$. Complicated dynamics is restricted to $R_1$, as defined in the $F_1$ tangle, $\widehat{R}_1$ and an island near $F_0$ and should be treated separately from predictable areas.
 
 Using this line of thought enables us to formulate bounds and estimates of the reaction rate. Before we proceed to quantitative results, we conclude this section by describing the evolution of tangles with increasing energy.
 
 \subsection{Higher energies}\label{subsec:higher energies}
 The based on the analysis in Sections \ref{subsec:homoclinic} and \ref{subsec:heteroclinic} for tangles at $0.02230$, here we discuss on the evolution of tangles at higher energies and their impact on dynamics in the interaction region. As the mechanisms have been described, most of our comments concern sizes of lobes and duration of escape from a tangle.
 
 An interesting question arises from the connection between bifurcations and changes in geometry of invariant structures. The causal relationship is not evident. Also bifurcations are mostly thought of as local events. However as they seem to affect invariant manifolds, a change in tangles propagates instantaneously throughout the whole space. This phenomenon reminds of the infinite propagation speed in the heat equation.
 
 The next bifurcation above $0.02230$ according to Sec. \ref{subsec:po} is a period doubling of $F_{21}$ at $0.02232$, followed by a saddle-centre bifurcation that creates $F_3$ and $F_4$ at $0.02254$ and a bifurcation of $F_3$ where $F_{31}$ and $F_{32}$ are created at $0.02257$. At around $0.02523$ follows another period doubling of $F_{21}$, $F_{21}$ collides with $F_2$ at $0.02651$ and subsequently $F_2$ collides with $F_0$ at $0.02654$.
 
 \begin{figure}
  \centering
     \includegraphics[width=\textwidth]{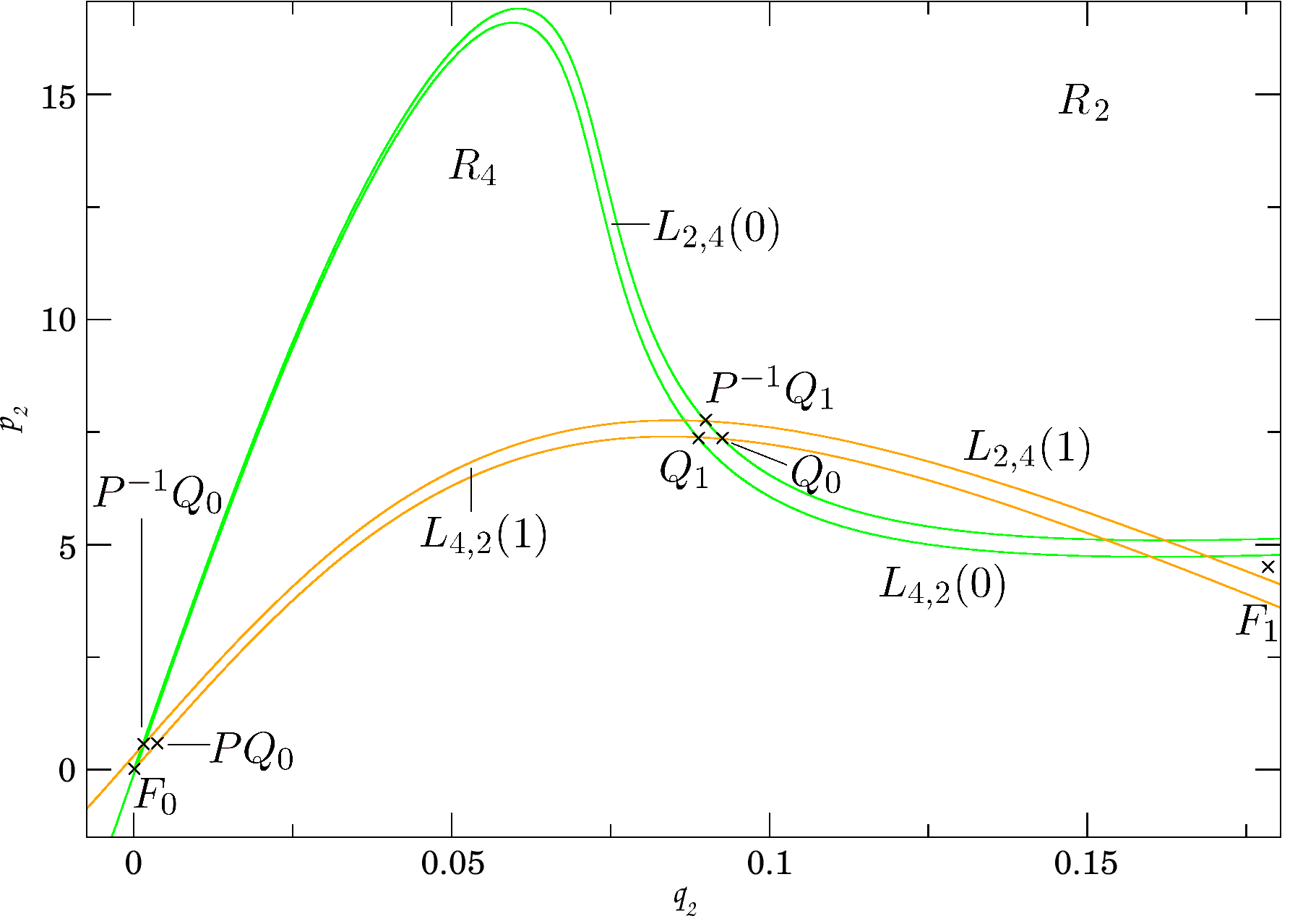}\\
     \includegraphics[width=\textwidth]{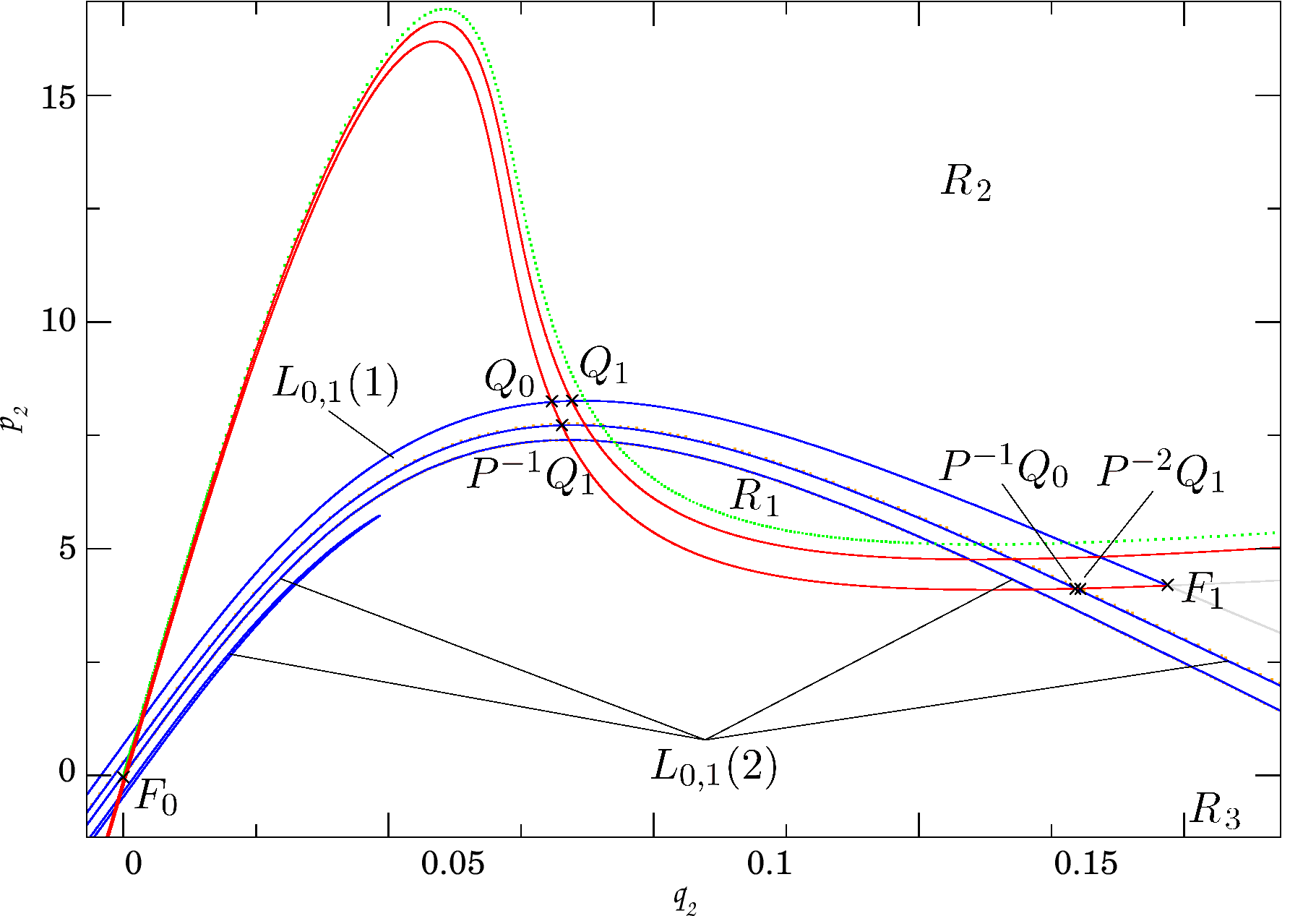}
         \caption{The $F_0$ tangle and the $F_1$ tangle at $0.02253$.}\label{fig:2253}
 \end{figure}
 
 The major consequence of the bifurcation of $F_{21}$ at $0.02232$ is a new intersection of $W^u_{F_0+}$ and $W^s_{F_0+}$ labeled $Q_0$ in Fig. \ref{fig:2253}. This reduces the number of pips between $Q_0$ and $PQ_0$ to one and therefore lobes are no longer made up of two disjoint sets. The $F_0$ tangle resembles the $F_0$-$F_1$ tangle at $0.02230$. Also the size of $R_4$ is reduced.
 
 In the $F_1$ tangle we see $L_{0,1}(2)$ cross DS$_0$ twice as shown in Fig. \ref{fig:2253}. All $L_{0,1}(k)$ for $k>2$ and also $L_{1,0}(k)$ with $k<-2$ therefore pass through $\widehat{R}_1$. Moreover, the tip of $L_{0,1}(2)$ approaching $R_1$ can be expected to pass cross $R_1$ after the bifurcations at $0.02254$ and $0.02257$.
 
 A small remark regarding notation. At this energy $L_{0,1}(2)$ lies in $R_0$, $\widehat{R}_0$, $R_1$, $\widehat{R}_1$, $R_2$ and $R_3$, but we maintain the notation for consistency.
 
 At $0.02400$, $L_{0,1}(2)$ in the $F_1$ tangle passes through $R_1$ twice and the number increases at higher energies. Almost all lobes lie in almost all regions, but the mechanism for fast entry and exit of the tangles remain the same.
 Fig. \ref{fig:2400 regions} shows $R_0$ and $\widehat{R}_0$. While $R_4$ is considerably larger than $R_1$ at $0.02253$, the opposite is true at $0.02400$. Recall that $R_4$ contains predominantly nonreactive trajectories that do not cross DS$_0$, whereas $R_1$ mostly contains ones that do. The overestimation of the reaction rate follows.
 
 The capture lobes in the $F_1$ tangle guide predominantly trajectories from products into $R_1$, as shown in Fig. \ref{fig:2400 regions}. A significant portion of $R_1$ is taken up by $L_{0,1}(0)\cap L_{1,0}(1)$ and it is prevented by $W^s_{F_1-}$ from escaping into reactants. Moreover, the a large part of the intersection lies below $W^s_{\widehat{F}_1+}$, see Fig. \ref{fig:2400 regions}, that guides it into back products as $W^s_{\widehat{F}_1+}$ contracts.
 
 Heteroclinic tangles mirror the changes of the homoclinc tangles (Fig. \ref{fig:2400 tangles}).
 
 \begin{figure}  
    \centering
     \includegraphics[width=\textwidth]{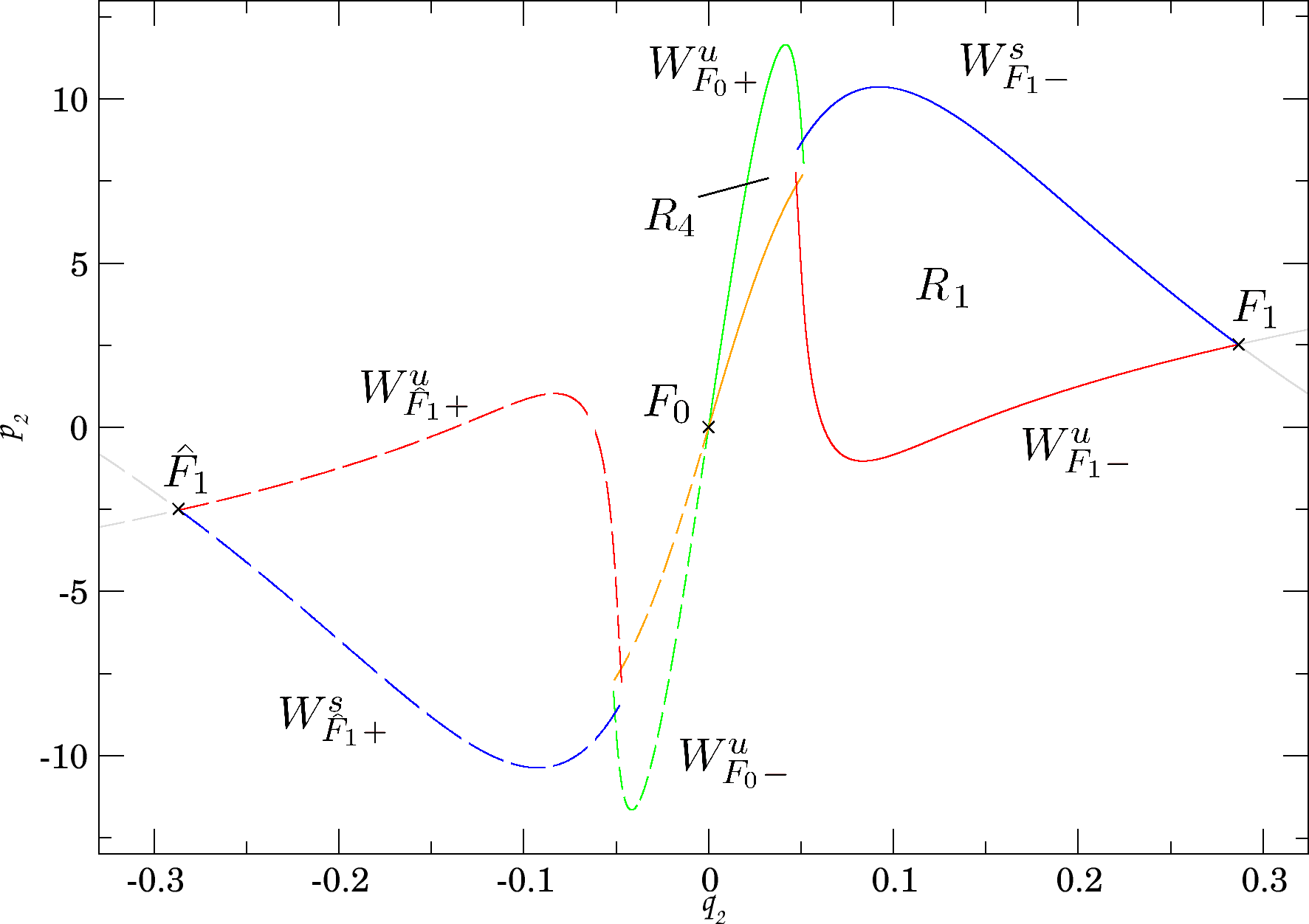}\\
     \includegraphics[width=\textwidth]{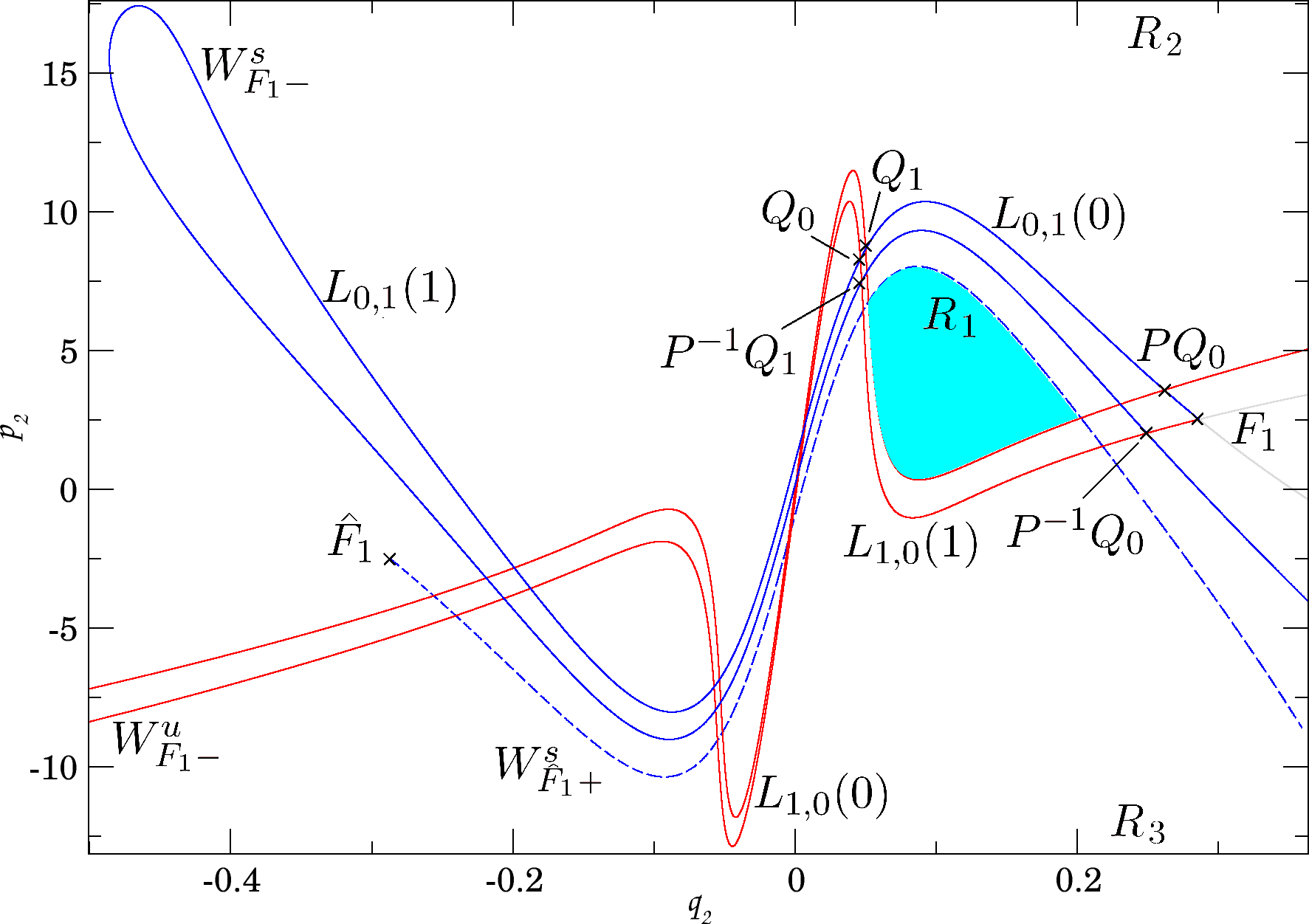}
         \caption{Indication of boundaries of $R_0$ and $R_4$ (above) and the $F_1$ tangle at $0.02400$ (below). The area in the $F_1$ tangle highlighted in cyan is the part of $L_{0,1}(0)\cap L_{1,0}(1)$ that originates in products and is guided by $W^s_{\widehat{F}_1+}$ (dashed) into products.}\label{fig:2400 regions}
  \end{figure}
  
  \begin{figure}
  \centering
     \includegraphics[width=0.49\textwidth]{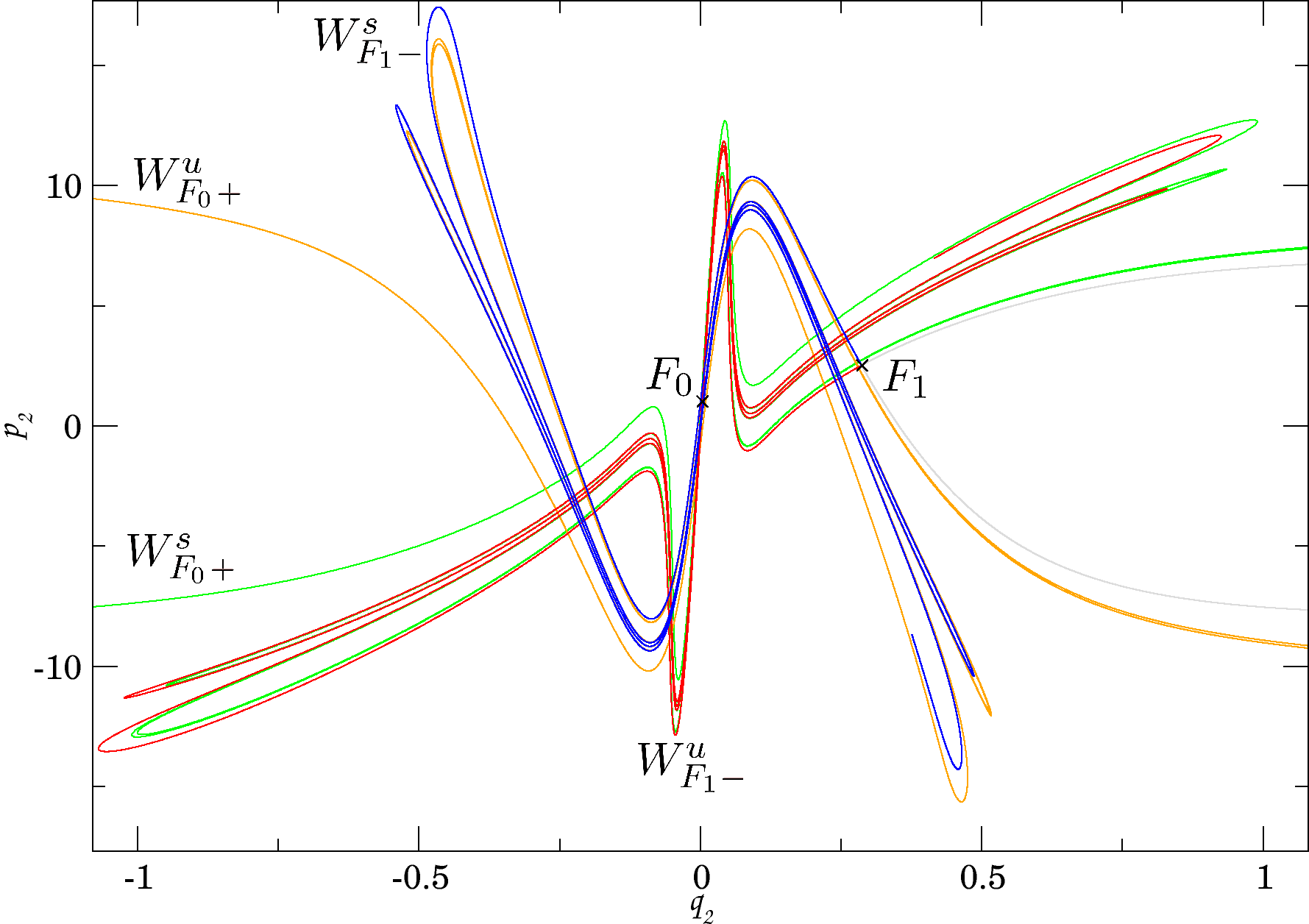}
     \includegraphics[width=0.49\textwidth]{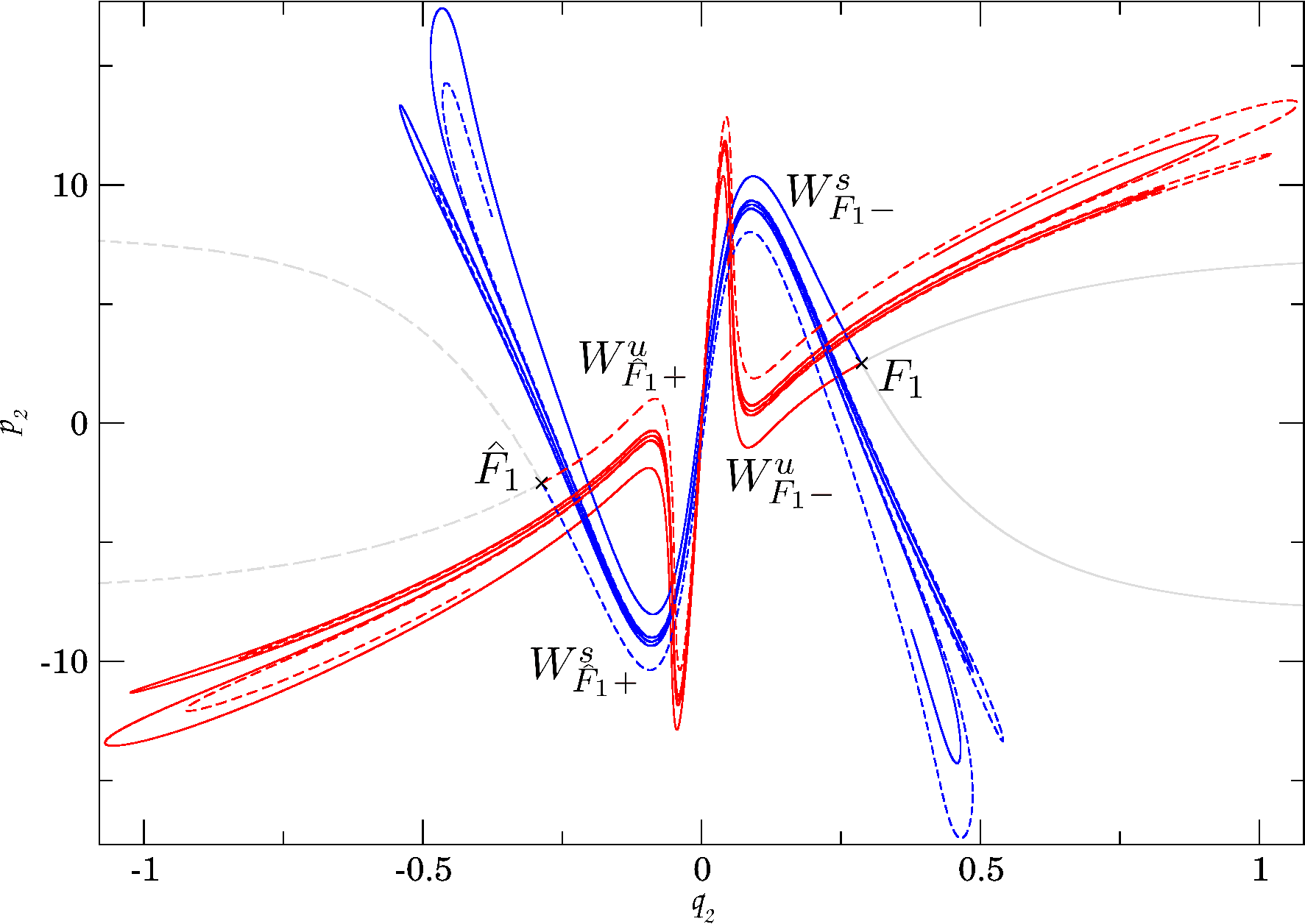}
         \caption{Structure of the heteroclinic tangles at $0.02400$. $W_{F_0+}$ and $W_{F_1-}$ making up the $F_0$-$F_1$ tangle (left) and the $F_1$-$\widehat{F}_1$ tangle (right). Unstable invariant manifolds are as indicated red and green, stable are blue and orange.}\label{fig:2400 tangles}
  \end{figure}
 
 \subsection{Loss of normal hyperbolicity} 
 \label{subsec:lossNH}
 $F_0$ loses normal hyperbolicity and becomes stable at $0.02654$, in a bifurcation involving $F_2$, $\widehat{F}_2$, $F_{21}$, $\widehat{F}_{21}$. TST cannot be based on $F_0$ and $W_{F_0}$ cease to exist. The sudden disappearance of invariant manifolds has no dramatic consequences. As can be deduced from Fig. \ref{fig:2400 tangles}, $W_{F_0}$ are at energies below $0.02654$, very close to $W_{F_1-}$ and $W_{\widehat{F}_1+}$ and naturally take over the role of $W_{F_0}$.
 Throughout the energy interval from $0.02206$ when $F_1$ appears to the loss of normal hyperbolicity at $0.02654$, we see a transition of dominance from $F_0$ to $F_1$-$\widehat{F}_1$. 
 
 The loss of normal hyperbolicity of $F_0$ simplifies dynamics due to the presence of fewer TSs, for example compare Figures \ref{fig:2400 tangles} and \ref{fig:2700}.
 
 At $0.02661$, $F_0$ collides with $F_4$ and becomes inverse hyperbolic. Due to the inverse hyperbolicity, $W_{F_0}$ exist, but they must contain a twist that is manifested as a reflection across the $F_0$ (see \cite{OzoriodeAlmeida90}), i.e. have the geometry of a M\"{o}bius strip. At the same time $W_{F_0}$ are enclosed between $W_{F_1-}$ along with $W_{\widehat{F}_1+}$, but with cylindrical structure. Consequences of the geometry of $W_{F_0}$ are unknown.

 \begin{figure}
  \centering
     \includegraphics[width=\textwidth]{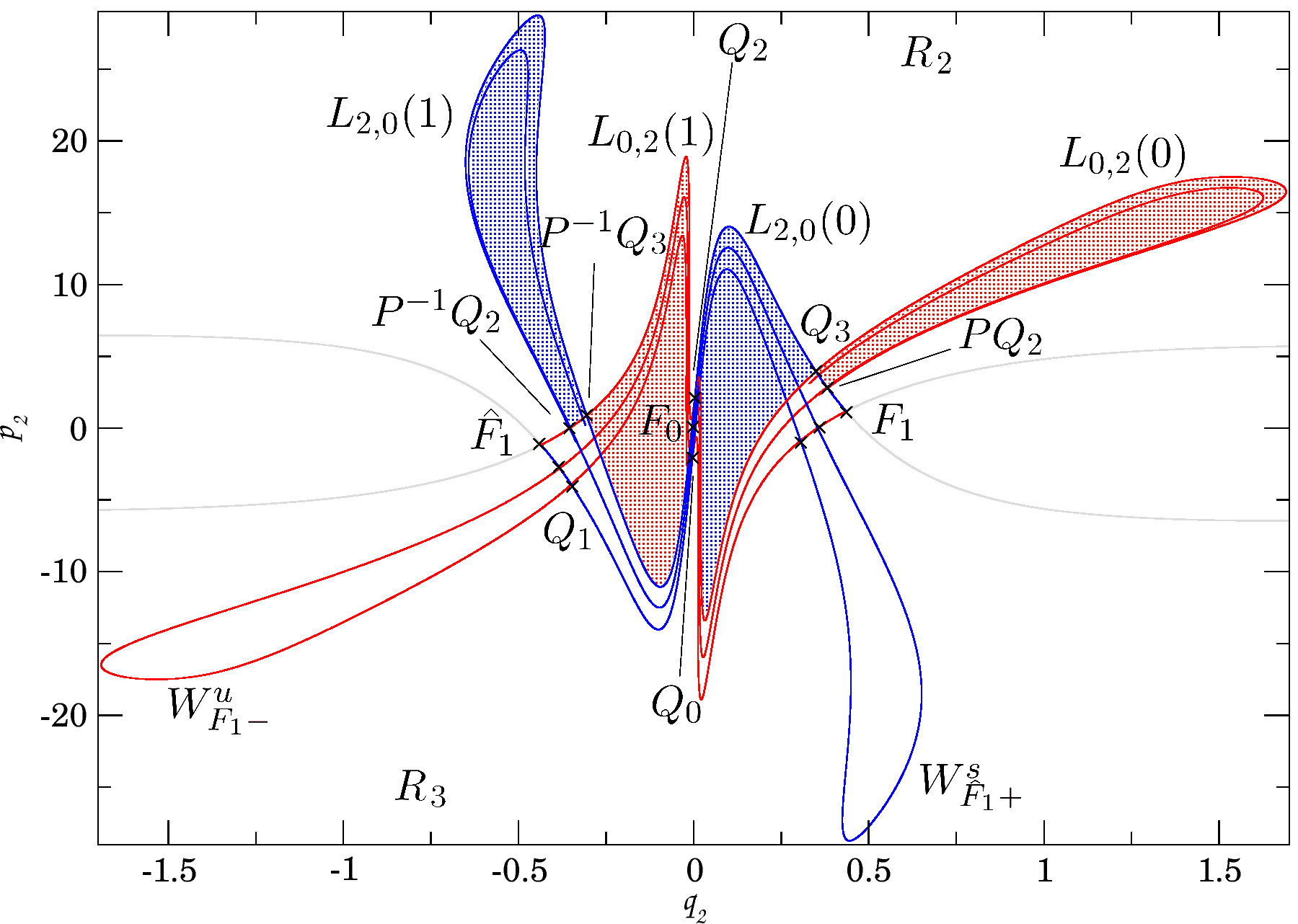}\\
     \includegraphics[width=\textwidth]{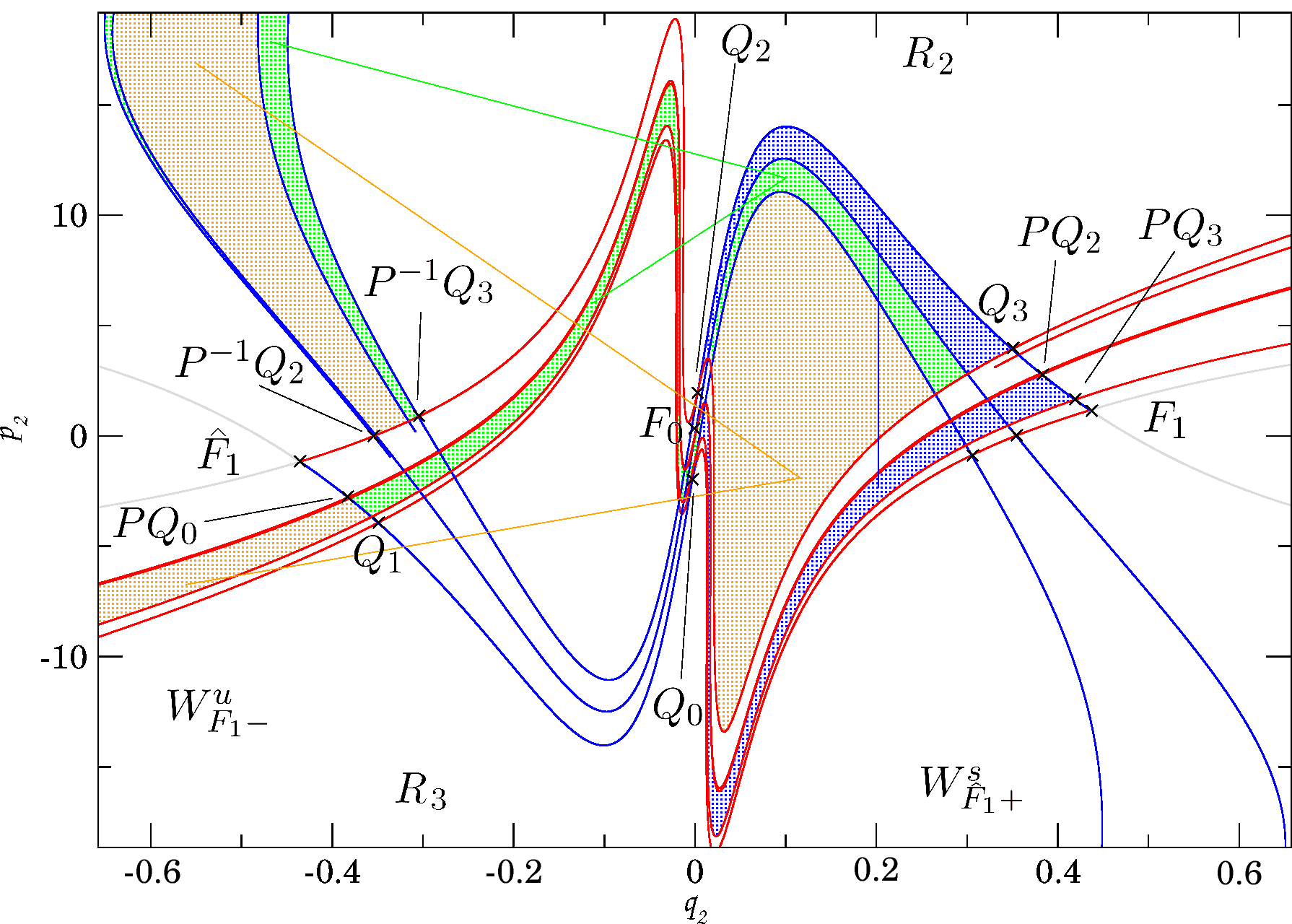} 
         \caption{The $F_1$-$\widehat{F}_1$ tangle at $0.02700$ and an indication how certain parts of lobes are mapped in this tangle.}\label{fig:2700}
 \end{figure}
 
 There are no more significant bifurcations above $0.02661$ and therefore apart from growing tangles and lobes, the tangles remains structurally the same.
 
 Together with $W_{F_0}$ we observe the disappearance of $R_4$ and of the mechanism that carries nonreactive trajectories through the $F_0$ tangle without crossing DS$_0$. Consequently, all trajectories that pass through the $F_1$-$\widehat{F}_1$ tangle cross DS$_0$ at least twice. Each hemisphere of DS$_1$ still possesses the no-return property, which means trajectories cross DS$_1$ at most twice. Trajectories that avoid the tangle cross both DSs once or not at all.
   
 Similarly to lower energies, $R_1$ is predominantly made up of $L_{0,1}(0)\cap L_{1,0}(1)$ in the $F_1$ tangle or $L_{2,0}(0)\cap L_{0,3}(1)$ in the $F_1$-$\widehat{F}_1$ tangle, as shown in Figure \ref{fig:2700}. The argument that trajectories in the $F_1$-$\widehat{F}_1$ tangle below and above all stable manifolds leave the interaction region is still valid. Capture lobes are disjoint, therefore it is not possible to reenter the bounded region. Although $R_1$ and $\widehat{R}_1$ admit return, $R_1\cup\widehat{R}_1$ possesses the no-return property.

 \subsection{Known estimate}\label{subsec:known estimate}
 
 Davis \cite{Davis87} formulated bounds and an estimate of the reaction rate based on numerical observation of dynamics. He observed that a significant portion of trajectories leave the heteroclinic tangle above $0.02654$ after one iteration and imposed the assumption of fast randomization on the remaining trajectories.
 
 As described above, Davis' observation is due a property of the $F_1$-$\widehat{F}_1$ tangle - $R_1$ is mostly occupied by $L_{0,1}(0)\cap L_{1,0}(1)$. We quantify this proportion below.
 
 The assumption of fast randomization of the other trajectories and a $50\%$ probability of them reacting is more difficult to support. From the analysis of lobes we know that however intricate the dynamics is, there is no reason for precisely half of the remaining trajectories to leave to reactants and half to products. Instead we find that for small energies, trajectories that spend $2$ and more iterations $R_0$ and $\widehat{R}_0$ make up a significant part of the tangles (up to half at $0.02350$), but their total proportion is very small and only grows slowly with increasing energy. In the interval up to $0.03000$, these trajectories make up at most $3\%$ of the total, $2\% $ below $0.02650$, see Table \ref{tab:areas}. Consequently, any estimate of the reaction rate that takes trajectories escaping after $1$ iteration into account is accurate to within $3\%$ below $0.03000$ and when we include trajectories escaping after $2$ iterations, this number drops to less than $1\%$.
 
 The difficulty lies in accurately calculating the amount of trajectories. At the cost of accuracy, Davis used VTST as a measure of trajectories entering the interaction region, $\mu(L_{3,0}(0))$ to estimate the size of the tangle and $\mu(L_{3,0}(0)\cap L_{0,2}(1))$ to subtract trajectories escaping after $1$ iteration. The upper and lower estimates assume all, respectively none, of the trajectories that escape after $2$ or more iterations are reactive. 

\section{The intricate energy interval}\label{sec:intricate interval}
 The energy interval $0.02215<E<0.02654$, when TST is not exact and $F_0$ is a TS, has been largely avoided in the past. The interaction of invariant manifolds of two TSs posed enough difficulties. Dividing tangles using pieces of invariant manifolds and following pips to understand dynamics within make this task possible. We divide tangles differently to the lobe dynamics approach, because we aim to describe and measure parts of heteroclinic tangles that do not necessarily fall into a single lobe.

 \subsection{Division of a tangle}\label{subsec:division}
 Davis \cite{Davis87} calculated pieces of invariant manifolds in this interval at an energy of $0.7\text{eV}\approx 0.02572$, but complexity of their intersections did not admit deeper insight. With current understanding it is not possible to consider all the invariant manifolds at once, because even identifying lobes is challenging, not to speak of their intersections.
 
 We use the approach outlined in Sec. \ref{subsec:heteroclinic} and concentrate on $W_{F_1}$ and $W_{\widehat{F}_1}$, while keeping $W_{F_0}$ in mind near $F_0$. A similar approach may be used for homoclinic tangles. We separate predictably evolving trajectories from chaotic ones, for example trajectories escaping after $1$, $2$ or $3$ iterations from the rest of the tangle. To our knowledge, tools for identifying particular lobe intersections and determining the area, a heteroclinic tangle surgery toolbox, have not been previously presented or reported.
 
 There is one more important property of the manifolds that stands out from all previous figures. Inside the $F_1$-$\widehat{F}_1$ tangle, $W^u_{F_1-}$ and $W^u_{\widehat{F}_1+}$ are restricted to the stripe between two pieces of unstable manifold, e.g. $U[\widehat{F}_1,Q_3]$ and $U[F_1,Q_1]$ at $0.02700$ in Fig. \ref{fig:2700} or $U[\widehat{F}_1,PQ_1]$ and $U[F_1,P\widehat{Q}_1]$ at $0.02400$ in Fig. \ref{fig:2400}. Similarly $W^s_{F_1-}$ and $W^s_{\widehat{F}_1+}$ are confined to a single stripe. We remark that $W_{F_0}$ are located between $W_{F_1-}$ and $W_{\widehat{F}_1+}$ and thereby confined as well. It therefore makes sense to study this stripe in detail.
  
 Consider the $F_1$-$\widehat{F}_1$ tangle at $0.02400$, where $R_1$ and $R_4$ are reasonably sized and nonreactive trajectories that do not cross DS$_0$ exist. Following the motto \emph{divide et impera}, we take the following steps:
 \begin{itemize}
  \item We identify new regions that have the no-return property.
  \item We use as few pieces of invariant manifolds as possible.
  \item We define subsets of regions containing reactive/nonreactive trajectories.
 \end{itemize}
 
 \begin{figure}
  \centering
      \includegraphics[width=0.49\textwidth]{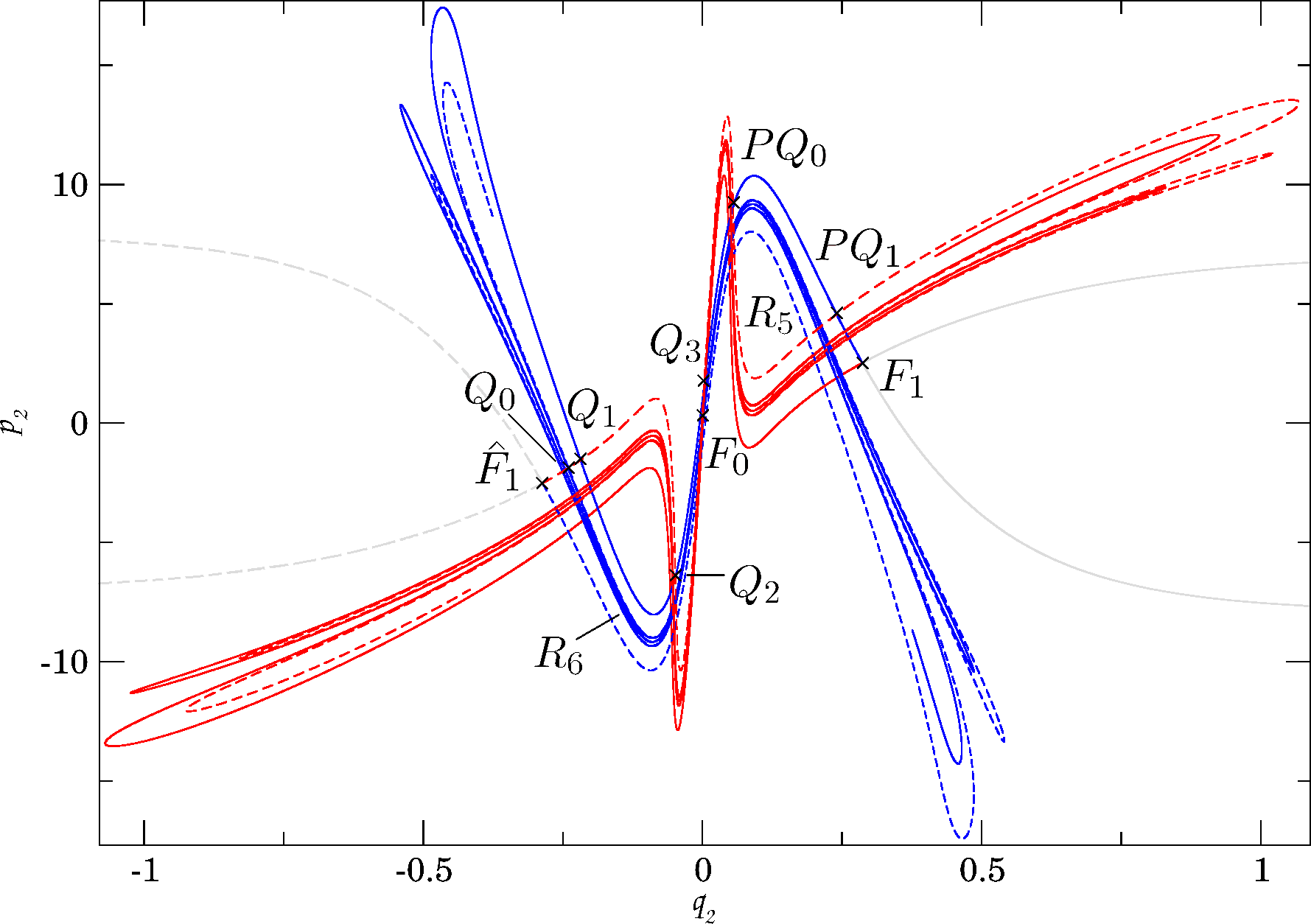}
      \includegraphics[width=0.49\textwidth]{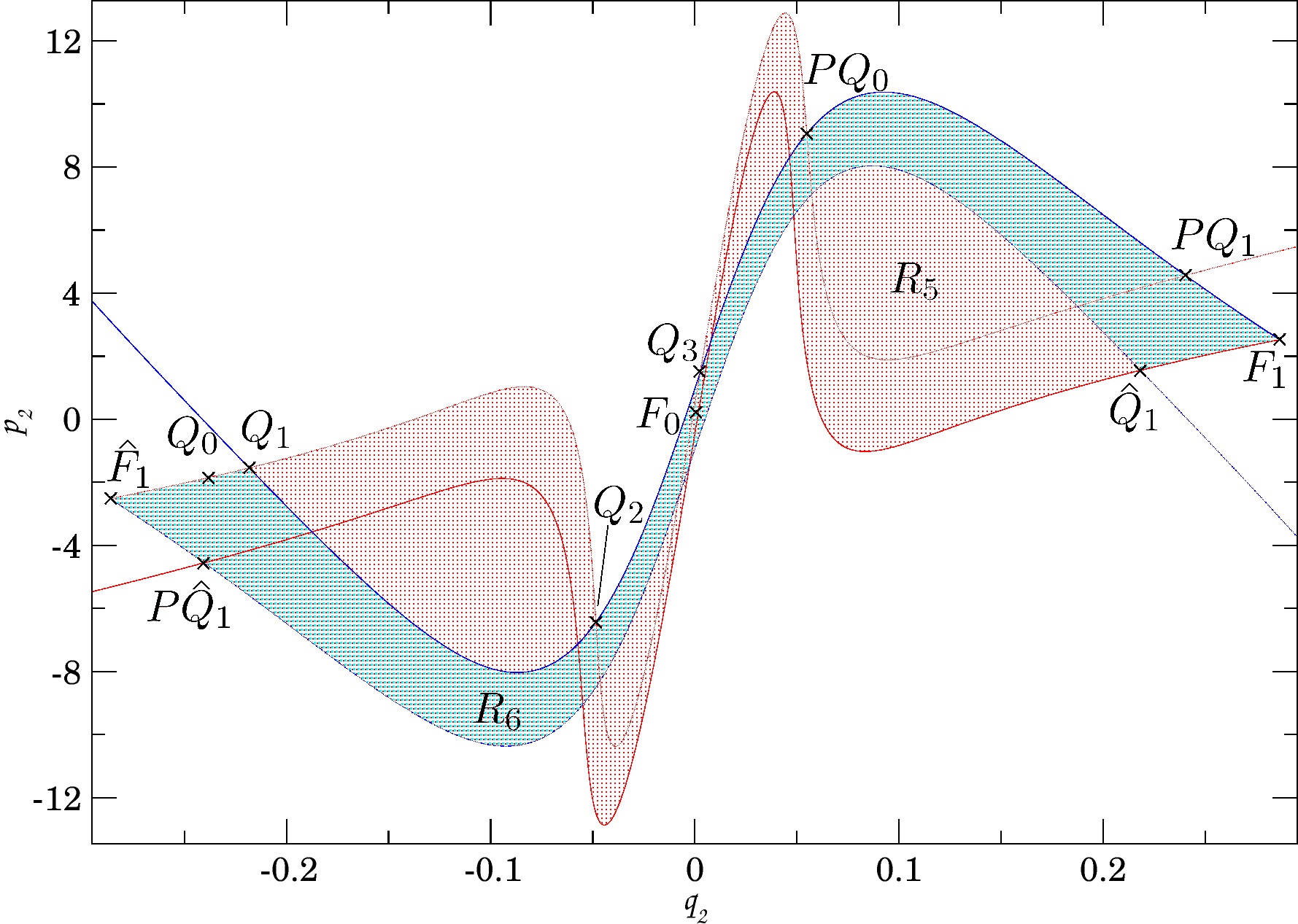}
	  \caption{The $F_1$-$\widehat{F}_1$ tangle at $0.02400$ (left) and its simplification (right). $W_{F_1-}$ are drawn with solid lines, $W_{\widehat{F}_1+}$ are dashed.}\label{fig:2400}
 \end{figure}

 Define $R_5$ as the bounded region inside the tangle, the upper part of the boundary is made up of $U[\widehat{F}_1, Q_2]$, $S[Q_2, Q_3]$, $U[Q_3,PQ_0] $ and $S[PQ_0,F_1]$, see Figure \ref{fig:2400}, and the lower part is symmetric to it. Each lobe consists of two disjoint sets, for example, $L_{2,5}(0)$ is bounded by $S[PQ_1,PQ_0]$, $U[PQ_0,PQ_1]$ and $S[Q_3,Q_2]$, $U[Q_2,Q_3]$. We remark that lobes do not intersect outside $R_5$ and leave the interaction region. Disjoint capture lobes imply:

 \begin{remark}
  $R_5$ has the no-return property.
 \end{remark}
 As found in Sec. \ref{subsec:higher energies}, a large part of $R_5$ behaves regularly and leaves the tangle within $1$ iteration. As argued in Sec. \ref{subsec:heteroclinic}, stable manifolds contract in forward time and thereby act as a barrier. Everything above $W^s_{F_1-}$ leaves at the next iteration to reactants, everything below $W^s_{\widehat{F}_1+}$ leaves to products. This agrees with the lobes $L_{5,3}(1)$ and $L_{5,2}(1)$ that leave $R_5$ by definition.
 
 The remainder of $R_5$ is the stripe between $W^s_{F_1-}$ and $W^s_{\widehat{F}_1+}$, the only part of $R_5$ where stable manifolds can lie. We refer to it as the \emph{capture stripe} and denote it $R_6$, see Fig. \ref{fig:2400}. Its boundary consists of $S[\widehat{F}_1,\widehat{Q}_1]$, $U[F_1,\widehat{Q}_1]$, $S[F_1,Q_1]$ and $U[\widehat{F}_1,Q_1]$.
 
 In backward time, the roles of stable and unstable manifolds switch - everything below $W^u_{F_1-}$ and above $W^u_{\widehat{F}_1+}$ escapes $R_5$. Define $R_7$, the \emph{escape stripe} bounded by $S[\widehat{F}_1,P\widehat{Q}_1]$, $U[F_1,P\widehat{Q}_1]$, $S[F_1,PQ_1]$ and $U[\widehat{F}_1,PQ_1]$. $R_5\setminus R_7$ escapes $R_5$ after $1$ iteration in backward time.
 
 We conclude that all complicated and chaotic dynamics is confined to $R_6\cap R_7$ and due to the no-return property of $R_5$:
 
 \begin{remark}
  $R_6$ and $R_7$ have the no-return property.
 \end{remark}

 Note that the boundary of $R_7$ is the image of the boundary of $R_6$. Necessarily
 $$PR_6=R_7,$$
 and due to preservation of area $\mu(R_6)=\mu(R_7)$.
 
 There are more regions with the no-return property in the $F_1$-$\widehat{F}_1$ tangle. Obviously, $R_5\setminus (R_6\cup R_7)$ must be a no-return region as it escapes the $R_5$ immediately after entering. Also $R_6\setminus R_5$ as the entry point to $R_7$ must have the no-return property as well as capture and escape lobes.

 \subsection{Dynamical properties}
 \label{subsec:dynamical properties}
 To shorten and facilitate the description of reactive and dynamical properties of $R_5$, $R_6$ and $R_7$, we introduce the following classification of trajectories.
 
 \begin{definition}
  We call the set of trajectories:\\
  \emph{directly reactive} ($DR$) if they remain in $R_2$ or $R_3$,\\
  \emph{directly nonreactive} ($DN$) if they do not enter the interaction region,\\
  \emph{captured reactive} after $n$ iterations ($CR_n$) if they react after $n$ iterations in $R_5$,\\
  \emph{captured nonreactive} after $n$ iterations ($CN_n$) if they return to the region of origin after $n$ iterations in $R_5$.
 \end{definition}

 Clearly $DR$ and $DN$ never enter $R_5$. Following Sec. \ref{subsec:higher energies} and Sec. \ref{subsec:division}, $R_5\setminus (R_6\cup R_7)$ is the region of $CN_1$ and $CR_1$ is always empty. $CR_2$ and $CN_2$ are pass through $R_6\setminus R_7$ and $R_7\setminus R_6$ and therefore never enter $R_6\cap R_7$.
 
 This leaves the complicated evolution and chaotic behaviour restricted to $R_6\cap R_7$. Below $0.02500$, $R_6\cap R_7$ consists of $5$ squares near $F_0$, $F_1$, $\widehat{F}_1$, $F_2$ and $\widehat{F}_2$. As $F_2$ and $\widehat{F}_2$ approach the bifurcation with $F_0$, the three squares near them merge into one around $0.02523$ when $F_{21}$ bifurcates, see Figures \ref{fig:2500} and \ref{fig:2550}.
 
 \begin{figure}
  \centering
     \includegraphics[width=0.49\textwidth]{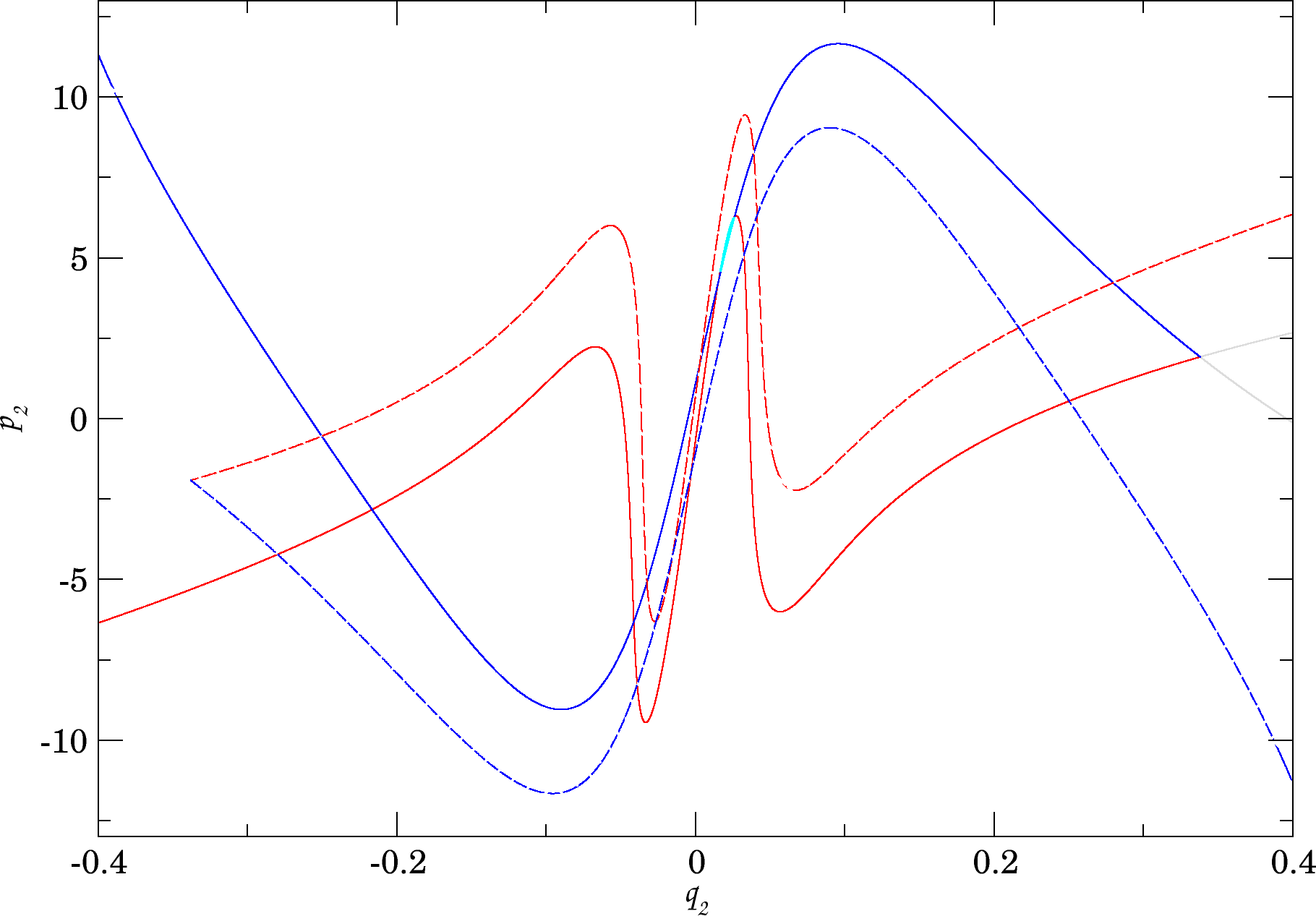}
     \includegraphics[width=0.49\textwidth]{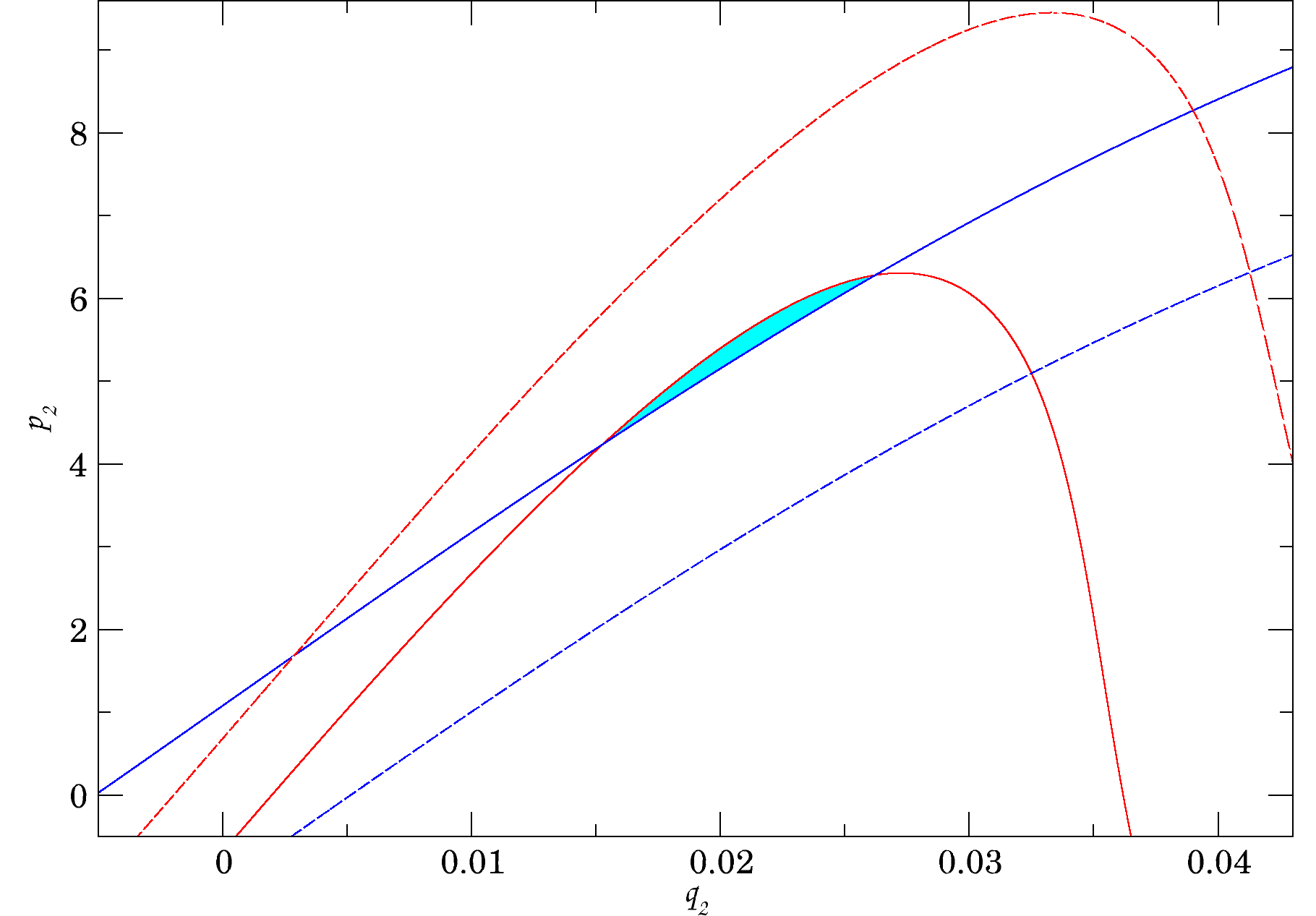}
         \caption{The $F_1-\widehat{F}_1$ tangle at $0.02500$ and a detail of the diminishing part of $R_5\setminus(R_6\cup R_7)$ highlighted in cyan.}\label{fig:2500}
 \end{figure}
 \begin{figure}
  \centering
     \includegraphics[width=0.49\textwidth]{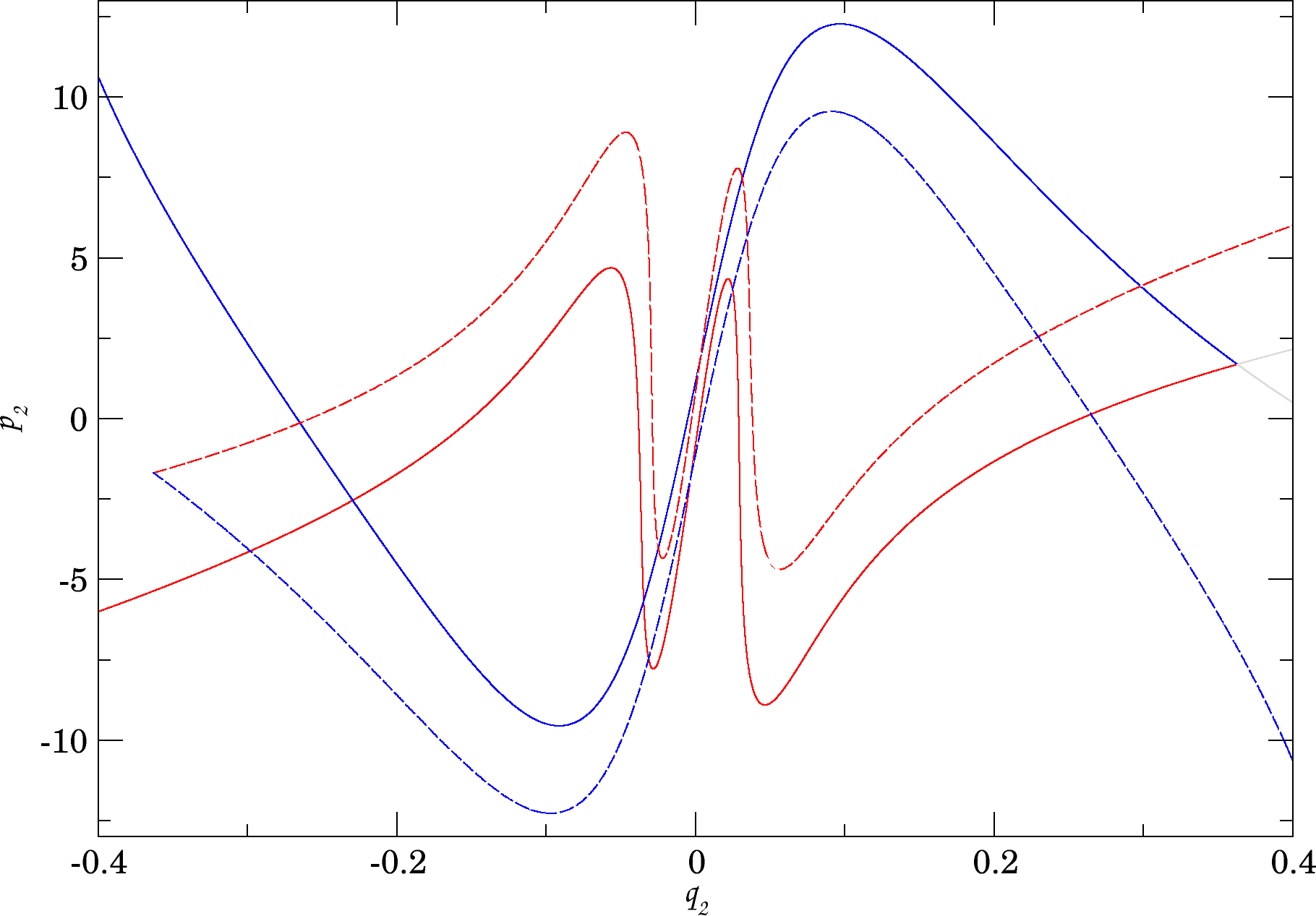}
     \includegraphics[width=0.49\textwidth]{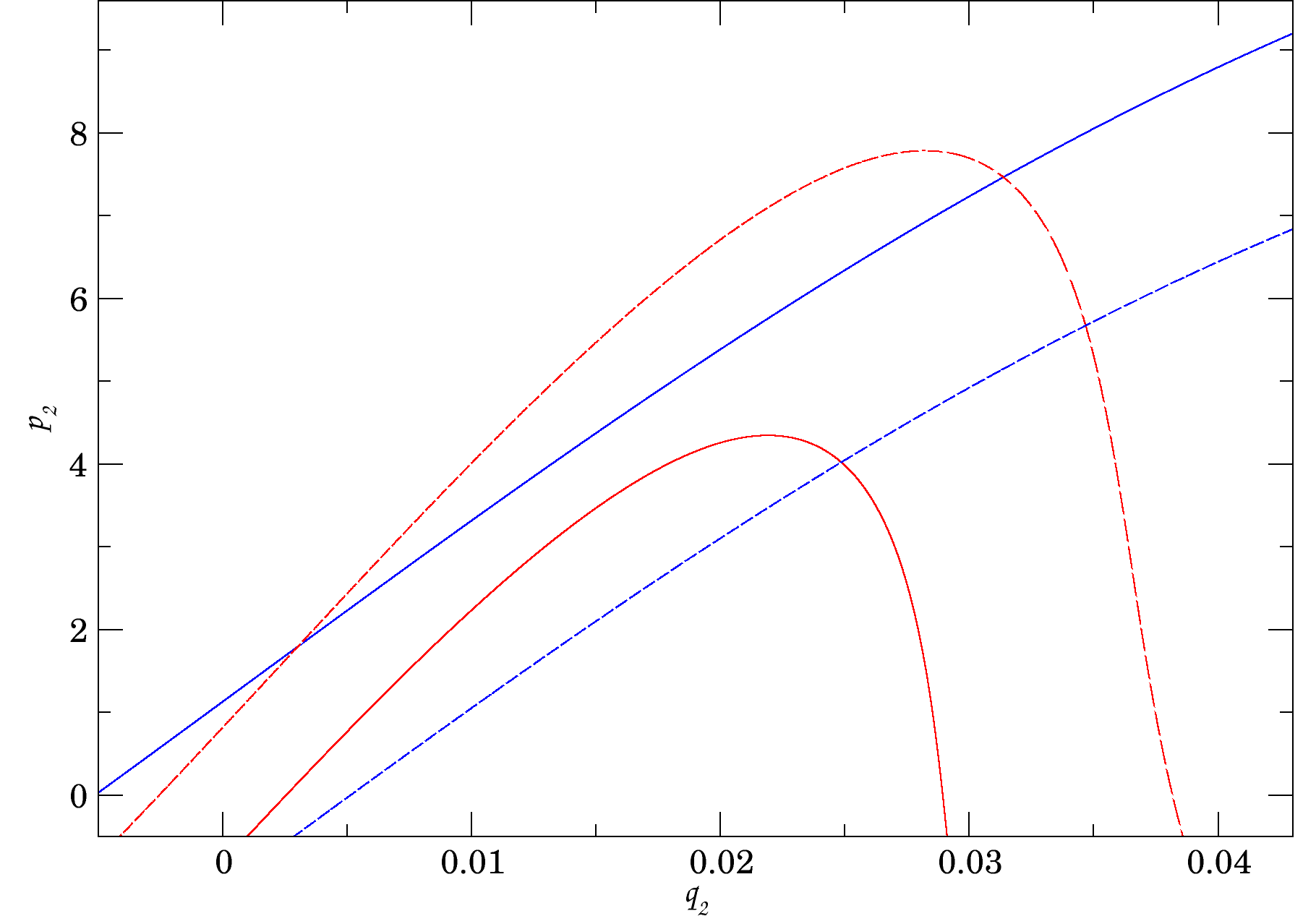}
         \caption{The $F_1-\widehat{F}_1$ tangle at $0.02550$ and a detail of the diminished part of $R_5\setminus(R_6\cup R_7)$.}\label{fig:2550}
 \end{figure}
 
 Trajectories enter $R_5$ via $R_6\setminus R_7$ and escape via $R_7\setminus R_6$, hence every trajectory crosses $R_6\setminus R_7$ and $R_7\setminus R_6$ at most once. The same is true for $R_5\setminus (R_6\cup R_7)$ consisting of $CN_1$. Therefore of $R_5$ only the size $R_6\cap R_7$ does not reflect the number of trajectories it contains. It follows that the area of $R_6\setminus R_7$, $R_7\setminus R_6$ and $R_5\setminus (R_6\cup R_7)$ on the surface of section $\Sigma_0$ is the same of their images on DS$_1$ and DS$_{\widehat{1}}$.
 
 \begin{figure}
  \centering
     \includegraphics[width=\textwidth]{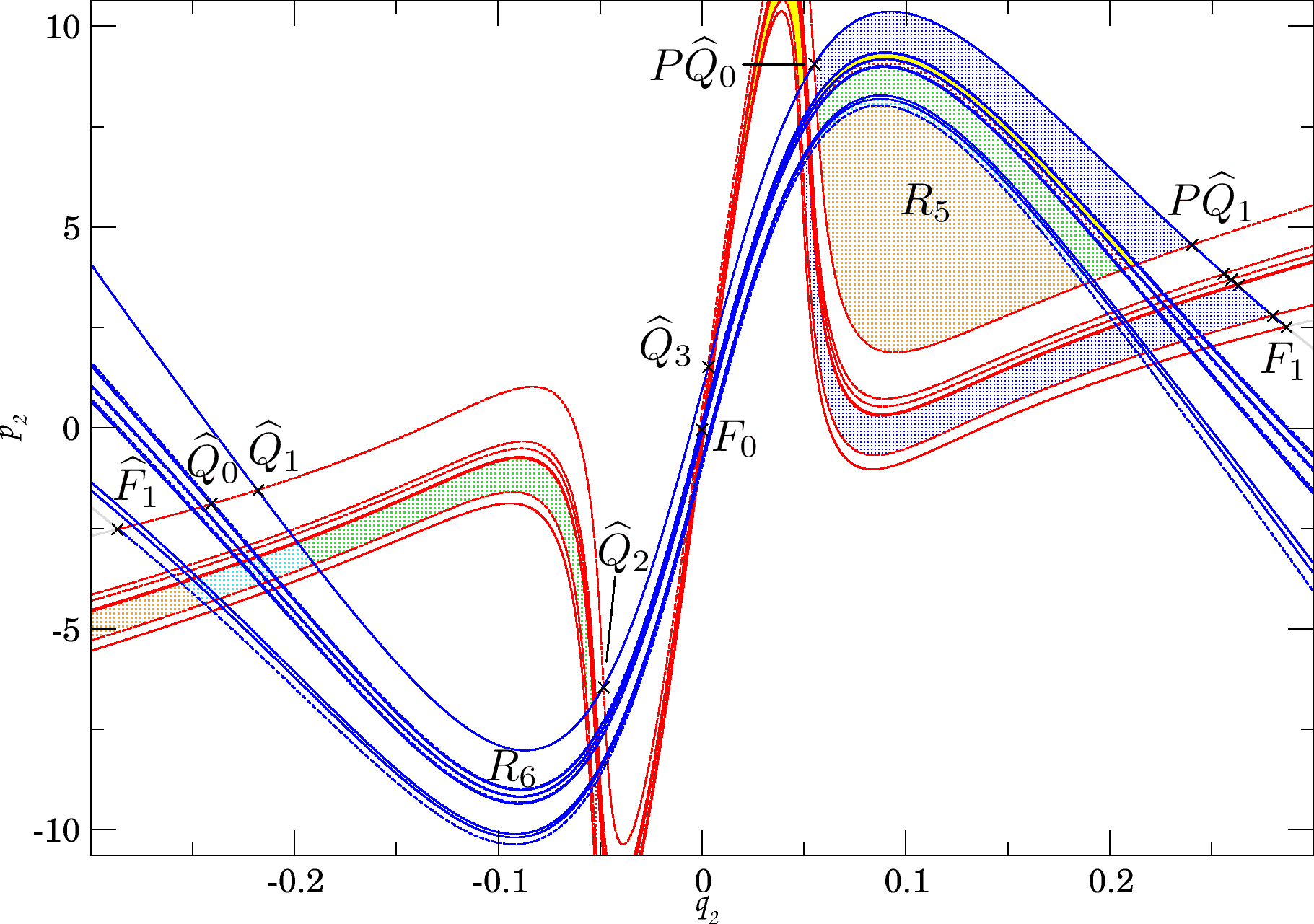}\\
     \includegraphics[width=\textwidth]{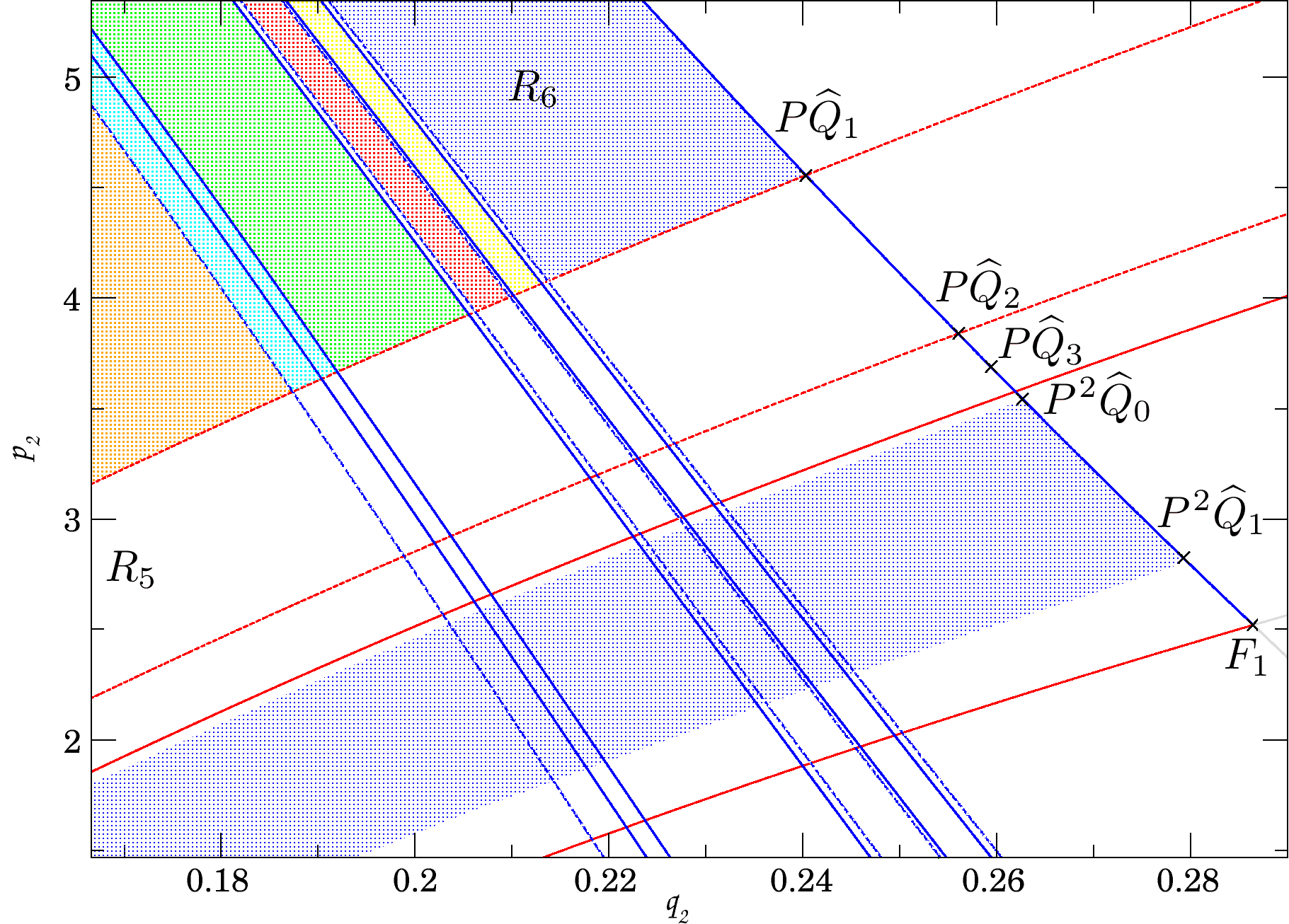}
         \caption{Coloured sets in $R_6$ showing how part of the capture stripe is mapped at $0.02400$. $W_{F_1-}$ are drawn with solid lines, $W_{\widehat{F}_1+}$ are dashed. $CN_1$ are shown in orange, $CR_2$ green and yellow, $CN_2$ red. Part of blue also belongs to $CN_2$. Blue, yellow, red and green are separated by white stripes that are mapped to $R_6\cap R_7$.}\label{fig:2400 dyn}
 \end{figure}
  
 Fig. \ref{fig:2400 dyn} shows a more detailed partitioning of $R_6$ and $R_7$. Essentially, $R_6$ is divided into finer stripes by pieces of $W^s_{F_1-}$ and $W^s_{\widehat{F}_1+}$ that are nearly parallel to the boundary. The boundary of $R_6$ illustrates how the content of the stripe is deformed when mapped into $R_7$. It is compressed along the stable manifolds towards the fixed points, e.g.
 $$P(S[F_1,Q_1])=S[F_1,PQ_1],$$
 and stretched along the unstable manifolds away from the fixed points. We remark that the whole highlighted set in $R_5\setminus R_7$ of Fig. \ref{fig:2400 dyn} is connected, only separated by stable invariant manifolds. When mapped forward it is stretched, but remains connected. The sets labeled by yellow, red and green are alternately mapped above and below the capture stripe.
 
 There is a connection between these coloured stripes and lobe intersections, but lobes do not distinguish how often trajectories cross DS$_0$, which is necessary to understand overestimation of the reaction rate by TST.
 
 The connected components of $R_6\cap R_7$ contain dynamics similar to Smale's horseshoe dynamics, \cite{Hirsch04}. As a consequence we observe a fractal structure, as can be seen in Fig. \ref{fig:2500 isl}. $R_6\cap R_7$ accounts for less than $12\%$ of the all trajectories that pass through $R_5$ below $0.02400$ when the dynamics is relatively slow. The proportion drops to roughly $7\%$ of $R_5$ at $0.03000$ and remains below $1\%$ of the total amount of trajectories.
 
  \begin{figure}
  \centering
     \includegraphics[width=0.7\textwidth]{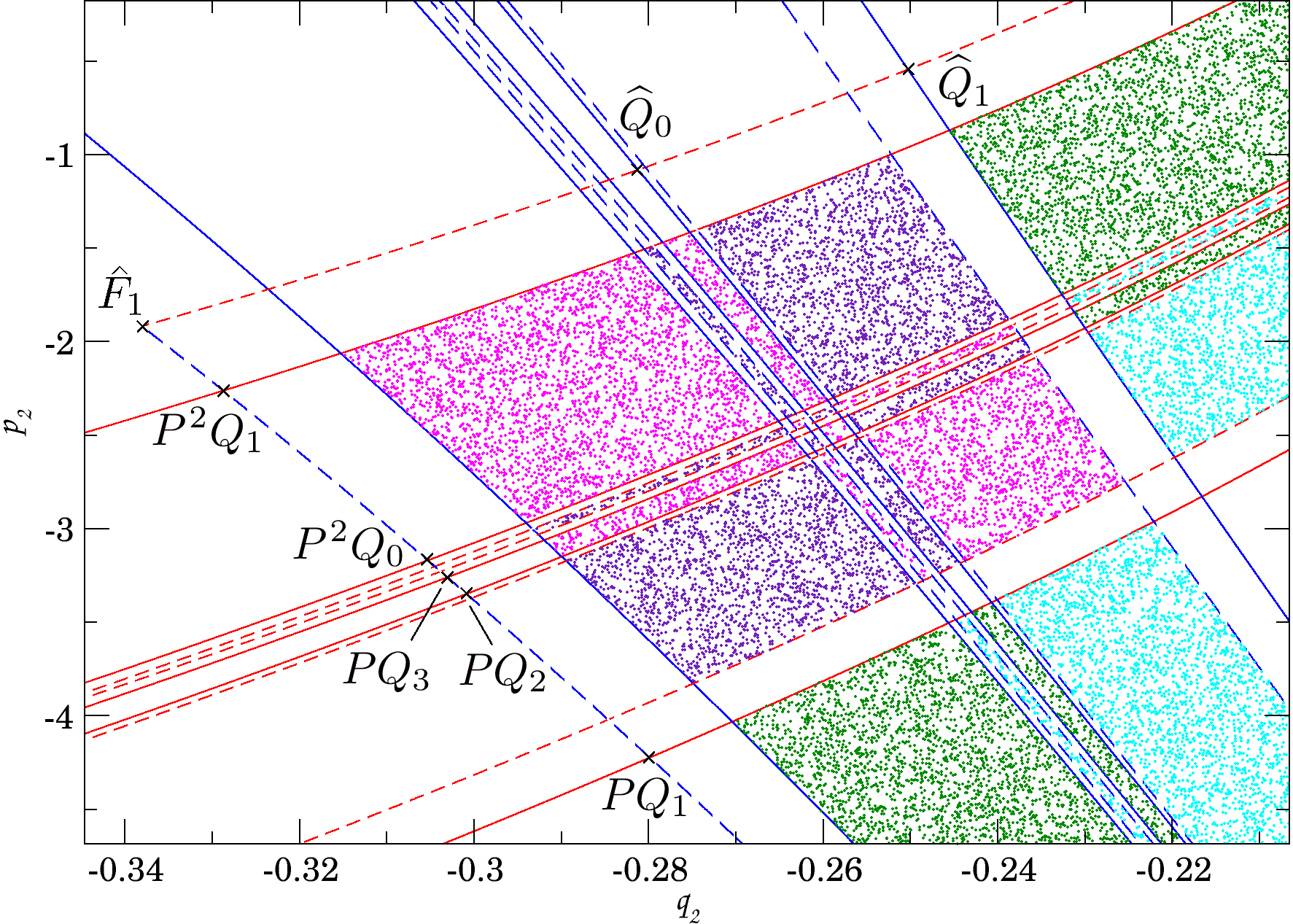}
         \caption{Detail of the island near $\widehat{F}_1$ and its content at $0.02500$. $CN_3$ are purple, $CR_3$ are magenta and the rest of the island is plain. $CN_2$ (green) and $CR_2$ (cyan) contained in the adjacent regions $R_6$ and $R_7$ are shown for completeness.}\label{fig:2500 isl}
 \end{figure}
 
 \subsection{Areas}
 In this system, determining the area of $R_5$, $R_6$, $R_7$, $R_5\setminus (R_6\cup R_7)$, $R_6\setminus R_7$ and $R_6\cap R_7$, is significantly easier than calculating lobe intersections. We employ a Monte Carlo based method that is expensive, yet simple. Ultimately the cost and accuracy depend on the level of detail in $R_6\cap R_7$, i.e. it can be determined a priori. We also tune initial and terminal conditions to obtain a high accuracy at a reasonable cost.
 
 Previous works seem to consider initial conditions on $r_1+\frac{r_2}{2}=50$, $p_{r_1}<0$, which is a surface near $q_2=1181$. We prefer to sample the hemisphere of DS$_1$, through which trajectories enter the interaction region. Directly we have that the difference between the inward hemisphere of DS$_1$ and $r_1+\frac{r_2}{2}=50$, $p_{r_1}<0$ corresponds to DN trajectories.
 
 The slowest of $DR$ are located near the boundary of $R_5$, and those near pips evolve similarly to pips. Using pips on $W^u_{F_1}$ we define \emph{checkpoints}, that mark distance these pips are mapped, i.e. the least distance $DR$ cover in the interaction region in $1$ iteration. Then all trajectories that pass the second checkpoint $1$ iteration after they pass the first checkpoint, are $DR$ and are not captured in $R_5$. Recall that all captured reactive trajectories spend at least $2$ iterations in $R_5$. Trajectories that have a delay of $n$ iteration between crossing of the checkpoints are $CR_n$. Since $R_5$ is symmetric, we use the symmetric counterpart of the second checkpoint to identify $CN_n$.
 
 \begin{figure}
  \centering
     \includegraphics[width=\textwidth]{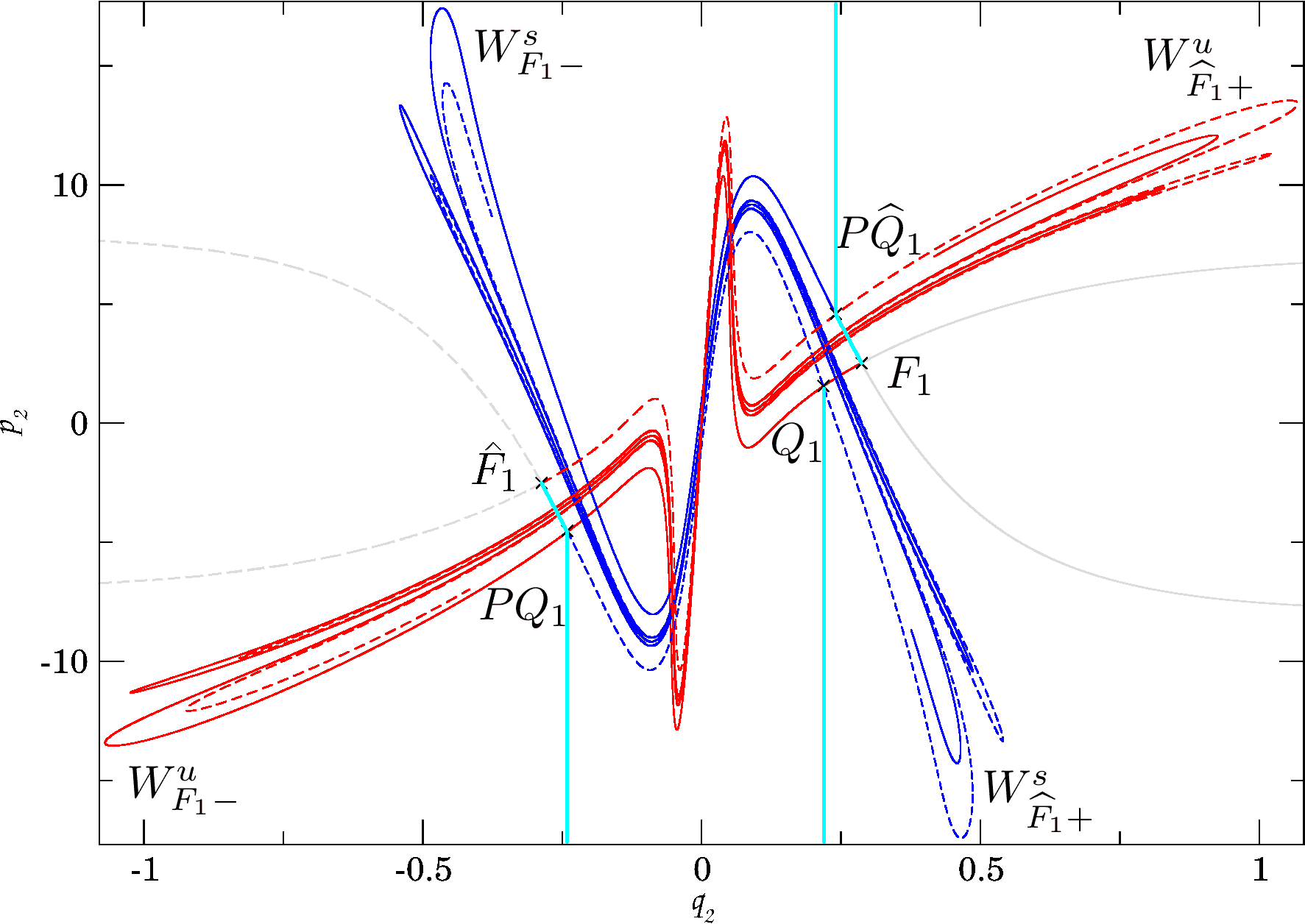}
         \caption{Checkpoints defined in the $F_1$-$\widehat{F}_1$ tangle at $0.02400$.}\label{fig:2400 chk}
 \end{figure}
 
 At $0.02400$, $Q_1$ and $PQ_1$ are the natural choice for checkpoints, because mark the endpoints of $R_6\setminus R_7$ via which trajectories enter $R_5$, see Figure \ref{fig:2400 chk}. We define checkpoint $Ch_{Q_1}$ as a vertical line passing through $Q_1$.
 It is necessary that $Ch_{Q_1}$ avoids capture lobes, therefore at energies above $0.02900$ the computationally most efficient solution is to use another vertical line between $Q_1$ and $F_0$.
 
 The role of the second checkpoint, $Ch_{PQ_1}$, is to distinguish trajectories in $R_5$ from those outside $R_5$. We use a linear approximation of $S[\widehat{F}_1,PQ_1]$, the boundary between the escape lobe and $R_5$, in conjunction with a vertical line passing through $PQ_1$. The checkpoint symmetric to $Ch_{PQ_1}$ is defined analogously and denoted $Ch_{P\widehat{Q}_1}$. If desired, we can track the number crossings of DS$_0$ using the sign of $q_2$.
 
 We can measure individual components of $R_5$: $\mu(R_5\setminus R_7)$ corresponds to the number of captured trajectories, $\mu(R_5\setminus(R_6\cup R_7))$ is given by $CN_1$. Then
 $$\mu(R_6\setminus R_7)=\mu(R_5\setminus R_7)-\mu(R_5\setminus(R_6\cup R_7)),$$
 and $R_6\cap R_7$ can be deduced from $CR_n$ and $CN_n$ where $n\geq3$. The latter follows from the fact that $CR_2$ and $CN_2$ do not pass through $R_6\cap R_7$.
 
 This method is not computationally cheap, but the computational difficulty can be easily estimated a priori. Determining the distribution up to $CR_n$ and $CN_n$ with $N$ initial conditions requires approximately $nN$ iterations of the map $P$, but considering the prevalence of $DR$ and $DN$, this number will be considerably lower.
 
 Alternative approaches to calculating lobe areas face the obstacle in distinguishing the inside from the outside of a lobe, not to mention their intersections. Recent developments \cite{krajnak2018phase,krajnak2018influence} suggest that a reactive island approach can be used to calculate areas of intersections in a cheaper and simpler manner.

 \section{Bounds of the reaction rate}\label{sec:bounds}
 \subsection{Quantification}\label{subsec:quantification}
 
 In Tab. \ref{tab:areas} we present proportions of areas of classes of trajectories on the plane $r_1+\frac{r_2}{2}=50$, $p_{r_1}<0$. Between $10^7$ and $2.10^8$ initial conditions were used to obtain these values.
 
 \begin{table}
 \centering
 \begin{tabular}{r|*{6}{c}}
    Energy & $DR$ & $DN$ & $CR_2$ & $CN_1$ & $CN_2$ & $Other$ \\
    \hline
    0.02205 & 0.590 & 0.410 & 0 & 0 & 0 & 0\\
    0.02214 & 0.595 & 0.405 & 0 & 0 & 0 & 0\\
    \hline
    0.02215 & 0.687 & 0.296 & 0 & 0.016 & 0.001 & 0.000\\ 
    0.02230 & 0.693 & 0.290 & 0.001 & 0.013 & 0.001 & 0.001\\ 
    0.02253 & 0.703 & 0.282 & 0.001 & 0.011 & 0.001 & 0.002\\ 
    0.02300 & 0.717 & 0.266 & 0.003 & 0.008 & 0.003 & 0.002\\ 
    0.02350 & 0.725 & 0.255 & 0.004 & 0.010 & 0.005 & 0.003\\ 
    0.02400 & 0.733 & 0.239 & 0.004 & 0.015 & 0.006 & 0.003\\ 
    0.02450 & 0.737 & 0.227 & 0.005 & 0.021 & 0.006 & 0.004\\ 
    0.02500 & 0.739 & 0.216 & 0.006 & 0.028 & 0.007 & 0.005\\ 
    0.02550 & 0.739 & 0.206 & 0.006 & 0.037 & 0.008 & 0.006\\ 
    0.02600 & 0.737 & 0.197 & 0.007 & 0.046 & 0.008 & 0.006\\ 
    0.02650 & 0.734 & 0.188 & 0.007 & 0.055 & 0.009 & 0.007\\ 
    \hline
    0.02662 & 0.734 & 0.186 & 0.008 & 0.057 & 0.009 & 0.007\\ 
    0.02700 & 0.731 & 0.180 & 0.008 & 0.064 & 0.010 & 0.007\\ 
    0.02800 & 0.723 & 0.166 & 0.009 & 0.083 & 0.011 & 0.008\\ 
    0.02900 & 0.714 & 0.153 & 0.010 & 0.101 & 0.012 & 0.009\\ 
    0.03000 & 0.705 & 0.145 & 0.011 & 0.116 & 0.013 & 0.010
  \end{tabular}
 
  \caption{Proportions of areas of classes of trajectories on the plane $r_1+\frac{r_2}{2}=50$, $p_{r_1}<0$. Directly reactive (DR) and directly nonreactive ($DN$) trajectories do not enter $R_5$. Captured reactive ($CR_2$) and captured nonreactive ($CN_1$, $CN_2$) enter and leave $R_5$ after $1$ or $2$ iterations. $Other$ trajectories do not leave $R_5$ within $2$ iterations after their entry and are inside $R_6\cap R_7$. Horizontal lines represent the creation of the homoclinic tangles and loss of normal hyperbolicity of $F_0$.}
 \label{tab:areas}
 \end{table}
 
 The proportion of $DN$ decreases steadily over the whole interval presented in Tab. \ref{tab:areas} and beyond. This is not surprising given that widening bottlenecks allow more trajectories enter the interaction region. Thereby nonreactive heavily oscillating trajectories enter the interaction region and consequently $R_5\setminus(R_6\cup R_7)$ grows faster than the rest of $R_5$.
 
 Note that the proportion of $DR$ culminates between $0.02500$ and $0.02550$. At this energies the geometry of the $F_1$-$\widehat{F}_1$ tangle simplifies with the consequence that all captured trajectories cross DS$_1$.
 The proportion of $DR$ above $0.02550$ decreases predominantly in favour of $CN_1$. We observe the growth of capture lobes mainly in the area of large $|p_2|$ momentum, containing predominantly $CN_1$ trajectories (Tab. \ref{tab:areas}), approaching the maximal values of $|p_2|$ at the given energy. This implies that for a given (small) $|p_1|$ momentum, trajectories are more likely to react at a lower energy due to smaller capture lobes.
 
 In the physical world large values of $|p_2|$ correspond by definition (Sec. \ref{subsec:momenta}) to large $|p_{r_1}|$ on the reactant side and large $|p_{r_2}|$ on the product side. Since trajectories are less likely to react at higher energies, the mechanism for transfer of kinetic energy between the degrees of freedom in the interaction region must be failing at high energies. Consequently, the energy passed from the incoming $H$ to the $H_2$ may be so high, that it repels the whole molecule instead of breaking its bond. This may be true for a whole class of collinear atom-diatom reactions, provided it is possible to define an interaction region multiple TSs.

 \begin{table}[ht]
  \centering
  \begin{tabular}{r|*{6}{c}}
    Energy & $P_{\text{TST}}$ & $P_{\text{VTST}}$ & $P_{\text{MC}}$ & $L_2$ & $U_2$ & $U_1$ \\
    \hline
    0.01600 & 0.181 & 0.181 & 0.181 &  &  & \\
    0.01800 & 0.383 & 0.383 & 0.383 &  &  & \\
    0.01900 & 0.469 & 0.469 & 0.469 &  &  & \\
    0.02000 & 0.545 & 0.545 & 0.545 &  &  & \\
    0.02100 & 0.615 & 0.615 & 0.615 &  &  & \\
    0.02205 & 0.681 & 0.681 & 0.681 &  &  & \\
    \hline
    0.02215 & 0.687 & 0.687 & 0.687 & 0.687 & 0.687 & 0.689\\ 
    0.02230 & 0.696 & 0.696 & 0.695 & 0.694 & 0.696 & 0.697\\
    0.02253 & 0.709 & 0.709 & 0.705 & 0.704 & 0.706 & 0.707\\
    0.02300 & 0.736 & 0.734 & 0.721 & 0.720 & 0.722 & 0.725\\
    0.02350 & 0.763 & 0.748 & 0.732 & 0.728 & 0.731 & 0.736\\
    0.02400 & 0.789 & 0.761 & 0.739 & 0.737 & 0.741 & 0.746\\
    0.02450 & 0.814 & 0.773 & 0.744 & 0.742 & 0.746 & 0.753\\
    0.02500 & 0.838 & 0.784 & 0.746 & 0.744 & 0.749 & 0.756\\
    0.02550 & 0.860 & 0.794 & 0.747 & 0.744 & 0.750 & 0.758\\
    0.02600 & 0.883 & 0.804 & 0.747 & 0.743 & 0.750 & 0.758\\
    0.02650 & 0.904 & 0.812 & 0.745 & 0.742 & 0.748 & 0.757\\
    \hline
    0.02662 & 0.909 & 0.814 & 0.744 & 0.741 & 0.748 & 0.757\\
    0.02700 & 0.924 & 0.820 & 0.743 & 0.739 & 0.746 & 0.756\\
    0.02800 & 0.963 & 0.835 & 0.736 & 0.732 & 0.740 & 0.751\\
    0.02900 & 0.999 & 0.847 & 0.729 & 0.725 & 0.734 & 0.746\\
    0.03000 & 1.033 & 0.858 & 0.720 & 0.716 & 0.726 & 0.739\\
    \hline
    0.04000 & 1.278 & 0.960 & 0.626 &  &  & \\
    0.05000 & 1.428 & 1.002 & 0.542 &  &  &
  \end{tabular}
  \caption{Comparison of results of TST and VTST with the actual reaction rate computed via Monte Carlo and our upper and lower estimates.}
  \label{tab:myP}
 \end{table}
 
 \subsection{MC based bounds}  
 Using Tab. \ref{tab:areas} we are able to formulate estimates of the reaction rate up to arbitrary precision. The idea is similar to \cite{Davis87}. An upper/lower bound on the reaction rate is obtained by assuming that all/none of the trajectories that remain in $R_5$ after $n$ iterations react. Tab. \ref{tab:myP} contains the resulting bounds.
 
 Denote $U_1$ the rate estimate obtained by assuming all trajectories in $R_5\setminus(R_6\cup R_7)$ react, or equivalently only $CN_1$ do not react. Since $CN_2$ and some of $Other$ do not react, the true reaction rate is lower.

 \begin{lemma}
  $U_1=\mu(DR)+\mu(CR_2)+\mu(CN_2)+\mu(Other)$ is an upper bound of the reaction rate.
 \end{lemma}
 
 $U_1$ can be easily improved by acknowledging that $CN_2$ are nonreactive. Denote this bound by $U_2$. Because the $R_6\cap R_7$ contain reactive as well as nonreactive trajectories,  the true reaction rate is lower.
 \begin{lemma}
  $U_2=\mu(DR)+\mu(CR_2)+\mu(Other)$ is an upper bound of the reaction rate.
 \end{lemma}

 A lower bound $L_2$ is obtained assuming all of $R_6\cap R_7$ are nonreactive trajectories.
 \begin{lemma}
  $L_2=\mu(DR)+\mu(CR2)$ is a lower bound of the reaction rate.
 \end{lemma}
 
 The difference between $L_2$ and $U_2$ is precisely $\mu(Other)$. This gives us an upper bound on the error of both estimates. An estimate of the reaction rate can be obtained using $L_2$ and $U_2$.
 
 \section{Conclusion}
 We have studied invariant manifolds of TSs to find an explanation for the decrease of the reaction rate. In the process of understanding how energy surface volume passes through homoclinic and heteroclinic tangles formed by these invariant manifolds we found the need for tools that would allow us to work with the tangles and not get lost in details of its chaotic structure.
 We introduced a suitable division of homoclinic and heteroclinic tangles that is simple and understandable based on reactive properties of trajectories.
 
 Once divided, the heteroclinic tangles decompose into areas of simple and more complicated dynamics. We were able to identify a large class of trajectories that are merely diverted by the tangles and areas of fractal horseshoe-like structure near hyperbolic or inverse-hyperbolic periodic orbits. In addition to a better understanding, the division provides an easy way calculating the corresponding areas. 
 
 Contrary to expectations, the decline of the reaction rate is not a result of loss of normal hyperbolicity. We may consider the decrease of the reaction rate and loss of normal hyperbolicity to be consequences of insufficient transfer of kinetic energy between the degrees of freedom. In physical terms, the single atom has so much kinetic energy, that it repels the whole molecule instead of becoming part of it.

%

\end{document}